\documentclass[12pt,a4paper]{article}

\usepackage{amsthm,amsmath,multirow,graphicx,subfigure,booktabs,url,mathtools,wrapfig,mathrsfs,bm,relsize,caption,color, enumitem, fullpage,float,bbm,xcolor}

\RequirePackage[authoryear]{natbib}
\RequirePackage[psamsfonts]{amssymb}
\usepackage[colorlinks=true,
linkcolor=blue,
urlcolor=blue,
citecolor=blue]{hyperref}

\usepackage[capitalise]{cleveref}

\numberwithin{equation}{section}
\newtheorem{conj}{Conjecture}
\newtheorem{theorem}[conj]{Theorem}
\newtheorem{cor}[conj]{Corollary}
\newtheorem{prop}[conj]{Proposition}
\newtheorem{lemma}[conj]{Lemma}
\newtheorem{ass}{Assumption}

\newtheorem{cond}{Condition}
\newtheorem{definition}{Definition}

\newtheorem{innercustomgeneric}{\customgenericname}
\providecommand{\customgenericname}{}
\newcommand{\newcustomtheorem}[2]{%
	\newenvironment{#1}[1] 
	{%
		\renewcommand\customgenericname{#2}%
		\renewcommand\theinnercustomgeneric{##1}%
		\innercustomgeneric
	}
	{\endinnercustomgeneric}
}

\newcustomtheorem{customAss}{Assumption}

\theoremstyle{remark}\newtheorem{remark}{Remark}

\def\Cov{{\rm Cov}}   
\def\PP{\mathbb{P}}
\def\EE{\mathbb{E}}
\def\RR{\mathbb{R}}

\def\wh{\widehat}

\def\bI{\mathbf{I}}

\def\T{\top}
\def\i{\infty}  
 
\def\sw{\Sigma_w}

\def\sxy{\Sigma_{XY}}
\def\dist{\bM^\T\Sigma_w^{-1}\bM}
\def\Var{{\rm Var}}
\def\supp{{\rm supp}} 

\def\rank{{\rm rank}}
\def\tr{{\rm tr}}
\def\op{{\rm op}}
\def\diag{{\rm diag}}

\def\rI{\textrm{I}}
\def\rII{\textrm{II}}

\def\Dt{\Delta}

\def\bM{{ M}}
\def\bX{{\mathbf X}}
\def\bY{{\mathbf Y}}

\def\be{{\bm e}}
\def\bD{{\bm D}}

\def\cE{\mathcal{E}}

\def\cO{\mathcal{O}}
\def\cS{\mathcal{S}}
\def\cN{\mathcal{N}} 
\def\cC{\mathcal{C}}
\def\cR{\mathcal{R}}
\def\cQ{\mathcal{Q}}
\def\1{\mathbbm{1}} 
\def\pen{\text{pen}}

\DeclareMathOperator*{\argmax}{arg\,max}
\DeclareMathOperator*{\argmin}{arg\,min}

\title{A New Regression Lens on Multi-Class Classification}
\author{
    Xin Bing\thanks{Department of 
  Statistical Sciences, University of Toronto. E-mail: \texttt{xin.bing@utoronto.ca}}
  ~~~~~ Bingqing Li\thanks{Department of 
  Statistical Sciences, University of Toronto. E-mail: \texttt{bbingqing.li@mail.utoronto.ca}}
    ~~~~~Marten Wegkamp\thanks{Department of Mathematics and  Department   of Statistics and Data Science, Cornell University.  E-mail: \texttt{marten.wegkamp@cornell.edu}.}
}

\begin{document}

\date{\today}
\maketitle

\begin{abstract}
Linear Discriminant Analysis (LDA) is a fundamental method for classification. Its simple linear structure facilitates interpretation, and it is naturally suited to multi-class settings. LDA is also closely connected to several classical multivariate techniques, including Fisher’s discriminant analysis, canonical correlation analysis, and linear regression.

In this paper, we strengthen the connection between LDA and multivariate response regression by establishing an explicit relationship between discriminant directions and regression coefficients. This characterization yields a new regression-based framework for multi-class classification that accommodates structured, regularized, and even non-parametric regression methods. In contrast to existing regression-based approaches, our formulation is particularly amenable to theoretical analysis: we develop a general strategy for deriving bounds on the excess misclassification risk of the proposed classifier across all such regression procedures.

As concrete applications, we provide complete theoretical guarantees for two widely used methods---$\ell_1$-regularization and reduced-rank regression---neither of which has previously been fully analyzed in the LDA context. The theoretical results are supported by extensive simulation studies and empirical evaluations on real data.

\end{abstract}


\section{Introduction}
 
Linear Discriminant Analysis (LDA) is a popular tool for predicting a categorical response using a set of explanatory variables. Its popularity stems from several favorable properties, such as ease of interpretation, reasonable robustness to departures from normality, and the capability to handle responses with multiple classes. See, \cite{hastie1995penalized} and many references in \cite{seber2009multivariate}.  This paper introduces a new regression-based approach to multi-class, high-dimensional linear discriminant analysis (LDA).

Assume that $(X_1,Y_1),\ldots,(X_n,Y_n)$ are i.i.d. copies of a random pair $(X, Y)$, where the label $Y$ takes values in $\{ \be_1,\ldots,\be_L\}$, the canonical basis vectors in $\RR^L$, with $L\ge 2$ classes, and the feature vector $X\in \RR^p$ has conditional means collected in the $p$ by $L$ matrix
 $   \bM  = (\mu_1, \ldots, \mu_L)$ 
and the  {\em within-class} covariance matrix 
$$
    \sw =\Cov(X\mid Y = \be_\ell),\quad \text{for all }\ell \in [L]:=  \{1,\ldots, L\}.
$$
We assume $\sw$ to be strictly positive definite and the class  probabilities $\pi_\ell = \PP\{Y = \be_\ell\}$, for all $\ell \in [L]$, are strictly positive. 
 We denote by $\Sigma$ the {\em unconditional} covariance matrix of $X$.
 Since we can always subtract the marginal mean of $X$, we assume  
 $\EE(X) = 0_p$.  

   For a new feature $x\in \RR^p$, LDA predicts its corresponding label as
    \begin{equation}\label{bayes_rule}
       \argmin_{\ell \in [L]} ~   (x-\mu_\ell)^\T \sw^{-1}(x- \mu_\ell)  -  2\log(\pi_\ell).             
    \end{equation} 
     The LDA rule in \eqref{bayes_rule} coincides with the Bayes rule when $X\mid Y$ is Gaussian.
     See, for instance, \cite{Izenman-book}.  
     Estimation of the classifier \eqref{bayes_rule} based on $(X_1,Y_1),\ldots,(X_n,Y_n)$  becomes challenging in (i) the high-dimensional setting $p>n$ and (ii) the multi-class setting $L>2$ that allows for $L\to\i$ as $n\to \i$. 

\subsection{Existing approaches}


Few approaches consider (ii) but notable exceptions are \citet{levy2023generalization,AP19,nibb22}.

We recall various approaches to deal with (i). 
  In high-dimensional settings $p>n$, 
 \citet{cai2011direct} assume that the direction  $\beta^*:=\sw^{-1} (\mu_2 - \mu_1)$ is sparse in the binary case $L=2$, and propose a Dantzig-type procedure to estimate this direction by approximately solving the equation $ \sw \beta^* = \mu_2-\mu_1$ for $\beta^*$. Other procedures based on the same or similar equations include, for instance, \cite{qiao2009sparse, Shao2011,Fan2012} for $L=2$, and \cite{Tibshirani2002,FanFan2008,caizhang2019,chen2022distributed,mai2019multiclass,gaynanova2016simultaneous} for $L>2$. Specifically, for multi-class responses, there are $L-1$ directions $\beta_\ell^* = \sw^{-1}(\mu_\ell - \mu_1)$ for $2\le \ell \le L$, which can be solved via 
 \begin{equation}\label{crit_EE}
        (\beta_2^*,\ldots,\beta_L^*) = \argmin_{\beta_2,\ldots,\beta_L}  ~ \sum_{\ell=2}^L \left({1\over 2} \beta_\ell^\T \sw \beta_\ell - \beta_\ell^\T (\mu_\ell - \mu_1)\right).
 \end{equation}
  For large $p>n$, sparsity is imposed on $\beta_2^*,\ldots,\beta_L^*$.  These directions can be estimated  by replacing $\sw$ and $\mu_\ell$ in \eqref{crit_EE} by their sample counterparts, and adding
  a penalty that encourages sparsity \citep{mai2019multiclass,gaynanova2016simultaneous,wang2021penalized,zeng2024subspace}.
    However, it is not always reasonable to assume that the directions $(\beta_2^*,\ldots,\beta_L^*)$ are sparse, especially when $\sw$ 
    has many non-negligible off-diagonal entries.
Moreover, the above procedure--based on solving a quadratic program--is less appealing than the regression-based approaches discussed below in \eqref{OS_B}, \eqref{OS_B_penalized_mat} and \eqref{new_procedure}, particularly in high-dimensional settings where structural or penalized estimation is desired and
    issues such as tuning parameter selection and computational efficiency become more critical. Also, see part (a) of our main contributions in   \cref{sec_contri}.
    
 \citet{Witten2011} alternatively use  the fact   that
  the first term in the LDA rule \eqref{bayes_rule} is equivalent with Fisher's discriminant rule
        \begin{equation}\label{Fisher_rule}
            \argmin_{\ell \in [L]} ~(x- \mu_\ell)^\T  F  {F^\T} (x-\mu_\ell),
        \end{equation} 
        where  the columns of $F = (F_1, \ldots, F_K)\in \RR^{p\times K}$, for some $K < L$, are defined via, for $k\in [K]$,
        \begin{equation}\label{FDA_B}
            \begin{split} 
            F_k := \argmax_{\beta\in\RR^p}   ~ \beta^\T \Sigma \beta \quad  \text{subject to}&  ~ \beta^\T \sw \beta = 1, ~ \beta^\T \sw F_i = 0, ~ ~\forall ~ i<  k.
            \end{split}
        \end{equation} 
     Criteria similar to \eqref{Fisher_rule} -- \eqref{FDA_B} and their connection to discriminant analysis have been considered in \cite{SafoAhn2016,ahn2021trace,jung2019penalized}.
    \citet{Witten2011} develop methodology  tailored to the high-dimensional scenarios with a diagonal matrix  $\sw$ and sparse columns of $F$.
    However, the formulation (\ref{FDA_B}) of the matrix $F$ makes it extremely difficult to obtain the global solution and to derive theoretical properties of the regularized estimator and the resulting classifier.  
    
    More closely related to our approach is optimal scoring which dates back to \citet{de1976additive,young1978principal} and regains attention in the LDA setting, see \citet{hastie1995penalized,hastie1994flexible,clemmensen2011sparse,Irina} and the references therein. By writing the response matrix and feature matrix as $\bY=(Y_1,\ldots,Y_n)^\T\in \{0,1\}^{n\times L}$ and $\bX = (X_1,\ldots,X_n)^\T\in \RR^{n\times p}$,
    optimal scoring is a regression-based approach that aims to solve a sequence of optimization problems: for $1\le k \le K < L$,  
        \begin{equation}\label{OS_B}
            \begin{split} 
            \min_{\ \theta \in \RR^L,\ \beta\in\RR^p} & ~   \|\bY \theta - \bX \beta\|_2^2\quad \text{s.t.}   ~ \theta^\T \bY^\T \bY \theta = n, ~   \theta^\T \bY^\T \bY \wh\theta_i  = 0, ~ \forall ~ i< k.
            \end{split}
        \end{equation} 
        Here $\wh \theta_i$ is the solution to (\ref{OS_B}) in the $i$th iteration.  
        In the low-dimensional  case ($p<n$), 
        \cite{hastie1995penalized} shows that the solutions $\wh \beta_1,\ldots, \wh \beta_K$ from \eqref{OS_B} can be computed from canonical correlation analysis (CCA), and are further related, via CCA, with the matrix $F$ obtained from the sample analogue of \eqref{FDA_B}. As a result, classification can be based on the dimension reduction directions $\wh \beta_1,\ldots, \wh \beta_K$. In the high-dimensional setting, regularization is needed.  For instance, \citet{hastie1995penalized} considers adding $\beta^\T\Omega \beta$ to the loss function of \eqref{OS_B} 
        for a chosen $p\times p$ symmetric, positive definite matrix $\Omega$, and show that  the global solution 
        can be computed. A possible choice of $\Omega= \lambda \bI_p$ for some $\lambda>0$ leads to  the ridge  penalty \citep{campbell1980shrunken,friedman1989regularized}. To accommodate more general regularization,  \cite{Irina} studies the matrix formulation of  \eqref{OS_B},  
         \begin{equation}\label{OS_B_penalized_mat}
            \begin{split} 
            \min_{\Theta, \ B} & ~   \|\bY \Theta- \bX B\|_F^2 + \text{pen}(B)\quad \text{s.t.}  ~ \Theta^\T \bY^\T \bY \Theta = n\bI_K, \quad \Theta^\T   \bY^\T \bY  1_n = 0_K,
            \end{split}
        \end{equation} 
        and proves that the
        global solution 
        to \eqref{OS_B_penalized_mat} can be computed, provided pen$(B)$  satisfies 
        $\text{pen}(B R) =  \text{pen}(B)$ for any orthogonal $K\times K$ matrix $R$. \cite{Irina} focuses on the group-lasso penalty $\|B\|_{1,2} = \sum_{j=1}^p \|B_{j\cdot}\|_2$ (which is orthogonal invariant) and analyzes the resulting classifier when $X\mid Y$ is Gaussian. However, it is not clear how the analysis can be extended to other penalized estimation. In particular, the requirement that pen($B$) is orthogonal invariant excludes the familiar lasso penalty $\|B\|_1 = \sum_{j,k}|B_{jk}|$ and the elastic net penalty. Worse, for penalties that are not orthogonal invariant, there is no guarantee of computing the global solution to \eqref{OS_B_penalized_mat}, let alone any theoretical property of the resulting estimator.

  \subsection{Our contributions}\label{sec_contri}

    {We propose a new regression based multi-class classification approach that is suitable in both low- and high-dimensional settings and has {\em provable theoretical guarantees.}}
  Key to our approach is the following reformulation of the problem.
We show in Lemma \ref{lem_key} of \cref{sec_method} that  the {\em discriminant directional matrix} in \eqref{bayes_rule} 
    \begin{equation}\label{def_B_star}
        B^* = \sw^{-1}   \bM \in \RR^{p\times L}
    \end{equation}
   and the {\em regression matrix} 
    \begin{equation}\label{def_B}
        B  := \Sigma^{-1}\sxy := [\Cov(X)]^{-1} \Cov(X, Y) \in \RR^{p\times L}
    \end{equation} 
    not only share the same column space, but have a closed-form connection, specifically,
    \begin{align}\label{key_BB}
        B^* &=  B H^{-1}
    \end{align}
    for some invertible $L\times L$ matrix $H$, which we derive explicitly. 
    This connection is novel and distinguishes our approach from the current literature. While existing literature has leveraged the fact that $B$ and $B^*$ share the same column space to develop two-step classification procedures, typically by first estimating this shared subspace (see, e.g., \citet{hastie1995penalized, ye2007least, lee2015equivalence, nie2022equivalence, gaynanova2016simultaneous, ahn2021trace}), these approaches lack theoretical guarantees for the resulting classifier. In contrast, our procedure is built upon the new closed-form characterization of $B$ given in \eqref{key_BB}, which not only leads to a straightforward classifier but also enables a rigorous theoretical analysis of its misclassification risk. To the best of our knowledge, only in the binary case ($L=2$), \cite{mai2012} employs sparse linear regression between an appropriately encoded label vector and the feature matrix $\bX$ for classification and proves consistency of their classifier.  Extension of such regression formulation to handle multi-class responses with high-dimensional features has been a longstanding open problem.
    
    In view of \eqref{key_BB}, we propose  to estimate the discriminant direction matrix $B^*$  by first estimating the regression matrix $B$ as 
    \begin{equation}\label{new_procedure} 
          \wh B =   \argmin_{B\in \RR^{p\times L}}   ~  {1\over n}  \|\bY  - \bX B\|_F^2 + \text{pen}(B)
    \end{equation} 
     and then estimating $H^{-1}$ based on $\wh B$ via its explicit expression in \cref{lem_key}. 
     The advantages of our approach as well as our main contributions are two-fold:
 \begin{enumerate}
     \item[(a)] From a methodological perspective, our approach in \eqref{new_procedure}, akin to optimal scoring, is regression-based. As pointed out by \citet{hastie1995penalized,hastie1994flexible}, regression-based methods are generally much easier to compute -- especially in high-dimensional settings -- than canonical correlation analysis \citep{Irina}, Fisher's discriminant rule \citep{Witten2011}, and the aforementioned procedures based on approximately solving in \eqref{crit_EE} as in \citep{mai2019multiclass,caizhang2019}. Moreover, regression-based approaches are more amenable to extensions involving regularized and non-parametric regression, as well as model selection. We can also easily tap into the large literature on estimation of the high-dimensional regression matrix $B\in \RR^{p\times L}$ in \eqref{new_procedure}. However, we caution that the nonlinearity of $\EE[\bY \mid \bX]$ in $\bX$ complicates the theoretical analysis.
     
     After computing $\wh B$ from \eqref{new_procedure}, we only need to estimate the inverse of an $L\times L$ matrix to estimate the discriminant direction matrix $B^*$ rather than that of a $p\times p$ matrix $\sw$ as in the classical LDA context. Our approach in \cref{sec_method_rule} uses the explicit form of $H$ in \cref{lem_key} and enjoys both computational advantage and numerical stability, compared to the existing approaches.   In \cref{sec_method_low}, we further offer an alternative way to estimate the LDA rule that performs Fisher's discriminant analysis on low-dimensional transformations $\wh B^\T x \in \RR^{L}$ of the original features $x\in \RR^p$.   This leads to a practically useful method, which could extract fewer discriminant directions than those in $B^*$, for both downstream analysis and visualization (see, for instance, the real data analysis in \cref{sec_real_data}). 
     
    \item[(b)]  Theoretically, unlike the intractable analysis of optimal scoring with general regularization, we are able to quantify the errors in estimating both $B$ and $B^*$, and provide theoretical guarantees for the resulting classifier, thanks to \eqref{key_BB}. Concretely, under the assumption that the distributions of $X\mid Y=\be_\ell$ are Gaussian $\cN_p(\mu_\ell, \sw)$, we provide in \cref{sec_theory_risk,sec_theory_Q} a general strategy for analyzing the excess misclassification risk of the proposed classifier, which is valid for any generic estimator of $B^*$, {including all existing estimators as well as our proposed estimator $\wh B^* = \wh B \wh H^{-1}$. }While various choices of $\wh B$  are feasible in our approach, more specifically in \eqref{new_procedure}, we apply the general result to two popular estimators: the lasso estimator in \cref{sec_theory_lasso} and the reduced rank estimator in \cref{sec_theory_rr}. The former also contrasts our approach with \citet{caizhang2019} and \citet{Witten2011} which assume sparsity of $B^*$ and $F$, respectively. The latter is particularly suitable when the rank of the discriminant direction matrix $B^*$ is low and the number of classes $L$ is moderate / large. For instance, our earlier work \citet{BW2023} considers a multi-class classification setting with  high-dimensional 
    features $X\in \RR^p$ that follow a factor model $X=AZ+W$ with unobserved (latent) features $Z\in \RR^r$  with $r<p$ and random noise $W$, that is independent of both $Z$ and the label $Y$. In this case, the discriminant direction matrix $B^*$ has rank no greater than $r$ when $r<L$. The procedure in \cite{BW2023}, however, is different than the one proposed here. We refer to \cref{rem_BW} in Section \ref{sec_method_rule} for more discussion.  
 \end{enumerate}

Finally, we provide a thorough simulation study in \cref{sec_sims} to corroborate our theoretical findings. \cref{sec_real_data} contains our findings in several real data studies.  The Supplement \citep{Supp_LDRR} contains all the proofs.

\subsection{Notation}
For any numbers $a,b\in \RR$, we write $a\vee b = \max\{a, b\}$. For any vector $v$ and $1\le q\le \i$, we use $\|v\|_q$ to denote the standard $\ell_q$-norm. For any positive integer $d$, we write $[d] := \{1,\ldots, d\}$. For any matrix $Q \in \RR^{d_1\times d_2}$, any $i\in [d_1]$ and $j\in [d_2]$, we write $Q_{i\cdot}$ for its $i$th row and $Q_{\cdot j}$ for its $j$th column. 
    For norms, we write $\|Q\|_\i = \max_{i,j}|Q_{ij}|$, $\|Q\|_F^2 = \sum_{i,j} Q_{ij}^2$, $\|Q\|_1 = \sum_{i,j} |Q_{ij}|$, $\|Q\|_{1,2} = \sum_i \|Q_{i\cdot}\|_2$ and $\|Q\|_\op = \sup_{v: \|v\|_2=1} \|Qv\|_2$. For any symmetric, semi-positive definite matrix $Q\in \RR^{d\times d}$, we use $\lambda_1(Q), \lambda_2(Q), \ldots, \lambda_d(Q)$ to denote its eigenvalues in non-increasing order. 
    For any sequences $a_n$ and $b_n$, we write $a_n\lesssim b_n$ if there exists some constant $C>0$ such that $a_n \le Cb_n$, and $a_n\asymp b_n$ if $a_n\lesssim b_n$ and $b_n \lesssim a_n$.

\section{Methodology}\label{sec_method}

    We discuss our classification procedure in this section.  Our goal is to estimate the rule in \eqref{bayes_rule} which entails estimating the discriminant functions: 
    \begin{align}\label{def_DF}
        \mu_\ell^\T \sw^{-1} \mu_\ell - 2x^\T \sw^{-1}\mu_\ell -2\log(\pi_\ell) =  \mu_\ell^\T B^*_{\cdot \ell} - 2x^\T B^*_{\cdot \ell}-2\log(\pi_\ell),   \quad  \text{for all }\ell \in [L].  
    \end{align} 
 Note that there is no need to estimate the $p\times p$ matrix $\sw^{-1}$ explicitly. 
      The following lemma establishes the connection between the discriminant directional matrix $B^*= (B_{\cdot 1}^*,\cdots,B_{\cdot L}^*)$ in (\ref{def_B_star}) and the regression matrix $B$ in (\ref{def_B}). Write 
    $\pi = (\pi_1,\ldots, \pi_L)^\T$ and $D_\pi = \diag(\pi)\in \RR^{L\times L}$. 
    \begin{lemma}\label{lem_key}
        We have 
       \begin{align}\label{bingo}
           B^* =  B H^{-1}
       \end{align}
        for a matrix  $H \in \RR^{L\times L}$ that  can be inverted and satisfies
        \begin{equation}\label{def_H}
           H = D_\pi -  B^\T\Sigma B= D_\pi -  D_\pi\bM^\T B = D_\pi -  B^\T \bM D_\pi.
        \end{equation}
    \end{lemma}
    \begin{proof}
        From the identity 
        \begin{align}\label{eq_sxy}
            &\sxy  = \EE(X Y^\T) - \EE(X)(\EE(Y))^\T = \bM D_\pi 
        \end{align}
        we see that
        \begin{equation}\label{eq_B}
            B = \Sigma^{-1}\bM D_\pi.            
        \end{equation}
        Next, using the fact that
        \begin{equation} \label{eq_Sigma}
             \Sigma =  \sw + \bM D_\pi \bM^\T,
        \end{equation}  
        the Woodbury matrix identity yields
        \[
            \sw^{-1} = \Sigma^{-1} + 
            \Sigma^{-1}\bM  \left(
                \bI_L -   D_\pi \bM^\T \Sigma^{-1}\bM 
              \right)^{-1}  D_\pi  \bM^\T \Sigma^{-1}.
        \]
        It follows that  
        \begin{align*}
            B^* &= \sw^{-1} \bM \\
            &= \Sigma^{-1}\bM  + \Sigma^{-1}\bM  \left(
                \bI_L -   D_\pi \bM^\T \Sigma^{-1}\bM 
              \right)^{-1}  D_\pi  \bM^\T \Sigma^{-1}\bM \\
            &= \Sigma^{-1}\bM   \left(
                \bI_L -  D_\pi \bM^\T \Sigma^{-1}\bM  
            \right)^{-1}  \\
            & = \Sigma^{-1}\bM  D_\pi \left(
                D_\pi -  D_\pi \bM^\T \Sigma^{-1}\bM  D_\pi
            \right)^{-1} \\
            & = B \left(D_\pi -  D_\pi \bM^\T B\right)^{-1} 
        \end{align*}
        which completes the proof. 
    \end{proof}

	\cref{lem_key} suggests that
	estimation of $B^*$ could be done by first estimating the matrix $B$ and then estimating $H^{-1}$ according to \eqref{def_H}. We discuss these two steps in detail in the next two sections.

    \subsection{Estimation of the regression matrix}

    To accommodate high-dimensional data, we estimate $B$ via penalized regression methods in \eqref{new_procedure}. 
    Choices of the penalty term depend on the concrete problem at hand. For instance, when entries of $B$ exhibit certain smoothness structure such as in some imaging or audio applications \citep{hastie1995penalized}, one could choose $\pen(B) = \tr(B^\T \Omega B)$ for some pre-specified ``roughness'' penalty matrix $\Omega$.
    In this paper we mainly focus on the following two types of structural estimation of $B$.  
    
    \subsubsection{Sparsity-based structural estimation.}  
    If we believe that only a subset of the original features are important for predicting the label, then sparsity structures of the regression matrix $B$ are reasonable assumptions.
     If the subset of features, that are pertinent for predicting the label, varies across different levels, then the entry-wise lasso penalty is appropriate. 
    The corresponding matrix estimator $\wh B$ is defined in \eqref{new_procedure} with the following lasso penalty \citep{Tibshirani1996}
     \begin{align}\label{pen:lasso}
         \pen(B) &= \lambda \sum_{j=1}^p \|B_{j\cdot}\|_1  
      =   \lambda \sum_{\ell=1}^L \| B_{\cdot \ell}\|_1 
     \end{align} 
     for some tuning parameter $\lambda>0$. In \cref{sec_theory_lasso} we provide a complete analysis of this estimator.
     
     If there exist many features that are not predictive for {\em all levels} of the response label, then the group-lasso penalty $\pen(B) = \lambda \sum_{j=1}^p \|B_{j\cdot}\|_2$ \citep{yuan2006model} is a more suitable choice. 
     In this case, one can retain a much smaller subset of important features and achieve feature selection. To further capture sparsity within the selected features, one might alternatively consider the sparse group lasso penalty $\pen(B) = \lambda \sum_{j=1}^p \|B_{j\cdot}\|_2 + \kappa \sum_{j=1}^p \sum_{\ell=1}^L |B_{j\ell}|$ \citep{VINCENT2014771,levy2023generalization}. By imposing an additional sparsity constraint within the non-zero rows, this approach allows for identifying specific relevant effects even within selected features to capture the local row-wise sparsity. From Lemma \ref{lem_key} we see that row-sparsity of $B$ is equivalent to that of $B^*$.  Our focus in this paper, however, is on prediction, not feature selection. Furthermore, one can also consider either a row-wise or a column-wise elastic net penalty \citep{zou2005regularization} if the features tend to (highly) correlate with each other. 
       Our numerical results in \cref{sec_sims,sec_real_data} include both the lasso penalty and the elastic net penalty.

    \subsubsection{Low-rank structural estimation.} Different from sparsity, if the goal is to construct a low-dimensional subset of linear combinations of the original feature for classification purposes, then a low-rank penalty in \eqref{new_procedure} could be more appropriate. For example, if the matrix of conditional means $\bM \in \RR^{p\times L}$ has a reduced-rank $r \ll \min\{p, L\}$, then \cref{lem_key} implies that $\rank(B) \le r$. In this case, estimation of $B \in \RR^{p\times L}$ in \eqref{new_procedure} could be done by using the following reduced-rank penalty (see, for instance, \cite{IzenmanRR})
    \begin{equation}\label{pen:rr}
        \pen(B) = \lambda ~ \rank(B)
    \end{equation}
    for some tuning parameter $\lambda > 0$. The reduced-rank regression estimator has a closed-form solution obtained via singular value decomposition. The procedure for deriving this solution is detailed in \cite{bunea2011}. A full theoretical analysis of using \eqref{pen:rr} is presented in \cref{sec_theory_rr}.  Some alternative penalties include, but are not limited to, the reduced-rank penalty plus the matrix ridge penalty in \cite{mukherjee2011reduced}, the nuclear norm penalty in \cite{Koltchinskii2011}, adaptive nuclear norm penalty in \cite{chen2013reduced} as well as the generalized cross-validation rank penalty in \cite{rank19}.
    \begin{remark}\label{rem_BW}
        We conclude this section by giving another concrete example under which $B$ has low rank. In \cite{BW2023} the authors consider the model $X = AZ+W$ with latent factors $Z\in \RR^r$ for some $r\ll p$ and some general additive noise $W$, that is independent of both $Z$ and the label $Y$. Conditioning on $Y = \be_\ell$ for all $\ell \in [L]$,  the low-dimensional factors $Z$ are assumed to follow Gaussian distributions with different conditional means $\alpha_\ell \in \RR^r$ but the same within-class covariance matrix $\Sigma_z \in \RR^{r\times r}$. When $W$ is also Gaussian, this model is a sub-model of \eqref{model}  with 
        $\mu_\ell = A\alpha_\ell$ and $\sw = A\Sigma_z A^\T + \Cov(W)$. When $r \le \min\{p,L\}$, it is easy to verify that $\rank(B) \le r$. The two-step approach in \cite{BW2023} 
        uses the {\em unsupervised} Principal Component Analysis (PCA) to reduce the feature dimension and then base classification on the reduced dimension. It is entirely different from our approach in \cref{sec_method_rule} but connected to the one in \cref{sec_method_low} in the way that our approach in \cref{sec_method_low} is also a two-step procedure but reduces the feature space in the first step in a {\em supervised} way via the regression step in \eqref{new_procedure}. Numerical comparison of both procedures are also included in our simulation study in \cref{app_sec_sim}.
    \end{remark}


    \subsection{Estimation of the discriminant functions}\label{sec_method_rule} 
    The discriminant functions in \eqref{def_DF}  necessitates that we estimate the conditional mean matrix $\bM$ and the discriminant directional matrix $B^*$. The latter also requires estimating the vector of prior probabilities $\pi$.
We adopt the following standard estimators
    \begin{align}\label{def_mu_pi_hat}
        \wh\bM = \bX^\T \bY (\bY^\T \bY)^{-1},\qquad  D_{\wh\pi} = \diag(\wh \pi) =  {1\over n}\bY^\T \bY
    \end{align}
    to estimate $\bM$ and $D_\pi$, respectively.
    Recall from \cref{lem_key} that $B^* = B H^{-1}$. Since the estimator $\wh B$ of $B$ is computed in \eqref{new_procedure},  in view of the closed-form of $H$ in \eqref{def_H}, we propose to estimate $H$ by the plug-in estimator
    \begin{equation}\label{def_H_hat}
   		\wh H = D_{\wh\pi} - {1\over n}  \wh B^\T \bX^\T \bX \wh B = {1\over n}\left(\bY^\T \bY - \wh B^\T \bX^\T \bX \wh B\right).
    \end{equation}
    Note that $\bX\wh B$ is simply
    the in-sample fit of the regression in \eqref{new_procedure}. Computation of $\bX\wh B$ thus can be done efficiently even in non-parametric regression settings, for instance, when the columns of $\bX$ consist of basis expansions.  
  Now,  using $B^* = BH^{-1}$,  we can further estimate $B^*$ by     
    \begin{equation}\label{def_B_star_est}
    	\wh B^* = \wh B \wh H^{+}.
    \end{equation}
    Here, for any matrix $M$, $M^+$ denotes its Moore-Penrose inverse.
   Finally, we estimate the discriminant functions in \eqref{def_DF}  by
    \begin{equation}\label{def_DF_hat}
         \wh\mu_\ell^\T \wh B^*_{\cdot \ell}   - 2x^\T \wh B^*_{\cdot \ell} -2\log(\wh \pi_\ell)  , \qquad \text{for all $\ell \in [L]$,}
    \end{equation}
    based on which classification can be done subsequently.
    Theoretically, we show in \cref{thm_Q} of \cref{sec_theory_Q} that $\wh H$ is invertible with overwhelming probability under a mild condition on $\wh B$ that is required for classification consistency. In \cref{thm_B_lasso} of \cref{sec_theory_lasso} and \cref{thm_B_rr} of \cref{sec_theory_rr}, this mild condition is verified for both the lasso estimator and the reduced rank estimator of $B$, respectively. Moreover, it is worth mentioning that 
     the inverse of $\wh H$  always exists  when any $\ell_2$-norm related penalty is deployed in \eqref{new_procedure}.


\section{Theoretical guarantees}\label{sec_theory}

    In this section, we provide a unified theory for analyzing the classifier based on the estimated discriminant functions in \eqref{def_DF_hat} via estimating the regression coefficient matrix. 
    In \cref{sec_theory_risk} we start with full generality and bound the excess risk of the classifier by a quantity that is related with the error of estimating $B^*$. The results there hold for any estimator $\wh B^*$ of $B^*$. In \cref{sec_theory_Q} we focus on the proposed estimator $\wh B^* = \wh B \wh H^+$ and provide a unified analysis of its estimation error that is valid for any generic estimator $\wh B$. Finally, we apply the general result to two particular choices of $\wh B$, the lasso estimator in \cref{sec_theory_lasso} and the reduced-rank estimator in \cref{sec_theory_rr}.  Throughout our analysis, both $p$ and $L$, as well as the parameters $M$ and $\sw$, are allowed to depend on the sample size $n$. For notational simplicity, we omit this dependence in our presentation.

    \subsection{A reduction scheme for bounding the excess risk}\label{sec_theory_risk}
    We adopt the following distributional assumption 
    \begin{equation}\label{model} 
        X  \mid Y =\be_\ell ~ \sim ~ \cN_p\left(\mu_\ell, \sw \right),\qquad \text{for all }  \ell \in [L].
    \end{equation}
    See \cref{rem_Gaussian} for extension to sub-Gaussian distributions of $X \mid Y$. 
    Let $(X, Y)$ be a new pair from \eqref{model} that is independent of the training data $\bD :=\{\bX, \bY\}$. The Bayes rule under \eqref{model}  reads as 
    \begin{equation}\label{def_G}
        g^*(x) = \argmin_{\ell \in [L]} G_\ell(x),\qquad \text{with }~  G_\ell(x) =  \mu_\ell^\T B^*_{\cdot \ell}\, - \, 2x^\T B^*_{\cdot \ell} - 2\log(\pi_\ell).
    \end{equation}
    For any classifier 
    \begin{equation}\label{def_G_hat}
        \wh g(x)= \argmin_{\ell \in [L]} \wh G_\ell(x),\qquad \text{with }~  \wh G_\ell(x) =  \wh\mu_\ell^\T \wh B^*_{\cdot \ell}   - 2x^\T \wh B^*_{\cdot \ell} - 2\log(\wh \pi_\ell),
    \end{equation} 
    its excess risk relative to the error rate of the Bayes rule is defined as 
    \[
    \cR(\wh g) := \PP\{Y \ne \wh g(X) \mid \bD \} - \PP\{Y \ne g^*(X)\}.
    \]
    Note that we use the identification $\{ Y=k\}  := \{ Y= \be_k\}$ in the above definition. Our analysis of $\cR(\wh g)$ requires the following condition on the priors $\pi = (\pi_1, \ldots, \pi_L)^\T$.
	\begin{ass}\label{ass_pi}
		There exist some absolute constants $0<c\le C< \i$ such that 
		\begin{align}
		   { c\over L} \le \min_{k\in[L]}\pi_k \le \max_{k\in[L]} \pi_k \le {C\over L}.
		\end{align}
	\end{ass}
	This is a common regularity assumption in binary classification problems (see, for instance, \cite{caizhang2019,mai2019multiclass,AP19}). We may relax Assumption \ref{ass_pi}  to $\min_{k\in[L]}\pi_k \ge c \log(n)/n$ in our analysis (see, for instance, \cref{lem_pi} of \cref{app_sec_aux_lem}) at the cost of more technical proofs and a less transparent presentation. 
    
    An important quantity in this problem is the pairwise Mahalanobis distance in the feature space between any two label classes. Let 
    $
             \Dt := \max_{k\ne\ell}(\mu_k-\mu_\ell)^\T \sw^{-1}(\mu_k-\mu_\ell).
    $
    \begin{ass}\label{ass_cond_separation}
        There exists 
        some absolute constant $c\in (0,1]$  such that 
    \begin{equation}\label{cond_separation} 
        \min_{k\ne \ell}(\mu_k-\mu_\ell)^\T \sw^{-1}(\mu_k-\mu_\ell)  \ge c \Dt.
    \end{equation}
    \end{ass}
    The following proposition states that for the purpose of bounding $\cR(\wh g)$, it suffices to study the estimation error of each discriminant function, that is,
    $\wh G_\ell(X) - G_\ell(X)$ for $\ell \in [L]$.
    The proof essentially follows the argument of proving Theorem 12 in \cite{BW2023} but with modifications to capture  {\em the sum of squared errors} of estimating all discriminant functions.

    \begin{prop}\label{prop_g}
    Under model \eqref{model} with \cref{ass_pi} and \cref{ass_cond_separation}, for any positive sequences $t_1,\ldots, t_L$ (that possibly depend on $\bD$), with probability equal to one, there exists an absolute constant $C$ such that
    \begin{equation}\label{fast_rate}
    \cR(\wh g) ~ \le  ~  
         {C\over  \sqrt \Dt}\sum_{\ell=1}^L  t_\ell^2  +   \sum_{\ell=1}^L \PP\left\{
            |\wh G_\ell(X) - G_\ell(X)| \ge t_\ell \mid \bD  
        \right\}.
    \end{equation}
    \end{prop}
    \begin{proof}
        See \cref{app_sec_proof_prop_g}.
    \end{proof}
    We remark that both \cref{ass_pi} and   \cref{ass_cond_separation} are needed to derive the above ``fast-rate'' bound in \eqref{fast_rate}. Although our proof reveals that the following ``slow-rate'' bound 
    \begin{equation}\label{slow_rate}
        \cR(\wh g) \le \max_{1\le \ell\le L} t_\ell + \sum_{\ell=1}^L \PP\left\{
            |\wh G_\ell(X) - G_\ell(X)| \ge t_\ell \mid \bD  
        \right\}
    \end{equation} 
    holds without these assumptions,  we will mainly focus on explicit excess risk bounds based on  \eqref{fast_rate}.

    \cref{prop_g} is valid for any estimator $\wh G_\ell$ of the discriminant function, in particular, for our proposed  $\wh G_\ell$ in \eqref{def_G_hat}. The choice of $(t_1,\ldots, t_L)$ depends on the estimation error of $|\wh G_\ell(X) - G_\ell(X)|$.  We therefore proceed to analyze 
    \begin{align}\label{diff_G}
		\wh G_\ell(X)  -  G_\ell(X)  &= 
		(\wh \mu_\ell - \mu_\ell)^\T B^*_{\cdot \ell} + ( \wh\mu_\ell- 2X)^\T (\wh B^*_{\cdot \ell} - B^*_{\cdot \ell}) + 2\log\left({\pi_\ell / \wh \pi_\ell}\right)
	\end{align}
    for each $\ell \in [L]$. It is worth mentioning that in this section we allow $\wh B^*$ in the above display to be any estimator of $B^*$ whereas $\wh \mu_\ell$ and $\wh \pi_\ell$ are given in \eqref{def_mu_pi_hat}.  
    \cref{ass_pi}, the Gaussian tail of $X \mid Y$ and the independence between $X$ and $(\bX, \bY)$ in \eqref{diff_G} entail the following upper bound of $|\wh G_\ell(X) - G_\ell (X)|$.

	\begin{theorem}\label{thm_G}
		Under model \eqref{model} with Assumptions \ref{ass_pi} and \ref{ass_cond_separation}, assume 
		$\Delta \ge 1$ and $L\log(n) \le  n$.  With probability at least $1- 3n^{-1}$, we have,  for all $\ell \in [L]$,  
    	\begin{align*}
    		 \left|\wh G_\ell (X) - G_\ell (X)\right|  ~ \lesssim  ~    
    		\sqrt{\Delta}\sqrt{L\log(n) \over n} &+ \|\sw^{1/2}(\wh B^*_{\cdot \ell} -B^*_{\cdot \ell})\|_2 \sqrt{\Delta  + \log(n)} +    |(\wh \mu_\ell - \mu_\ell)^\T (\wh B^*_{\cdot \ell}- B^*_{\cdot \ell})|.
    	\end{align*}
    \end{theorem}
    \begin{proof}
        See \cref{app_sec_proof_thm_G}. 
    \end{proof}

    Condition $L\log(n) \le n$ together with \cref{ass_pi} ensures that $\wh \pi_\ell \asymp \pi_\ell$ for all $\ell \in [L]$. The  reasonable condition $\Delta  \ge 1$   requires that the pairwise separation of conditional distributions between distinct classes does not vanish. It is assumed  to simplify the presentation, but our analysis can be easily extended to the case  $\Delta = \Delta(p) \to 0$ as $p\to \i$.
     We refer to \cite{BW2023}, in particular
Theorems 3 \& 7 and Corollary 11 for $L=2$ and Theorem 12 for $L>2$.

    Combining \cref{prop_g} with \cref{thm_G} immediately yields the following corollary.  Denote by $\EE_{\bD}$ the expectation taken with respect to the training data $\bD$ only. For any estimator $\wh B^*$ of $B^*$, define  the following (random) quantity
	\begin{equation}\label{def_Q_metric}
		\cQ(\wh B^*) :=   \|\sw^{1/2}(\wh B^* -B^*)\|_F^2 (\Delta  + \log n)  +    \sum_{\ell=1}^L \left[(\wh \mu_\ell - \mu_\ell)^\T (\wh B^*_{\cdot \ell}- B^*_{\cdot \ell})\right]^2.
	\end{equation}

    \begin{cor}\label{cor_g}
        Under model \eqref{model} with Assumptions \ref{ass_pi} and   \ref{ass_cond_separation}, assume 
		$\Delta  \ge 1$. For any estimator $\wh B^*$ of $B^*$,  the corresponding classifier $\wh g$ in \eqref{def_G_hat}  satisfies 
        \[
            \EE_\bD\left[\cR(\wh g)\right]  ~ \lesssim ~  \EE_{\bD}\left[\min\left\{\Delta L^2 {\log(n) \over n} + \cQ(\wh B^*), ~ 1\right\}\right].
        \]
    \end{cor}
    \begin{proof}
        See \cref{app_sec_proof_cor_g}. 
    \end{proof}

    \cref{cor_g} ensures that bounding  $\cR(\wh g)$ can be reduced to controlling $\cQ(\wh B^*)$ for any estimator $\wh B^*$ of $B^*$.  In particular,   \cref{cor_g} is valid for any existing approach based on estimating $B^*$, such as the estimator proposed in \cite{mai2019multiclass} to which our result could yield a smaller risk bound using their rate-analysis of $\|\wh B^*-B^*\|_F^2$. In the next section we provide a unified analysis of $\cQ(\wh B^*)$ for our proposed estimator $\wh B^* = \wh B \wh H^+$ based on a generic $\wh B$.   As shown later, the first term in \eqref{def_Q_metric} is typically the dominating term.

    \subsection{A unified analysis of the estimation of $B^*$ based on generic  regression coefficient estimation}\label{sec_theory_Q}
	
	In this section we provide more explicit bounds  of  $\cQ(\wh B^*)$ in \eqref{def_Q_metric} for the estimator 
	$ 
        \wh B^* = \wh B \wh H^+
	$
	with $\wh H$ obtained from \eqref{def_H_hat}.  Our results in this section are valid for any estimator $\wh B$ of $B$ that satisfies the following assumption.
	\begin{cond}
	    \label{ass_B}
		There exists some deterministic sequences $\omega_1$ and   $\omega_2$ such that the following holds with probability at least $1-n^{-1}$, 
		\begin{align*}
			   \max\left\{n^{-1/2}\| 
      \bX(\wh B - B)\|_F, ~ \|\sw^{1/2}(\wh B - B)\|_F \right\} \le  \omega_2,  \quad\|\wh B - B\|_{1,2} \le \omega_1.
		\end{align*}
	\end{cond}
	 
	We need the following regularity conditions on $\sw$. It assumes the entries in $\sw$ are bounded from above and the smallest eigenvalue of $\sw$ is bounded away from zero. The lower bound condition on $\lambda_p(\sw)$ can be relaxed to a restricted eigenvalue condition on $\sw$, see the discussion after \cref{thm_B_lasso}.
	\begin{ass}\label{ass_sw}
		There exist some positive constants  such that 
		$c< \lambda_p(\sw) \le  \|\sw\|_\i \le C<\i$.
	\end{ass}  
	
	
	\begin{theorem}\label{thm_Q}
		Under model \eqref{model} with Assumptions \ref{ass_pi}, \ref{ass_cond_separation} \& \ref{ass_sw}, assume $\Delta \asymp 1$ and $ L\log (n) \le c n$ for some small constant $c>0$. For any $\wh B$ satisfying \cref{ass_B} with  
		$	  \omega_2\sqrt{L} \le c, $
		we have, with probability at least $1-  2n^{-1}$, that
  \begin{align*}
      \text{the matrix 
      $\wh H$ in \eqref{def_H_hat} is invertible,}
  \end{align*} and 
		\begin{align}\label{bd_Q_all}
            & \cQ(\wh B^*)	\lesssim  ~      L^2 {\log^2(n) \over  n}     +  \omega_2^2 L^2 \log(n)    + \omega_1^2 L^3 {\log(n\vee p) \over n}.
        \end{align}
	\end{theorem}
	\begin{proof}
		See \cref{app_proof_thm_Q}.
	\end{proof}
	
	\cref{thm_Q} relates $\cQ(\wh B^*)$ to the estimation error of $\wh B$ only in  both the Frobenius norm and the $\ell_1/\ell_2$-norm. 
	In the bound of \eqref{bd_Q_all}, the second term on the right-hand-side is usually the leading term and its rate depends on specific choice of $\wh B$.  This allows us to exploit certain properties, such as sparsity or low-rank, of the regression matrix $B$. We provide concrete examples in \cref{sec_theory_lasso} for the lasso estimator and in \cref{sec_theory_rr} for the reduced-rank estimator.
	
	Condition $\omega_2\sqrt{L} \le c$ ensures that, with high probability, $\lambda_K(\wh H) \ge \lambda_K(H)/2$ holds so that  $\wh H$ is invertible, as $H$ is invertible by \cref{lem_key}.  Furthermore, the bound in \eqref{bd_Q_all} suggests that the condition $\omega_2\sqrt{L} \le c$ is also needed for consistent classification. 
	 

	\subsection{Estimation of $B$ with entry-wise $\ell_1$-regularization}\label{sec_theory_lasso}

    In this section we focus on the case where the columns of $B$ are (potentially) sparse and allow for different sparsity patterns across its columns. Therefore, for some positive integer $s \le p$, we consider the following parameter space of $B$:
    \begin{equation*}
        \Theta_s := \left\{
        B \in \RR^{p\times L}: \max_{\ell \in [L]}\|B_{\cdot \ell}\|_0 \le s
        \right\}.
    \end{equation*} 
    For ease of presentation, we simply use the largest number of non-zero entries among all columns of $B$. However, our analysis can be easily extended to  $\|B_{\cdot \ell}\|_0 \le s_\ell$ by allowing different $s_\ell$ for $\ell \in [L]$. Our analysis also needs the following mild regularity condition on the matrix $\bM$ of conditional means. 
    \begin{ass}\label{ass_mu}
        There exists an absolute constant $C <\i$ such that $\|\bM\|_\i \le C$.
    \end{ass}   
 
    The following theorem establishes explicit rates of $\omega_2$ and $\omega_1$ in \cref{ass_B}  
    for the regression matrix estimator $\wh B$ in \eqref{new_procedure} using the lasso-penalty in \eqref{pen:lasso}. 
 
    \begin{theorem}\label{thm_B_lasso}
        Under model \eqref{model} with Assumptions \ref{ass_pi}, \ref{ass_cond_separation}, \ref{ass_sw} and \ref{ass_mu},  
        assume that there exists  some sufficiently small constant $c>0$ such that
  $B \in \Theta_s$ for some $1\le s \le p$ with
        $
            (s\vee L)\log(n\vee p)\le c ~ n.
        $ 
    Let the estimator 
     $\wh B$ be obtained from \eqref{new_procedure} with any $\lambda$ satisfying
            \begin{equation}\label{rate_lbd}
                \lambda \ge C \sqrt{\|\sw\|_\i + \|\bM\|_\i^2}\sqrt{\log(n\vee p) \over  n L}
            \end{equation}
         and some finite, large enough constant $C>0$.  
    Then, with probability at least $1-3n^{-1}$, \cref{ass_B} holds for 
    \[
        \omega_2 \asymp  
                \lambda\sqrt{s L},\qquad \omega_1 \asymp 
                \lambda   s L.
    \] 
    \end{theorem}
    \begin{proof}
        See \cref{app_sec_proof_thm_B_lasso}.
    \end{proof}

    Our proof is in fact based on a weaker condition than  \cref{ass_sw} by replacing $\lambda_p(\sw)$ with the Restricted Eigenvalue condition (RE) on $\sw$ (see \cref{ass_RE} in   \cref {app_sec_proof_thm_B_lasso} of the supplement). 
    The  two main difficulties in our proofs are
    (a) to control $\max_{j\in[p]} \|(\bY - \bX B)^\T \bX \be_j\|_\i $, and
    (b) to bound from below the restricted eigenvalue of   $\wh \Sigma=n^{-1} \bX^\T \bX$. 
     For (a), as pointed out in \citet{Irina}, the difficulty in this model is elevated relative to the existing analysis in sparse linear regression models due to the fact that $\bY$ is not a linear model in $\bX$. We establish a sharp control of $\max_{j\in[p]} \|(\bY - \bX B)^\T\bX\be_j\|_\i$ in \cref{lem_XE}. For (b), the existing result in \citet{RudelsonZhou} is not readily applicable as rows of $\bX$ are generated from a mixture of Gaussians. In Lemma~\ref{lem_Sigma_hat_sup} we establish a uniform bound of $v^\T (\wh \Sigma - \Sigma) v$ over $s$-sparse vectors $v$ which, in conjunction with the reduction arguments in \citet{RudelsonZhou}, proves (b). Our analysis allows the number of classes $L$ to grow with the sample size $n$
    
         \citet{Irina}  considers  a group-lasso regularized regression similar to \eqref{new_procedure} and derives the bound of $\max_{j\in[p]} \| (\bY' - \bX B')^\T\bX\be_j\|_2$ for some different response matrix $\bY'$ and different target matrix $B'$. Their proof is different and in particular requires that $L$ is fixed, independent of $n$.  

    In the rest of the paper, we assume that $\lambda$ is chosen as the order specified in \eqref{rate_lbd}.  Although both $\|\sw\|_\i$ and $\|M\|_\i$ in \eqref{rate_lbd} could, in principal, be consistently estimated, we recommend,
    on the basis of our regression formulation, to select $\lambda$ in practice simply  via cross-validation, for instance, via the built-in function \textsf{cv.glmnet} of the R-package \textsf{glmnet}.\\

    As an immediate corollary of \cref{cor_g}, \cref{thm_Q} and \cref{thm_B_lasso}, we have the following upper bound of the misclassification rate of $\wh g$ based on the entry-wise lasso estimator with $\lambda$ chosen as the rate in \eqref{rate_lbd}.
    
    \begin{cor}\label{cor_lasso}
        Under the conditions in \cref{thm_B_lasso}, assume that 
		$\Delta  \asymp 1$.
        Then
    	\[
            \EE_{\bD}\left[\cR(\wh g)\right] ~ \lesssim ~  \min\left\{ {s L^2 \over n}\log^2(n\vee p), ~ 1\right\}.
    	\]
    \end{cor}

    Our results allow both the sparsity level $s$ and the number of classes $L$ to grow with the sample size $n$. The existing literature \citep{caizhang2019,Irina} that analyze the excess risk under model \eqref{model}
    only consider the case $L \asymp 1$, while \citet{AP19} does allow $s$, $L$ and $n$ to grow, but studies   the risk $\PP\{Y\ne g(X)\}$ instead of the excess risk $\cR(g)$. Guarantees on the excess risk are stronger since they, together with the Bayes error, which is oftentimes easy to calculate, imply guarantees of the risk. 

    \cref{cor_lasso}  focuses on the statistically most interesting case 
    $\Dt\asymp 1$ where the pairwise separation between different classes is of constant order.
    Otherwise, if $\Dt\to \i$ as $p\to\i$, 
    the classification problem becomes much easier and we have super-fast rates of convergence (see, for instance, \citet{caizhang2019,BW2023}). Indeed, 
     using the arguments in the proof of \citet[Theorem 12]{BW2023}, our analysis can be modified easily  to show that 
    \[
        \EE_{\bD}\left[\cR(\wh g)\right]    ~ \lesssim ~  {sL^2\over n}\log^2(n\vee p) e^{-c\Delta} + n^{-C_0}
    \]
    for some arbitrary large constant $C_0>0$ and some constant $c>0$. The first term on the right in the above rate is exponentially fast in $\Delta$.

  We conclude this section with a comparison 
  between our penalized linear discriminant method
  and the penalized multinomial logistic regression
  approach in \cite{levy2023generalization,abramovich2021multiclass,lei2019data}.
  The main difference between both approaches is that logistic regression models the conditional distribution of the label $Y$ given the feature $X=x$ via
  \[ \log { \PP\{ Y=1\mid X=x\} \over \PP\{ Y=k\mid X=x\}} = \langle\beta_k^*, x\rangle + \beta_{0,k}^*
  \]
  while discriminant analysis postulates the distribution of the feature $X$  given the label $Y=k$.
This distinction complicates any comparison between both methods.  In particular, comparison of minimax lower bounds between the two methods is irrelevant. The sparse linear classifier $\wh g$ in
 \cite{abramovich2021multiclass} satisfies
   \begin{equation}\label{rate_sparse_logis_reg}
        \EE_{\bD}\left[\cR(\wh g)\right] ~  \lesssim ~ \sqrt{sL + s\log(e p/s) \over n},
   \end{equation} 
    when the  true coefficients $\beta_k^*$  are $s$-sparse ($s\le p$). Furthermore, \cite{abramovich2021multiclass} proves that this rate is minimax optimal in their sparse setting. 

 We emphasize that \cite{levy2023generalization,abramovich2021multiclass,lei2019data} assume that the directions $\beta_k^*$ are sparse.
 In contrast, we assume structure on the regression matrix $B$ (not the matrix $B^*$ of directions).
 In particular, if we assume that our regression matrix has $s$ entire rows equal to zero, then $B^*$ has $s$-sparse columns.
 
 In this setting, we observe that
 the “slow-rate” bound in \eqref{slow_rate} implies a risk bound of order $\sqrt{sL/n}$ for our procedure (up to logarithmic factors) and that this rate matches the rate in  \eqref{rate_sparse_logis_reg}.  Hence, the excess risk of our classifier is no worse than the bound in \eqref{rate_sparse_logis_reg} and is in fact faster,
 using the rate in \cref{cor_lasso},
 when either $sL^3 \ll n$ or $\Dt \gg \log(L)$. 
 [A reasonable assumption, given that we primarily focus on small to moderate values of $L$ in this paper.  Determination of  the optimal dependence on $L$ in the excess risk under our model remains an interesting open question, which we leave for future research.] 
 Finally, although \cite{abramovich2021multiclass} shows that the rate in \eqref{rate_sparse_logis_reg} can be improved under an additional margin condition that extends the one introduced by \cite{tsybakov2004optimal} in binary classification, a direct comparison with our risk bound is challenging, as verifying this margin condition under our model poses a difficult problem in its own right.

	\subsection{Estimation of $B$ via reduced-rank regression}\label{sec_theory_rr}

    We analyze the classifier that estimates the regression coefficient matrix $B$ by reduced-rank regression. 
    Specifically, let $\wh B$ be obtained from  \eqref{new_procedure} with $\pen(B) = \lambda~  \rank(B)$ for some tuning parameter $\lambda > 0$. The following theorem establishes the rates of $\omega_2$ and $\omega_1$ for which  \cref{ass_B} holds. Write $r = \rank(B)$.

   \begin{theorem}\label{thm_B_rr}
   		Under model \eqref{model} with Assumptions \ref{ass_pi},\ref{ass_cond_separation} and \ref{ass_sw}, assume 
   		$
   			 (p + L) \log(n) \le  c ~ n
   		$
   		for some sufficiently small constant $c>0$. 
   		For the estimator $\wh B$ in \eqref{new_procedure} with any
   		\begin{equation}\label{rate_lbd_rr}
   				\lambda \ge  C \sqrt{\|\sw\|_\op(1+\Delta)}\sqrt{(p+L)\log(n) \over  n L}
   		\end{equation}
        and some large constant $C>0$, 
   		 with probability at least $1-2n^{-1}$, \cref{ass_B} holds for 
   		\[
   			\omega_2 \asymp  \sqrt{\lambda r},\qquad \omega_1  \asymp  \sqrt{\lambda  r p}.
   		\]
   \end{theorem} 
	\begin{proof}
		The proof is deferred to \cref{app_sec_proof_thm_B_rr}.
	\end{proof}

	For the reduced-rank estimator, establishing its in-sample prediction risk can be done without any condition on the design matrix (see, for instance, \cite{bunea2011,giraud2011}).  However, since \cref{ass_B} is also related with the out-of-sample prediction risk, we need the smallest eigenvalue of the Gram matrix $\wh \Sigma$ to be bounded away from zero. The inequality $(p+L)\log(n) \le cn$ is assumed for this purpose. 
     In case  $p$ is large, it is recommended that an additional ridge penalty be added in $\text{pen}(B)$ to alleviate the singularity issue of $\wh \Sigma$ \citep{mukherjee2011reduced}.  Another main ingredient of our proof of \cref{thm_B_rr} is to control $\|\bX^\T(\bY - \bX B)\|_\op$ in \cref{lem_XE_op}. It requires non-standard analysis for the same reason mentioned before that $\bY$ is not linearly related with $\bX$.

   Combining \cref{cor_g}, \cref{thm_Q} and \cref{thm_B_rr} immediately yields the following guarantee on the excess risk of the classifier using the reduced-rank estimator of $B$. 
   
   \begin{cor}\label{cor_rr}
   Under model \eqref{model} with Assumptions \ref{ass_pi} and \ref{ass_cond_separation}, assume 
   		$
   			 (p + L) \log(n) \le  c ~ n
   		$
   		for some sufficiently small constant $c>0$. Further assume $\Delta \asymp 1$ and $\lambda_p(\sw) \asymp \lambda_1(\sw) \asymp 1$.   
        Then  
   		\[
       		\EE_{\bD}\left[\cR(\wh g)\right]  ~ \lesssim ~  
   		   \min\left\{{(p+L)  r L \over n}\log^2(n) , ~ 1\right\}.
   		\]
   \end{cor} 
    Note that for the $p\times L$ matrix $B$ with rank $r$, its effective number of parameters is exactly $(p+L)r$ \citep{IzenmanRR} which is much smaller than $pL$ when $r\ll (p\wedge L)$. \cref{cor_rr} suggests that the classifier using the reduced-rank estimator has benefit when $r$ is relatively small comparing to $p$ and $L$. 
    
     The same remark as the one at the end of \cref{sec_theory_lasso} applies, and our analysis can be extended to $\Delta \to \i$, where the rate in \cref{cor_rr} becomes exponentially fast in $\Delta$.  Comparing to  \cref{ass_sw},  a stronger condition of bounded eigenvalues of $\sw$ is assumed here to simplify the presentation.  The latter holds, for instance, when the features in $X$ exhibit weak dependence or are uncorrelated within each response class. 

     We  compare  our penalized linear discriminant method and the penalized multinomial logistic regression approach in \cite{levy2023generalization} in the low-rank setting
 when the rank of the true coefficient matrix is less than $r$.
  The low-rank classifier $\wh{g}$ in \cite{levy2023generalization} satisfies
    \begin{equation}\label{rate_lr_logis_reg}
        \EE_{\bD}\left[\cR(\wh g)\right] ~  \lesssim ~ \sqrt{r(L-1 + p) \over n},
   \end{equation} 
and
we 
observe that
 the “slow-rate” bound in \eqref{slow_rate} implies a risk bound of order $\sqrt{r(L+p)/n}$ for our procedure (up to logarithmic factors) and that this rate matches the rate in  \eqref{rate_lr_logis_reg}.
 Hence, the excess risk of our classifier is no worse than the bound in \eqref{rate_lr_logis_reg} and is in fact faster,
 using the rate in \cref{cor_rr},
 when $r(p+L)L^2 \ll n$. 

    \begin{remark}[Extension to sub-Gaussian conditional distributions]\label{rem_Gaussian}
        Most of our theoretical results, for instance, \cref{thm_G,thm_Q,thm_B_lasso,thm_B_rr},  do not require the Gaussianity of $X \mid Y$ in \eqref{model} and can be easily generalized to sub-Gaussian distributions. The Gaussianity plays a key role in deriving \cref{prop_g} which further leads to the fast rate of convergence of the excess risk bounds $\cO(sL^2/n)$ in \cref{cor_lasso} and $\cO((p+L)rL/n)$ in \cref{cor_rr}, modulo the logarithmic factors. When $X\mid Y$ is not Gaussian but sub-Gaussian, the excess risk $\cR(\wh g)$ is defined relative to $g^*$ in \eqref{def_G} rather than the Bayes rule. By adopting this notion, a straightforward modification of our proof leads to
        \[
             \cR(\wh g) ~ \le  ~  
            \max_{1\le \ell \le L} t_\ell  +  \sum_{\ell=1}^L \PP\left\{
                |\wh G_\ell(X) - G_\ell(X)| \ge t_\ell \mid \bD  
            \right\}.
        \] 
        Consequently the bounds of $\cR(\wh g)$ in \cref{cor_lasso} and \cref{cor_rr} converge in  slower rates as  $\cO( \sqrt{sL/n})$  and $\cO(\sqrt{(p+L)rL/n})$, respectively.  Finally, faster rates could be derived under certain margin conditions, such as  
        \[  
           \max_{\ell \in [L]} ~ \PP\left(0<\min_{k\in [L]\setminus \{\ell\}} G_k(X) - G_\ell(X) <2t \mid Y=\be_\ell\right) \le  C t^\alpha,\quad \text{for all }t\ge 0,
        \] 
        and for some $C\ge 0$ and $\alpha \ge 1$.
    \end{remark} 

\section{Simulation study}\label{sec_sims}
    In our simulation study, we evaluate our proposed procedure in \cref{sec_method_rule} that we term as Linear Discriminant Regularized Regression (LDRR). We also include the procedure in \cref{sec_method_low}  that is based on Fisher's discriminant rule after estimating $B$ (called LDRR-F).  For estimating $B$ from \eqref{new_procedure}, in the sparse scenarios of  \cref{sec_sims_sparse}, we consider both the lasso penalty \eqref{pen:lasso} (L1)  and the elastic net penalty (L1+L2), $\text{pen}(B) = \lambda \sum_{\ell=1}^L[\alpha \|B_{\cdot \ell}\|_1 + (1-\alpha)/2\|B_{\cdot \ell}\|_2^2]$, for some $\alpha \in [0,1]$ and $\lambda >0$ chosen by cross-validation from the R-package \textsf{glmnet}. In the low-rank scenarios of \cref{sec_sims_lowrank}, we consider both the reduced-rank penalty \eqref{pen:rr} (RR) and the reduced-rank penalty plus the ridge penalty (RR+L2), $\text{pen}(B) = \lambda [\alpha~ \rank(B) + (1-\alpha)/2 \|B\|_F^2]$, with some $\alpha \in [0,1]$ and $\lambda > 0$ chosen by cross-validation. 

    We compare our proposed methods with  the nearest shrunken centroids classifier (PAMR) of \citet{Tibshirani2002}, the shrunken centroids regularized discriminant analysis (RDA) method of \citet{guo2007regularized}, the $\ell_1$-penalized linear discriminant (PenalizedLDA) method of  \citet{Witten2011}, the multiclass sparse discriminant analysis (MSDA) method of \citet{mai2019multiclass} and the sparse multiclass discrimination with trace regularization (SLDTR) method of \citet{ahn2021trace}. These methods are available in  R-packages \textsf{pamr}, \textsf{rda}, \textsf{penalizedLDA}, \textsf{msda} and   Matlab codes for \textsf{SLDTR} and we use their default methods for selecting tuning parameters.  Finally,  the oracle procedure that uses the true values of $\bM$, $\Sigma_w$, and $\pi$ in computing the discriminant functions $G_k(x)$ in  \eqref{def_G} serves as our benchmark.  
 
    \subsection{A practical classification based on reduced dimensions}\label{sec_method_low}

    Since in practice it is often desirable to extract fewer discriminant directions than those in $B^*$, we discuss in this section an alternative procedure for this purpose, which estimates the LDA rule in \eqref{bayes_rule} by also using $\wh B$ from \eqref{new_procedure} but in a different way.

    We start with the fact (see  \cref{app_sec_proof_equiv} for a proof) that the rule in \eqref{bayes_rule} is equivalent to 
    \begin{equation}\label{Bayes_FDA_B_star}
         \argmin_{\ell \in [L]} ~(x- \mu_\ell)^\T B^* (B^{*\T}\sw B^*)^+ B^{*\T} (x-\mu_\ell)-2\log(\pi_\ell)  
    \end{equation}  which, by \cref{lem_key},  further equals to 
    \begin{equation}\label{Bayes_FDA_B}
        \argmin_{\ell \in [L]} ~(x- \mu_\ell)^\T B (B^\T\sw B)^+ B^\T (x-\mu_\ell) -2\log(\pi_\ell)  .
    \end{equation} 
    The formulation in \eqref{Bayes_FDA_B} has the following two-step interpretation:
    \begin{enumerate}
        \item [(1)] Transform the original feature $x\in \RR^p$ to $z:= B^\T x$, with the latter belonging to a subspace of $\RR^L$;
        \item[(2)] Perform linear discriminant analysis on the space of $z$.
    \end{enumerate}
    The second point follows by noticing that 
    $
        C_w := B^\T \sw B \in \RR^{L\times L}
    $ 
    is indeed the within-class covariance matrix of $Z=B^\T X$ under model \eqref{model}.
    
    Due to the fact that $C_w$ is rank deficient with $\rank(C_w) < L$ and for the purpose of extracting fewer discriminant directions, we adopt the equivalent formulation of \eqref{Bayes_FDA_B} from the Fisher's perspective in the sequel. Write the between-class covariance matrix of $Z = B^\T X$ as 
    $
        C_b := B^\T \bM D_\pi \bM^\T B.
    $
    For a given reduced dimension $1\le K < L$, we can find $K$ discriminant directions in the space of $Z$ by solving, for each $k\in[K]$, the optimization problem
     \begin{equation*} 
        \begin{split} 
         \alpha_k  = \argmax_{\alpha\in\RR^L} & ~ \alpha^\T  C_b  ~ \alpha \quad \text{ s.t.}  ~ \alpha^\T C_w  ~  \alpha = 1,  ~ \alpha^\T C_w  ~\alpha_i = 0, ~~ \forall ~ i< k.
        \end{split}
    \end{equation*}
    The resulting classification rule based on these $K$ discriminant directions is 
    \begin{equation}\label{Bayes_FDA_B_approx}
        \argmin_{\ell \in [L]} ~(x- \mu_\ell)^\T  B  \sum_{k=1}^K   \alpha_k \alpha_k^\T  B^\T (x-\mu_\ell)-2\log(\pi_\ell) .
    \end{equation} 
    In practice, after computing $\wh B$ from \eqref{new_procedure}, we propose to estimate $C_b$ and $C_w$ by 
    \begin{align*}
        &\wh C_b ~ = \wh B^\T \wh \bM   D_{\wh \pi} \wh \bM^\T \wh B  
                = {1\over n} \wh B^\T \bX^\T P_{\bY} \bX \wh B,\\
         & \wh C_w  =  {1\over n} \wh B^\T \bX^\T(\bI_n - P_{\bY})\bX \wh B,
    \end{align*}
    with $P_{\bY} = \bY(\bY^\T \bY)^{-1}\bY^\T$. Subsequently, the estimated discriminant directions 
     $\wh \alpha_k$, for $1\le k\le K$, are computed by solving 
    \begin{equation}\label{Bayes_FDA_alpha}
        \begin{split} 
         \wh\alpha_k  = \argmax_{\alpha\in\RR^L} & ~ \alpha^\T \wh C_b~ \alpha \quad \text{ s.t.}  ~ \alpha^\T \wh C_w ~ \alpha = 1,  ~ \alpha^\T \wh C_w ~ \wh \alpha_i = 0, ~~ \forall ~ i< k,
        \end{split}
    \end{equation} 
   By similar reasoning for solving the Fisher's discriminant analysis in \eqref{FDA_B}, we propose to solve the optimization problem  \eqref{Bayes_FDA_alpha} via eigen-decomposing  
   $
        [\wh C_w^+ ]^{1\over2}  \wh C_b  [\wh C_w^+]^{1\over2}.
   $
    The eigenvalues of the above matrix can also be used to determine the upper bound for $K$. For instance, $K$ should be chosen no greater than the rank of the above matrix. 
    \begin{remark}[Numerical stability of computing \eqref{Bayes_FDA_alpha}]
        The procedure in \eqref{Bayes_FDA_alpha} computes Fisher's discriminant directions $(\wh \alpha_1, \ldots, \wh \alpha_K)$ within the subspace $\text{span}(\bX \wh B) \subset \text{span}(\bX)$, whose dimension is strictly less than $L$. Consequently, as long as 
        \begin{equation}\label{rank}
            \rank((\bI_n - P_{\bY}) \bX \wh B) = \rank(\bX\wh B),
        \end{equation}
        we do not encounter numerical difficulties when solving \eqref{Bayes_FDA_alpha}, in stark contrast with the algorithmic issue of solving   \eqref{FDA_B} (or a penalized version thereof) in the original $p$-dimensional space. As pointed out by \citet{wu15}, the latter leads to meaningless solutions with the objective function equal to infinity. Finally, we remark that the requirement in \eqref{rank} is easy to satisfy as $\rank(\bX\wh B) \le L$ and $\rank(\bI_n-P_{\bY}) = n-L$ when we observe at least one data point per class. 
    \end{remark}

\subsection{Sparse scenarios}\label{sec_sims_sparse}

We generate the data from the following mean-shift scenario. A similar setting is also used in \citet{Witten2011}. Specifically, for any $\ell \in [L]$ and its corresponding mean vector $\mu_\ell \in \RR^p$, we sample $[\mu_\ell]_j$, for  $5(\ell-1)\leq j \leq 5\ell$, independently from  $N(0, 2^2)$ and set the other entries of $\mu_\ell$ to 0. Regarding the within-class covariance matrix, we set $\Sigma_w=\sigma^2 W$ where  $W$ is generated by independently sampling its diagonal elements from Uniform(1,3), and setting its off-diagonal elements as 
$W_{ij}=\sqrt{W_{ii}W_{jj}} \rho^{|i-j|}$, for all $i \neq j$.
The quantity  $\sigma>0$ controls the overall noise level while the coefficient $\rho \in [0,1]$ controls the within-class correlation among the features. 
We generate the class probabilities $ \pi = (\pi_1,\ldots, \pi_L)^\T$ as follows. For a given $\alpha \ge 0$, we set
    $
        \pi_\ell =  \nu_\ell^\alpha/(\sum_{i=1}^L \nu_i^\alpha)
    $
    for all $\ell\in[L]$. Here  $\nu_1,\ldots, \nu_L$ are independent draws from the Uniform$(0,1)$. The quantity $\alpha$ controls the imbalance of each class size: a larger $\alpha$ corresponds to more imbalanced class sizes. At the other extreme,
     $\alpha = 0$ corresponds to the balanced case $\pi = (1/L)1_L$.
In the following, we examine how the performance of all algorithms depends on different parameters by 
varying one at a time:
   the sample size   $n \in \{100, 300, 500, 700, 900\}$;  the feature dimension $p \in \{100, 300, 500, 700, 900\}$; the number of classes $L \in \{2, 5, 8, 11, 14\}$; the correlation coefficient $\rho \in \{0, 0.2, 0.4 , 0.6, 0.8\}$; the noise parameter $\sigma \in \{0.6, 0.9, 1.2, 1.5, 1.8\}$; the imbalance of the class probabilities  $\alpha \in \{0, 0.4, 0.8, 1.2, 1.6\}$.

For all settings, we fix $n = 300$, $p = 500$, $L=5$, $\rho = 0.6$, $\sigma = 1$ and $\alpha = 0$ when they are not varied. We use $500$ test data points to compute the misclassification errors of each procedure, and for each setting, we consider $50$ repetitions and report  their averaged misclassification errors in Figure \ref{fig:simulation}.

    From Figure \ref{fig:simulation}, it is clear that our procedures consistently outperform other methods across all settings. The performance of our algorithms improves as $n$ increases, $L$ decreases and $\sigma$ decreases. Moreover, neither the imbalance of the class probabilities nor the ambient dimension $p$ seems to affect the misclassification error of our algorithms. These findings are in line with our theory in \cref{sec_theory_lasso}. It is a bit surprising to see that
    the classification error of our benchmark decreases as the correlation coefficient $\rho$ keeps increasing after $0.4$. Nevertheless, 
    both our algorithms and RDA adapt to this situation. 
    Finally, we find that both LDRR and LDRR-F as well as their lasso and elastic net variants have very comparable performance in all scenarios. This is as expected as the data is simulated according to model \eqref{model}, although in practice the Fisher's discriminant analysis (LDRR-F) with the elastic net penalty could have more robust performance, as revealed in our real data analysis in the next section.

\begin{figure}[htbp]
    \centering
    \includegraphics[width=\linewidth]{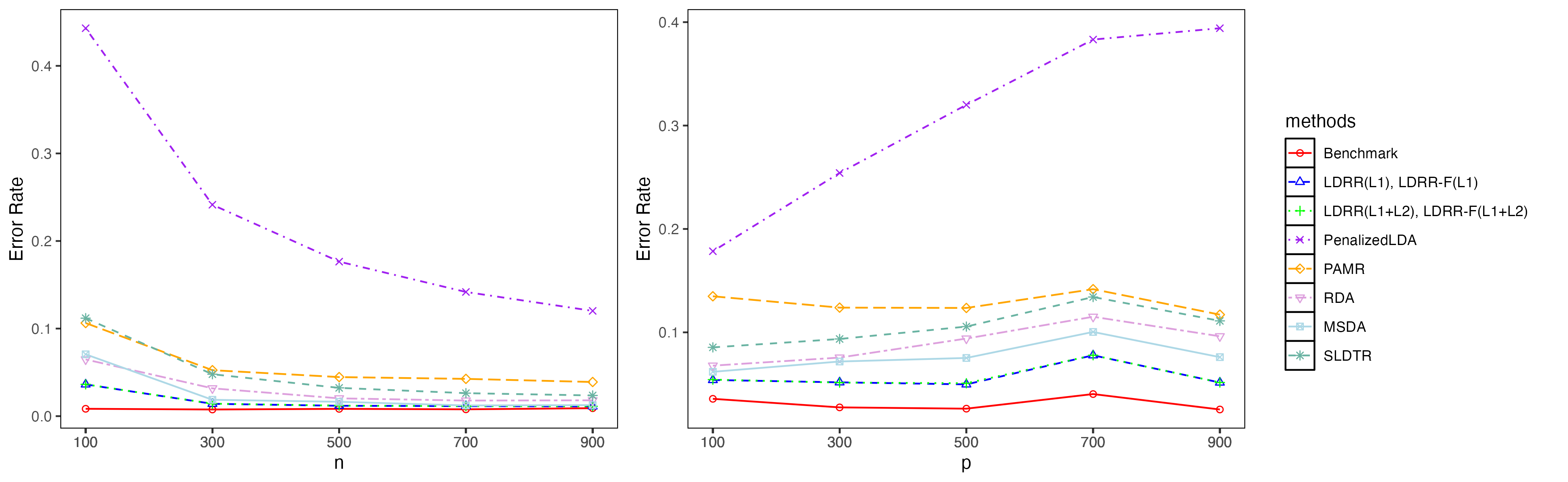}
    \includegraphics[width=\linewidth]{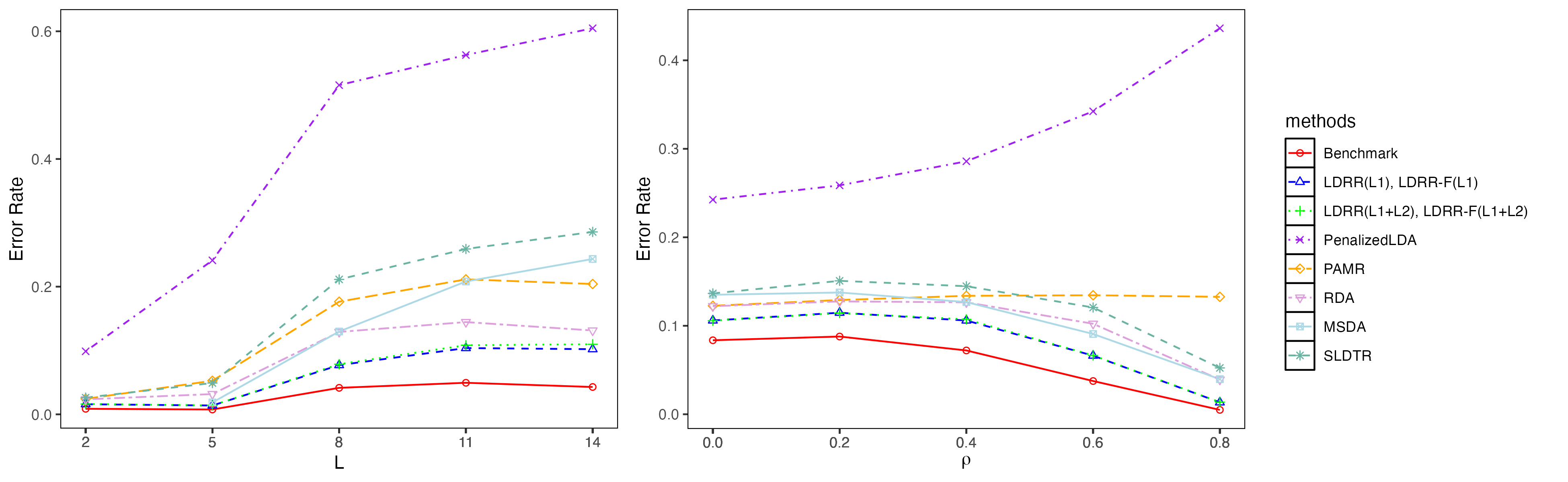}
    \includegraphics[width=\linewidth]{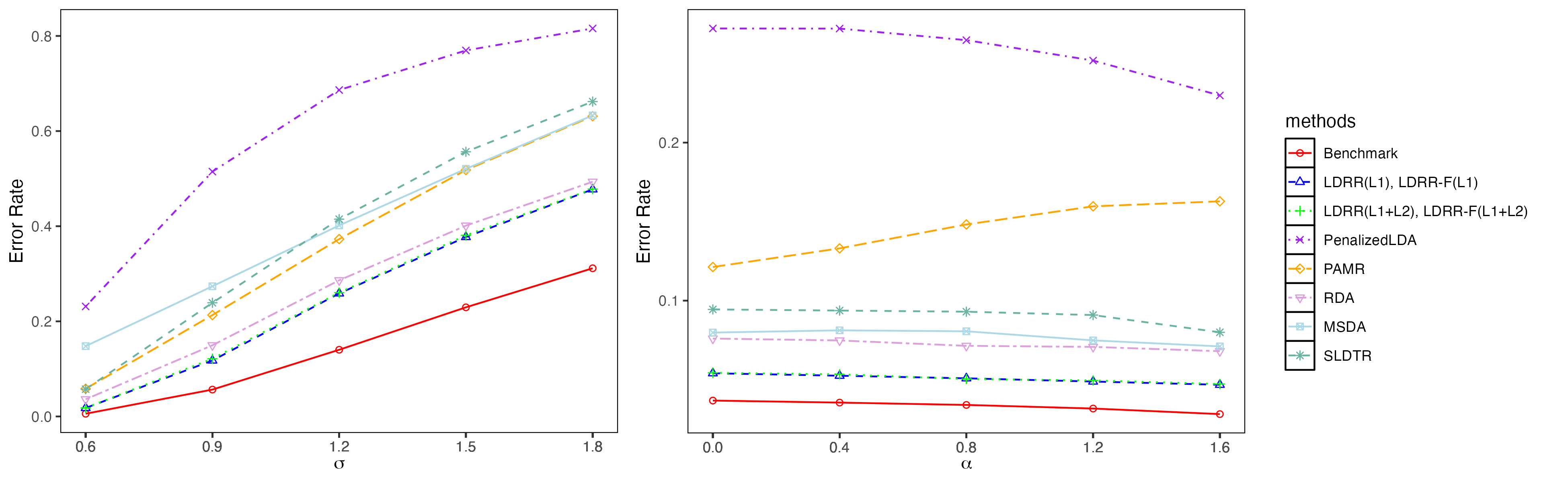}
    \caption{The averaged misclassification errors in sparse scenarios.}
    \label{fig:simulation}
\end{figure}

\subsection{Low-rank scenarios}\label{sec_sims_lowrank}

In this section we evaluate our proposed procedure when the rank of the conditional mean matrix $M$ is low and hence $B$ is also low-rank. 

\subsubsection{Low-rank model (1)}
We first consider a general within-class covariance matrix $\sw$ here, and defer to \cref{app_sec_sim} for the factor model mentioned in \cref{rem_BW} where $\sw$ has approximately low-rank.  
Specifically, let the within-class covariance matrix $[\sw]_{ij}=\rho^{|i-j|}$ with $\rho = 0.6$. To generate $M$, for some scalar $\eta>0$ and some randomly generated orthogonal matrix $A \in \RR^{p\times r}$ with $A^\T A =\bI_r$, we set $M = \eta A \alpha$ with entries of $\alpha \in \RR^{r\times L}$ generated i.i.d. from $N(0,32/r)$. The quantity $\eta$ controls the magnitude of separation between the conditional means. 
 
We examine how the performance of all algorithms depends on various parameters by varying one at a time: the number of classes $L \in \{5, 15, 25, 35, 45\}$;  the sample size $n \in \{300, 500, 700, 900, 1100\}$; the feature dimension $p \in \{100,200, 300,400, 500\}$; the separation scalar $\eta \in \{0.3, 0.6, 0.9, 1.2, 1.5\}$. 
For each setting, we fix $n = 1000, L = 10, r = 3, p = 100$ and $\eta=1$ when they are not varied. We set $n=500$ when we vary $\eta$. 
Figure \ref{fig:simulation2} depicts the misclassification errors of each algorithm averaged over 50 repetitions.

From Figure \ref{fig:simulation2} we see that the proposed approaches of using the reduced-rank penalty generally outperform the other methods. As $L$ gets larger with $r$ fixed, the benefit of using reduced-rank penalty becomes more visible, in line with \cref{cor_rr}. When the dimension ($p$) increases and gets closer to the sample size, the performance of LDRR(RR) deteriorates quickly while LDRR(RR+L2) still performs the best. For relative large $p$, incorporating the ridge penalty is beneficial as it improves the estimation of $B$ when the sample covariance matrix $\wh \Sigma=n^{-1} \bX^\T \bX$ becomes ill-conditioned.

\begin{figure}[htbp]
    \centering
    \includegraphics[width=\linewidth]{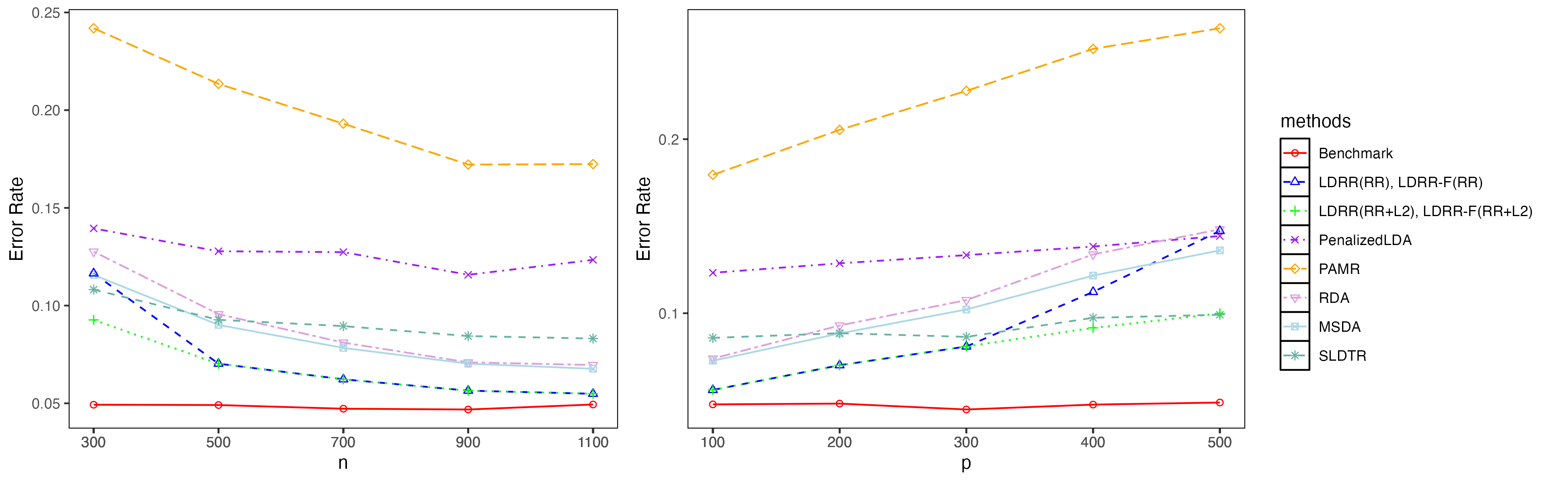}
    \includegraphics[width=\linewidth]{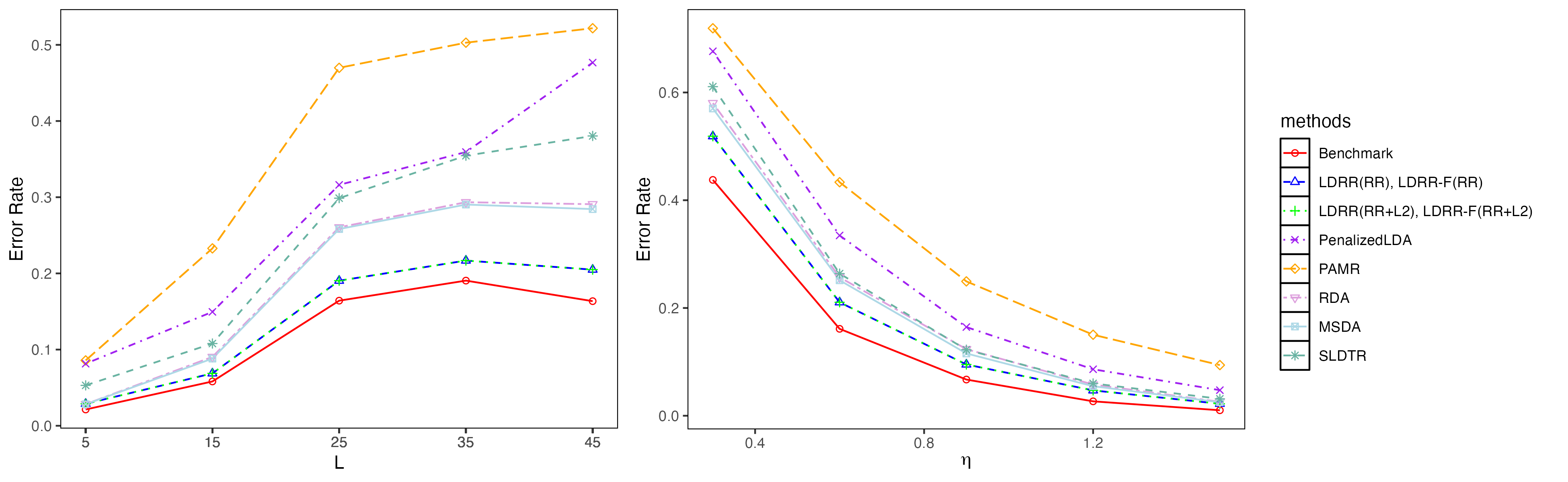}
    \caption{The averaged misclassification errors in low-rank model (1).}
    \label{fig:simulation2}
\end{figure}

 \subsubsection{Low-rank model (2).}\label{app_sec_sim}

        In this setting, we consider the setting in \cref{rem_BW} and also compare with the method (PCLDA)  in \cite{BW2023}. As mentioned in \cref{rem_BW}, PCLDA is a two-step procedure similar to LDRR-F but with the first step done in a unsupervised  way by using PCA. 
        We generate the data according to model \eqref{model} with 
        $M = \eta A \alpha$  as in  \cref{sec_sims_lowrank} and $\sw =  \eta^2 A \Sigma_z A^\T + \Cov(W)$ 
        where $[\Sigma_z]_{ij}=(0.6)^{|i-j|}$ for $i,j\in [r]$ and $\Cov(W)$ is generated the same as $\sw$ in \cref{sec_sims_sparse} with $\rho = 0.2$ and  its diagonal elements from Uniform(0,1). The quantity $\eta$ represents the signal-to-noise ratio (SNR) between the low-dimensional signal $\eta A Z$ and the noise $W$. As pointed out in \cite{BW2023}, classification becomes easier as $\eta$ increases. 
        
        We examine how the performance of all algorithms depends on different parameters by varying the parameters one at a time: the sample size $n \in \{300, 500, 700, 900, 1100\}$; the feature dimension $p \in \{100,200, 300, 400,500\}$; the number of classes $L \in \{5, 15, 25, 35, 45\}$; the scalar $\eta \in \{0.5, 1.0, 2.0, 4.0, 6.0\}$. 
        For all settings, we fix $n = 1000, L = 10, r = 3, p = 100$ and $\eta=2$ when they are not varied. We set $n=500$ when we vary $\eta$. 
        Figure \ref{fig:simulation3} depicts the misclassification errors of each algorithm averaged over 50 repetitions.
        
        From Figure \ref{fig:simulation3} we have similar findings as in the low-rank setting in \cref{sec_sims_lowrank} that the proposed approaches of using the reduced rank penalty outperform other methods when $n$ is much larger than $p$. In the last setting when we vary the SNR $\eta$, it is interesting to notice that PCLDA, the procedure that reduces the dimension by PCA in its first step, has comparable performance to the LDRR-F(RR+L2) when $\eta$ is sufficiently large ($\eta \ge 2$ in these settings). This aligns with the minimax optimality of PCLDA derived in \cite{BW2023} in the regime of moderate / large $\eta$. But when $\eta$ is small such as $\eta < 2$, LDRR-F has benefit by reducing the dimension in a supervised way. 

        \begin{figure}[htbp]
            \centering
            \includegraphics[width=\linewidth]{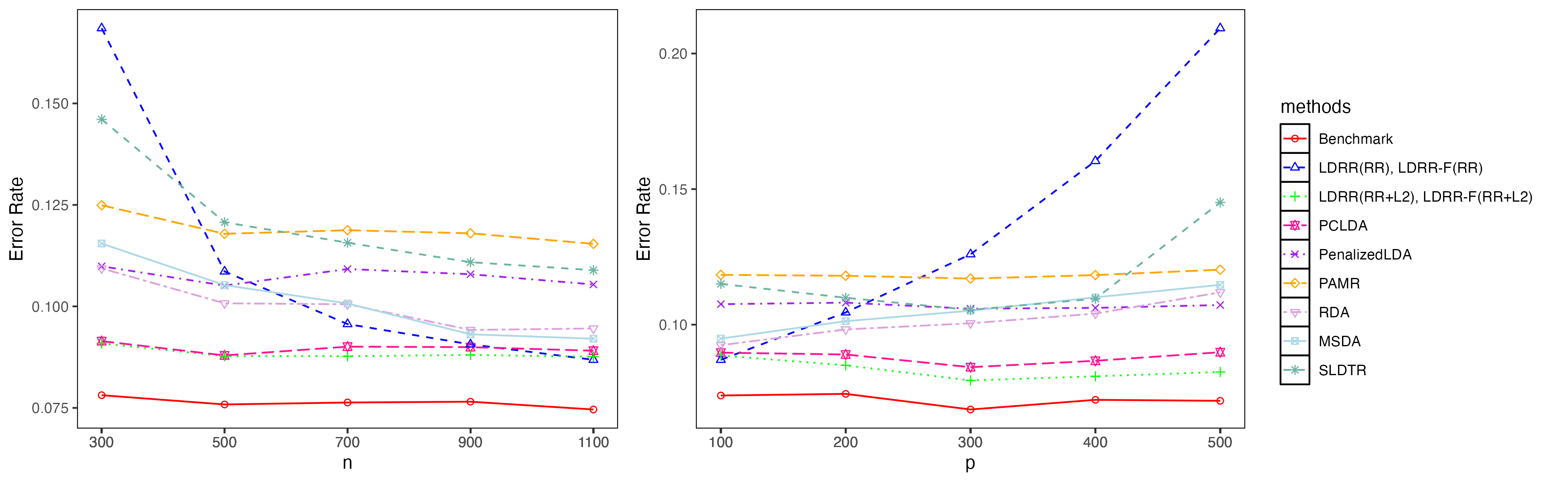}
            \includegraphics[width=\linewidth]{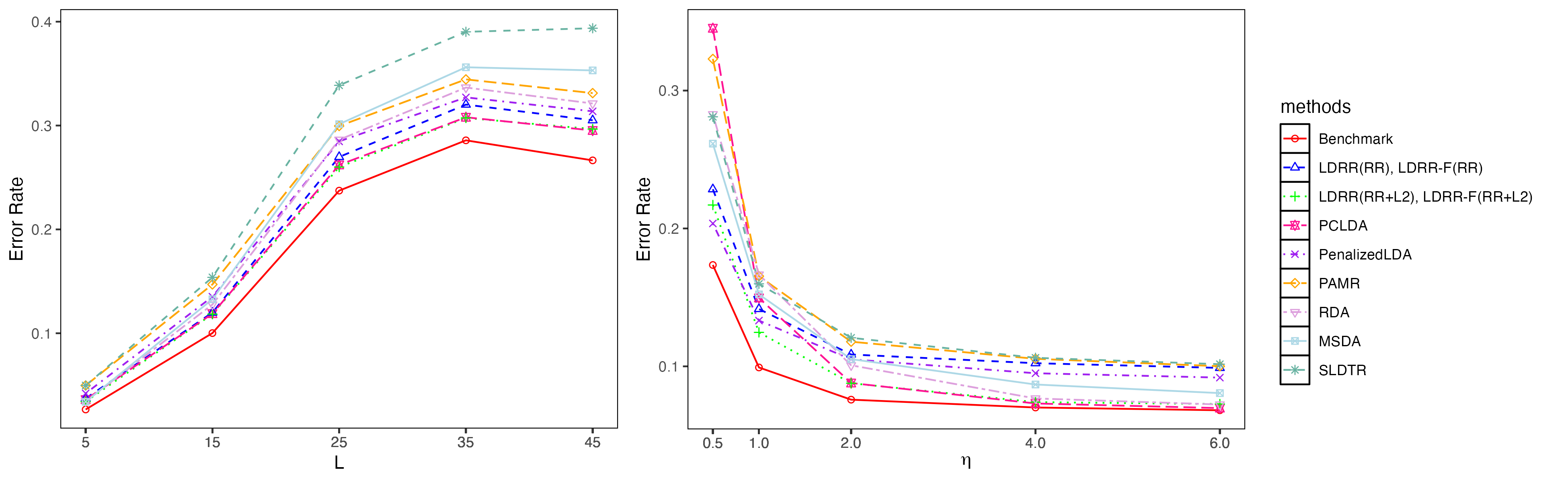}
            \caption{The averaged misclassification errors in low-rank model (2).}
            \label{fig:simulation3}
        \end{figure}

\section{Real data analysis}\label{sec_real_data}


In this section we evaluate the performance of our approach and compare it with other algorithms on three high-dimensional biological data sets and one low-dimensional football dataset.

The Ramaswamy dataset has $n=198$ samples and each sample has features of $p=16,063$ gene expressions and belongs to one of the $L=14$ distinct cancer subtypes \citep{ramaswamy2001multiclass}.  
The second dataset has $n=62$ samples and consists of $L=3$ distinct lymphoma subtypes. These lymphomas are categorized into diffuse large B-cell lymphoma, follicular lymphoma  and chronic lymphocytic leukemia. Each sample has $p=4026$ gene expressions as its features. We used the preprocessed Lymphoma data in \citet{dettling2004bagboosting}.  The Football dataset from kaggle\footnote{\url{https://www.kaggle.com/datasets/vivovinco/20222023-football-player-stats/data}} contains 2022 -- 2023 football player stats per 90 minutes. We excluded categorical data such as the player's nationality, name, team name, and league from our prediction. We used age, year of birth and performance data $(p=118)$ to predict the $(n=2689)$ football player's position $(L=10)$. 
Brain A is another dataset used in \citep{dettling2004bagboosting}. It has sample size $n=42$ with $p=5597$ features and  contains $L=5$ different tumor types. 
For all four data sets, the features are centered to have zero mean.
Each dataset is randomly divided with 75\% of the samples to serve as the training data and the remaining 25\% serves as the test data. This randomization was repeated 50 times and the averaged misclassification errors of each algorithm are reported in Table \ref{tab:realdata}.


\begin{table}[ht]
    \centering
    \caption{The averaged misclassification errors (in percentage) and standard errors}
    \label{tab:realdata}
    \resizebox{\textwidth}{!}{%
        \begin{tabular}{lcccccccc}
            & LDRR-F(L1+L2) & LDRR-F(L12+L2) & LDRR-F(RR+L2) & PenalizedLDA & PAMR & RDA & SLDTR & MSDA \\
            Ramaswamy & 22.3(0.051) & 15.9(0.042) & 22.9(0.044) & 37.4(0.052) & 28.6(0.043) & 17.1(0.038)&17.9(0.045) &30.4(0.067)\\
            Lymphoma & 1.9(0.029) & 3.1(0.051)& 1.9(0.029) & 3.1(0.044) & 2.0(0.029) & 1.9(0.029)& 1.8(0.037)& 5.9(0.050)\\
            Football & 31.5(0.013) & 32.2(0.011) & 30.5(0.008) & 52.6(0.008) & 43.8(0.016) & 31.9(0.014)& 34.11(0.011)& 31.6 (0.011)\\  
            Brain A & 24.3(0.106) & 16.3(0.105) & 13.8(0.082)  & 18.8 (0.107) & 18.8(0.090) & 19.3(0.080)& 12.8(0.062)& 39.0(0.134) \\
        \end{tabular}%
    }
\end{table}

Given the similar performance in \cref{sec_sims} between LDRR and LDRR-F as well as their variants of using the lasso and the elastic-net penalty, we only compare LDRR-F using the elastic net penalty (L1+L2) and the reduced rank plus ridge penalty (RR+L2) with PenalizedLDA, PAMR and RDA in our real data analysis due to its robustness against model misspecification. We also include LDRR-F that uses the group-lasso penalty plus the ridge penalty (L12+L2) for comparison. From Table \ref{tab:realdata}, we observe that our procedure LDRR-F exhibits relatively better performance across all datasets. Choice of the best penalty depends on the data structure. For instance, when the data has the group-sparsity structure such as in Ramaswamy data set, the classifier using the group-lasso penalty has better performance; when the sparsity structure varies across the response levels such as in Lymphoma, the classifier using the elastic-net penalty performs better; when the data meets the low rank structure, the benefit of the classifier using the reduced-rank estimator is apparent (the selected rank is found to be $\wh r = 3$ in the Football dataset with $p=118$ and $L=10$). 
{We remark that in practice one can also use cross-validation to find the most suitable choice of penalty.} Furthermore, in Figure \ref{fig:realdata}, we show that LDRR-F with $K = 2$ in \eqref{Bayes_FDA_B_approx} can be used for visualization in the 2D space of the first two discriminant directions. The algorithm RDA has comparable performance in all settings, but it cannot be used for low-dimensional representation / visualization.

\begin{figure}[!ht]
    \centering
    \includegraphics[width=0.4\linewidth]{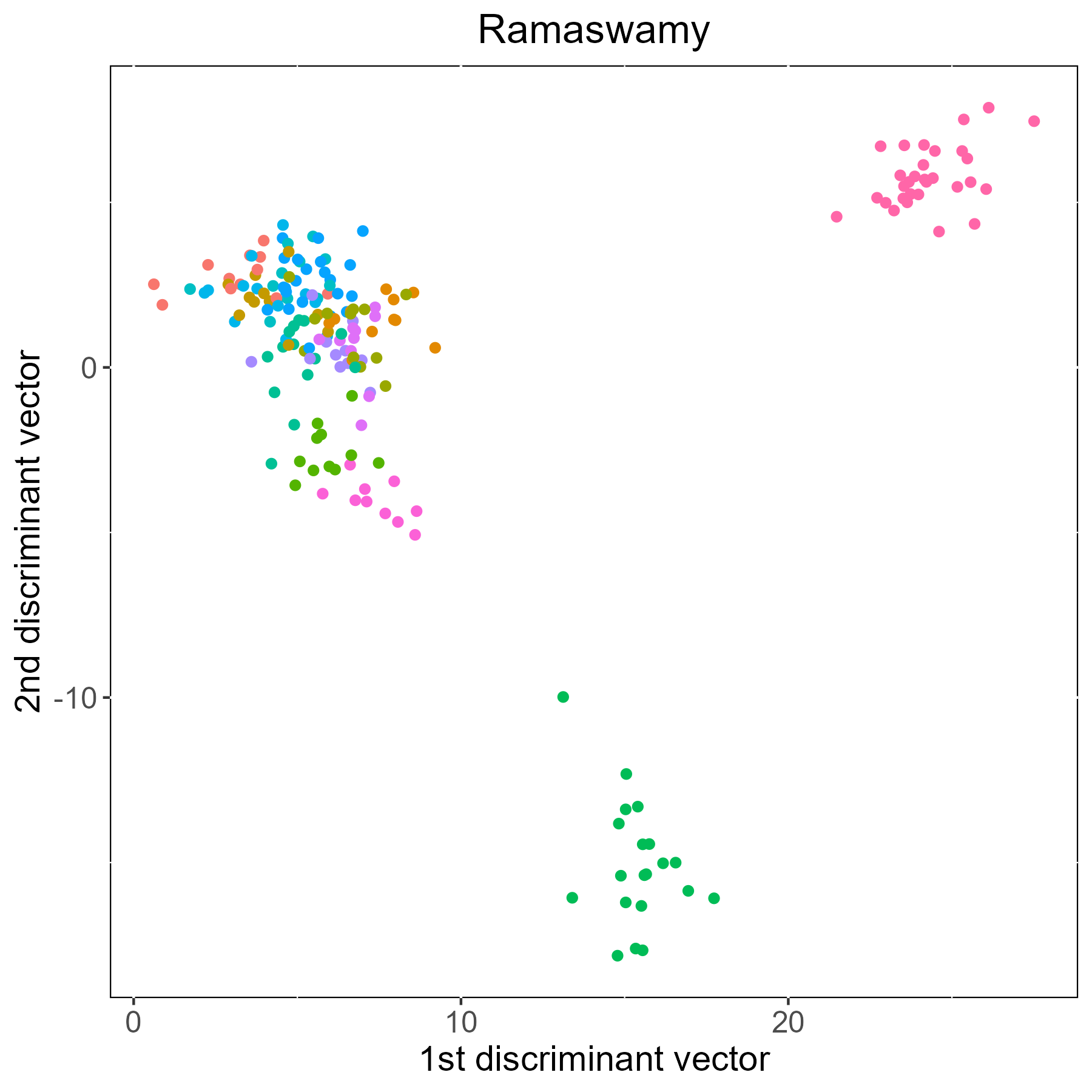} 
    \includegraphics[width = 0.4\linewidth]{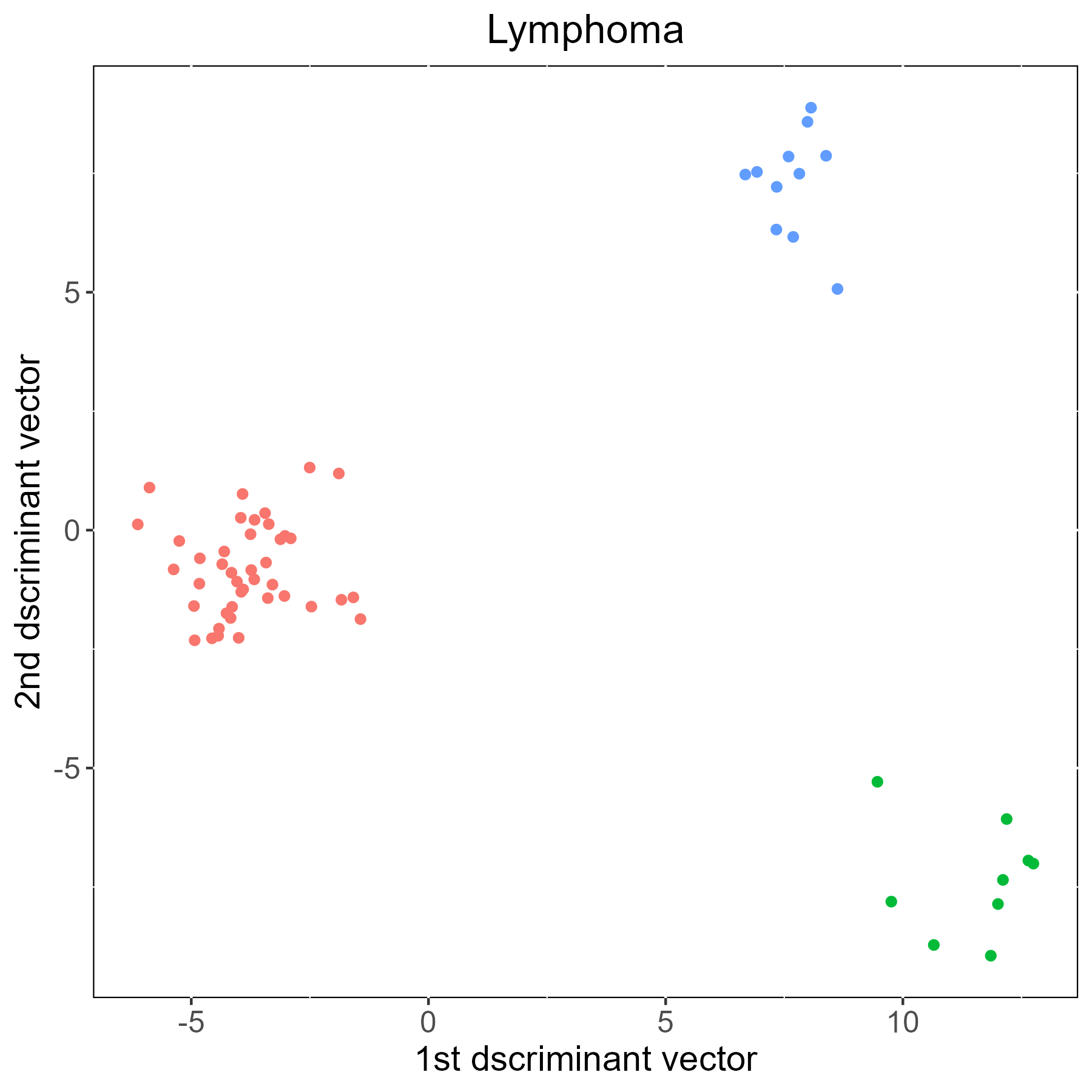}
    \includegraphics[width = 0.4\linewidth]{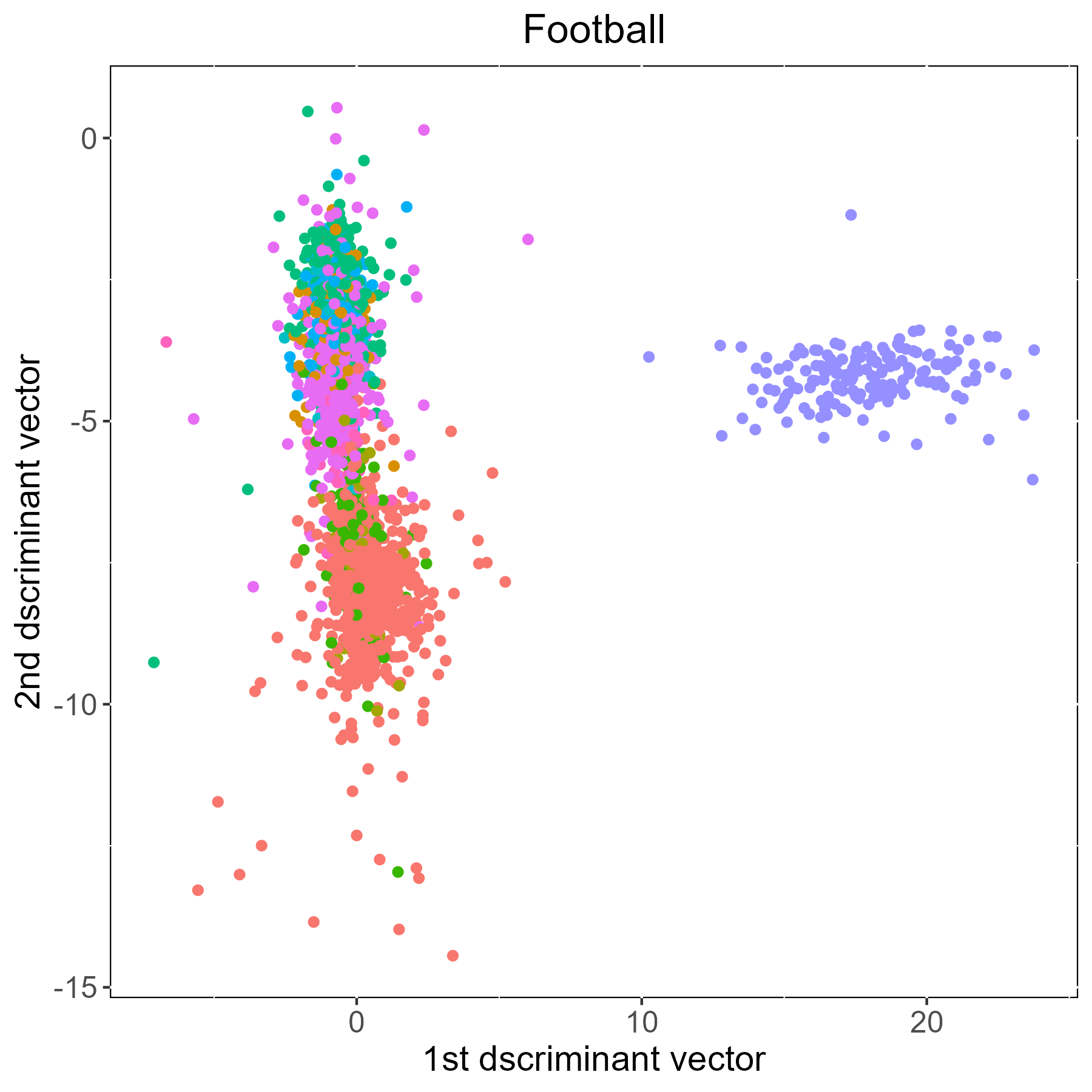}
    \includegraphics[width = 0.4\linewidth]{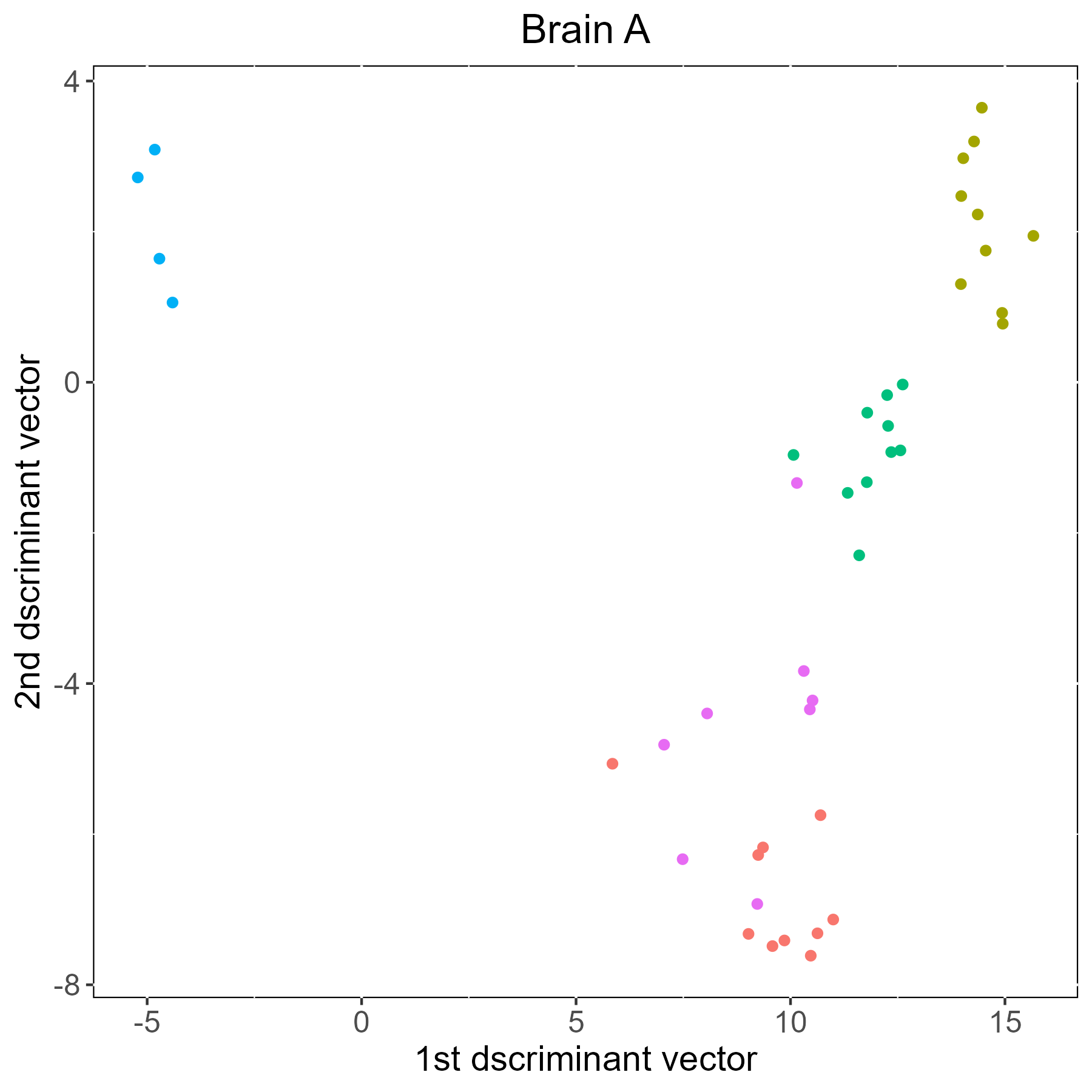}
    \caption{Visualization 
    in the space of the first two discriminant vectors.}
    \label{fig:realdata}
\end{figure}

\section*{Acknowledgements} Wegkamp is supported in part by the National Science Foundation  grant NSF DMS-2210557.

    {\small 
        \setlength{\bibsep}{0.85pt}{
        \bibliographystyle{ims}
        \bibliography{ref}

@article{nibb22,
  title={Multiclass-penalized logistic regression We develop a model for clustering classes in multi-class logistic regression},
  author={Didier Nibbering and Trevor Hastie},
  journal={Comput. Statist. Data Anal.},
  volume={169},
  year={2022},
  publisher={Elsevier}
}

@article{SafoAhn2016,
  title     = {General sparse multi-class linear discriminant analysis},
  author    = {Safo, Sandra E. and Jeongyoun Ahn},
  year      = {2016},
  month     = jul,
  doi       = {10.1016/j.csda.2016.01.011},
  language  = {English (US)},
  volume    = {99},
  pages     = {81--90},
  journal   = {Comput. Stat. Data Anal.},
  issn      = {0167-9473}
}

@article{wang2021penalized,
  title={Penalized interaction estimation for ultrahigh dimensional quadratic regression},
  author={Wang, Cheng and Jiang, Binyan and Zhu, Liping},
  journal={Statistica Sinica},
  volume={31},
  number={3},
  pages={1549--1570},
  year={2021},
  publisher={JSTOR}
}

@article{zeng2024subspace,
  title={Subspace estimation with automatic dimension and variable selection in sufficient dimension reduction},
  author={Zeng, Jing and Mai, Qing and Zhang, Xin},
  journal={Journal of the American Statistical Association},
  volume={119},
  number={545},
  pages={343--355},
  year={2024},
  publisher={Taylor \& Francis}
}

@book{BG2011,
    author = {B\"uhlmann, Peter and Van de Geer, Sara },
    title = {Statistics for High-Dimensional Data},
    publisher = {Springer} ,
    year = {2011}
}

@book{G2021,
    author = {Christophe Giraud} ,
    title ={Introduction to High-Dimensional Statistics
}, series={Monographs on Statistics and Applied Probability}, number= {139},
    publisher ={ CRC Press, Taylor \& Francis Group} ,
    year = {2021}
}

@article{hastie1995penalized,
  title={Penalized discriminant analysis},
  author={Hastie, Trevor and Buja, Andreas and Tibshirani, Robert},
  journal={The Annals of Statistics},
  volume={23},
  number={1},
  pages={73--102},
  year={1995},
  publisher={Institute of Mathematical Statistics}
}

@article{friedman1989regularized,
  title={Regularized discriminant analysis},
  author={Friedman, Jerome H},
  journal={J. Amer. Statist. Assoc},
  volume={84},
  number={405},
  pages={165--175},
  year={1989},
  publisher={Taylor \& Francis}
}

@article{campbell1980shrunken,
  title={Shrunken estimators in discriminant and canonical variate analysis},
  author={Campbell, Norm A},
  journal={Journal of the Royal Statistical Society: Series C (Applied Statistics)},
  volume={29},
  number={1},
  pages={5--14},
  year={1980},
  publisher={Wiley Online Library}
}

@article{cai2011direct,
  title={A direct estimation approach to sparse linear discriminant analysis},
  author={Cai, Tony and Liu, Weidong},
  journal={J.  Amer. Statist. Assoc.},
  volume={106},
  number={496},
  pages={1566--1577},
  year={2011},
  publisher={Taylor \& Francis}
}

@InProceedings{wu15,
  title = 	 {Understanding and Evaluating Sparse Linear Discriminant Analysis},
  author = 	 {Wu, Yi and Wipf, David and Yun, Jeong-Min},
  booktitle = 	 {Proceedings of the Eighteenth International Conference on Artificial Intelligence and Statistics},
  pages = 	 {1070--1078},
  year = 	 {2015},
  editor = 	 {Lebanon, Guy and Vishwanathan, S. V. N.},
  volume = 	 {38},
  series = 	 {Proceedings of Machine Learning Research},
  address = 	 {San Diego, California, USA},
  month = 	 {09--12 May},
  publisher =    {PMLR}
}

@article{AP19,
author = {Abramovich, Felix and Pensky, Marianna},
year = {2019},
month = {08},
pages = {104536},
title = {Classification with many classes: Challenges and pluses},
volume = {174},
journal = {Journal of Multivariate Analysis},
doi = {10.1016/j.jmva.2019.104536}
}

@article{mai2019multiclass,
  title={Multiclass sparse discriminant analysis},
  author={Mai, Qing and Yang, Yi and Zou, Hui},
  journal={Statistica Sinica},
  volume={29},
  number={1},
  pages={97--111},
  year={2019},
  publisher={JSTOR}
}

@article{tsybakov2004optimal,
  title={Optimal aggregation of classifiers in statistical learning},
  author={Tsybakov, Alexander B},
  journal={The Annals of Statistics},
  volume={32},
  number={1},
  pages={135--166},
  year={2004},
  publisher={Institute of Mathematical Statistics}
}

@article{Shao2011,
author = {Jun Shao and Yazhen Wang and Xinwei Deng and Sijian Wang},
title = {{Sparse linear discriminant analysis by thresholding for high dimensional data}},
volume = {39},
journal = {The Annals of Statistics},
number = {2},
publisher = {Institute of Mathematical Statistics},
pages = {1241--1265},
year = {2011},
doi = {10.1214/10-AOS870}
}

@article{Tibshirani2002,
	author = {Tibshirani, Robert and Hastie, Trevor and Narasimhan, Balasubramanian and Chu, Gilbert},
	title = {Diagnosis of multiple cancer types by shrunken centroids of gene expression},
	volume = {99},
	number = {10},
	pages = {6567--6572},
	year = {2002},
	doi = {10.1073/pnas.082099299},
	publisher = {National Academy of Sciences},
	issn = {0027-8424},
	eprint = {https://www.pnas.org/content/99/10/6567.full.pdf},
	journal = {Proceedings of the National Academy of Sciences}
}

@article{FanFan2008,
author = {Jianqing Fan and Yingying Fan},
title = {{High-dimensional classification using features annealed independence rules}},
volume = {36},
journal = {The Annals of Statistics},
number = {6},
publisher = {Institute of Mathematical Statistics},
pages = {2605--2637},
year = {2008},
doi = {10.1214/07-AOS504}
}

@article{clemmensen2011sparse,
  title={Sparse discriminant analysis},
  author={Clemmensen, Line and Hastie, Trevor and Witten, Daniela and Ersb{\o}ll, Bjarne},
  journal={Technometrics},
  volume={53},
  number={4},
  pages={406--413},
  year={2011},
  publisher={Taylor \& Francis}
}

@article{Witten2011,
author = {Witten, Daniela M. and Tibshirani, Robert},
title = {Penalized classification using Fisher's linear discriminant},
journal = {Journal of the Royal Statistical Society: Series B (Statistical Methodology)},
volume = {73},
number = {5},
pages = {753--772},
doi = {https://doi.org/10.1111/j.1467-9868.2011.00783.x},
eprint = {https://rss.onlinelibrary.wiley.com/doi/pdf/10.1111/j.1467-9868.2011.00783.x},
year = {2011}
}

@article{caizhang2019,
	author = {Cai, Tony and Zhang, Linjun},
	title = {High dimensional linear discriminant analysis: optimality, adaptive algorithm and missing data},
	journal = {Journal of the Royal Statistical Society: Series B (Statistical Methodology)},
	volume = {81},
	number = {4},
	pages = {675--705},
	doi = {10.1111/rssb.12326},
	eprint = {https://rss.onlinelibrary.wiley.com/doi/pdf/10.1111/rssb.12326},
	year = {2019}
}

@article{mai2012,
	Author = {Mai, Qing and Zou, Hui and Yuan, Ming},
	Journal = {Biometrika},
	Month = {12},
	Number = {1},
	Pages = {29--42},
	Title = {{A direct approach to sparse discriminant analysis in ultra-high dimensions}},
	Volume = {99},
	Year = {2012}}

@book{Izenman-book,
	author = {Izenman, Alan Julian},
	title = {Modern Multivariate Statistical Techniques: Regression, Classification, and Manifold Learning},
	Publisher = {Series: Springer Texts in Statistics},
	pages = {760},
	year={2008}
}

@article{bunea2011,
	author = "Bunea, Florentina and She, Yiyuan and Wegkamp, Marten H.",
	doi = "10.1214/11-AOS876",
	fjournal = "The Annals of Statistics",
	journal = "Ann. Statist.",
	month = "04",
	number = "2",
	pages = "1282--1309",
	publisher = "The Institute of Mathematical Statistics",
	title = "Optimal selection of reduced rank estimators of high-dimensional matrices",
	volume = "39",
	year = "2011"
}

@article{giraud2011,
	author = "Giraud, Christophe",
	doi = "10.1214/11-EJS625",
	fjournal = "Electronic Journal of Statistics",
	journal = "Electron. J. Statist.",
	pages = "775--799",
	publisher = "The Institute of Mathematical Statistics and the Bernoulli Society",
	title = "Low rank multivariate regression",
	volume = "5",
	year = "2011"
}

@inproceedings{RudelsonZhou,
  title={Reconstruction from anisotropic random measurements},
  author={Rudelson, Mark and Zhou, Shuheng},
  booktitle={Conference on Learning Theory},
  pages={10--1},
  year={2012},
  organization={JMLR Workshop and Conference Proceedings}
}

@article{BW2023,
  title={Optimal discriminant analysis in high-dimensional latent factor models},
  author={Bing, Xin and Wegkamp, Marten},
  journal={The Annals of Statistics},
  volume={51},
  number={3},
  pages={1232--1257},
  year={2023},
  publisher={Institute of Mathematical Statistics}
}

@article{Irina,
    author = {Irina Gaynanova},
    title = {{Prediction and estimation consistency of sparse multi-class penalized optimal scoring}},
    volume = {26},
    journal = {Bernoulli},
    number = {1},
    publisher = {Bernoulli Society for Mathematical Statistics and Probability},
    pages = {286--322},
    year = {2020},
    doi = {10.3150/19-BEJ1126}
}

@article{rank19,
    author = "Bing, Xin and Wegkamp, Marten H.",
    doi = "10.1214/18-AOS1774",
    fjournal = "Annals of Statistics",
    journal = "Ann. Statist.",
    month = "12",
    number = "6",
    pages = "3157--3184",
    publisher = "The Institute of Mathematical Statistics",
    title = "Adaptive estimation of the rank of the coefficient matrix in high-dimensional multivariate response regression models",
    volume = "47",
    year = "2019"
}

@article{ramaswamy2001multiclass,
  title={Multiclass cancer diagnosis using tumor gene expression signatures},
  author={Ramaswamy, Sridhar and Tamayo, Pablo and Rifkin, Ryan and Mukherjee, Sayan and Yeang, Chen-Hsiang and Angelo, Michael and Ladd, Christine and Reich, Michael and Latulippe, Eva and Mesirov, Jill P and others},
  journal={Proceedings of the National Academy of Sciences},
  volume={98},
  number={26},
  pages={15149--15154},
  year={2001},
  publisher={National Acad Sciences}
}

@article{dettling2004bagboosting,
  title={BagBoosting for tumor classification with gene expression data},
  author={Dettling, Marcel},
  journal={Bioinformatics},
  volume={20},
  number={18},
  pages={3583--3593},
  year={2004},
  publisher={Oxford University Press}
}

@article{guo2007regularized,
  title={Regularized linear discriminant analysis and its application in microarrays},
  author={Guo, Yaqian and Hastie, Trevor and Tibshirani, Robert},
  journal={Biostatistics},
  volume={8},
  number={1},
  pages={86--100},
  year={2007},
  publisher={Oxford University Press}
}

@article{de1976additive,
  title={Additive structure in qualitative data: An alternating least squares method with optimal scaling features},
  author={De Leeuw, Jan and Young, Forrest W and Takane, Yoshio},
  journal={Psychometrika},
  volume={41},
  number={4},
  pages={471--503},
  year={1976},
  publisher={Springer}
}

@article{young1978principal,
  title={The principal components of mixed measurement level multivariate data: An alternating least squares method with optimal scaling features},
  author={Young, Forrest W and Takane, Yoshio and de Leeuw, Jan},
  journal={Psychometrika},
  volume={43},
  pages={279--281},
  year={1978},
  publisher={Springer}
}

@article{hastie1994flexible,
  title={Flexible discriminant analysis by optimal scoring},
  author={Hastie, Trevor and Tibshirani, Robert and Buja, Andreas},
  journal={Journal of the American statistical association},
  volume={89},
  number={428},
  pages={1255--1270},
  year={1994},
  publisher={Taylor \& Francis}
}

@inproceedings{ye2007least,
  title={Least squares linear discriminant analysis},
  author={Ye, Jieping},
  booktitle={Proceedings of the 24th international conference on Machine learning},
  pages={1087--1093},
  year={2007}
}

@inproceedings{lee2015equivalence,
  title={On the equivalence of linear discriminant analysis and least squares},
  author={Lee, Kibok and Kim, Junmo},
  booktitle={Proceedings of the AAAI Conference on Artificial Intelligence},
  volume={29},
  number={1},
  year={2015}
}

@article{nie2022equivalence,
  title={On the Equivalence of Linear Discriminant Analysis and Least Squares Regression},
  author={Nie, Feiping and Chen, Hong and Xiang, Shiming and Zhang, Changshui and Yan, Shuicheng and Li, Xuelong},
  journal={IEEE Transactions on Neural Networks and Learning Systems},
  year={2022},
  publisher={IEEE}
}

@inproceedings{chen2022distributed,
  title={Distributed Sparse Multicategory Discriminant Analysis},
  author={Chen, Hengchao and Sun, Qiang},
  booktitle={International Conference on Artificial Intelligence and Statistics},
  pages={604--624},
  year={2022},
  organization={PMLR}
}

@article{qiao2009sparse,
  title={Sparse linear discriminant analysis with applications to high dimensional low sample size data.},
  author={Qiao, Zhihua and Zhou, Lan and Huang, Jianhua Z},
  journal={IAENG International Journal of Applied Mathematics},
  volume={39},
  number={1},
  year={2009}
}

@article{mukherjee2011reduced,
  title={Reduced rank ridge regression and its kernel extensions},
  author={Mukherjee, Ashin and Zhu, Ji},
  journal={Statistical analysis and data mining: the ASA data science journal},
  volume={4},
  number={6},
  pages={612--622},
  year={2011},
  publisher={Wiley Online Library}
}

@article{Fan2012,
    author = {Fan, Jianqing and Feng, Yang and Tong, Xin},
    title = {A road to classification in high dimensional space: the regularized optimal affine discriminant},
    journal = {Journal of the Royal Statistical Society: Series B (Statistical Methodology)},
    volume = {74},
    number = {4},
    pages = {745-771},
    doi = {https://doi.org/10.1111/j.1467-9868.2012.01029.x}, 
    eprint = {https://rss.onlinelibrary.wiley.com/doi/pdf/10.1111/j.1467-9868.2012.01029.x},
    year = {2012}
}

@article{yuan2006model,
  title={Model selection and estimation in regression with grouped variables},
  author={Yuan, Ming and Lin, Yi},
  journal={Journal of the Royal Statistical Society Series B: Statistical Methodology},
  volume={68},
  number={1},
  pages={49--67},
  year={2006},
  publisher={Oxford University Press}
}

@article{zou2005regularization,
  title={Regularization and variable selection via the elastic net},
  author={Zou, Hui and Hastie, Trevor},
  journal={Journal of the Royal Statistical Society Series B: Statistical Methodology},
  volume={67},
  number={2},
  pages={301--320},
  year={2005},
  publisher={Oxford University Press}
}

@article{Tibshirani1996,
 ISSN = {00359246}, 
 author = {Robert Tibshirani},
 journal = {Journal of the Royal Statistical Society. Series B (Methodological)},
 number = {1},
 pages = {267--288},
 publisher = {[Royal Statistical Society, Wiley]},
 title = {Regression Shrinkage and Selection via the Lasso}, 
 volume = {58},
 year = {1996}
}

@book{seber2009multivariate,
  title={Multivariate observations},
  author={Seber, George AF},
  year={2009},
  publisher={John Wiley \& Sons}
}

@article{chen2013reduced,
  title={Reduced rank regression via adaptive nuclear norm penalization},
  author={Chen, Kun and Dong, Hongbo and Chan, Kung-Sik},
  journal={Biometrika},
  volume={100},
  number={4},
  pages={901--920},
  year={2013},
  publisher={Oxford University Press}
}

@article{Koltchinskii2011,
    author = {Vladimir Koltchinskii and Karim Lounici and Alexandre B. Tsybakov},
    title = {{Nuclear-norm penalization and optimal rates for noisy low-rank matrix completion}},
    volume = {39},
    journal = {The Annals of Statistics},
    number = {5},
    publisher = {Institute of Mathematical Statistics},
    pages = {2302 -- 2329},
    year = {2011},
    doi = {10.1214/11-AOS894}
}

@article{levy2023generalization,
  title={Generalization error bounds for multiclass sparse linear classifiers},
  author={Levy, Tomer and Abramovich, Felix},
  journal={Journal of Machine Learning Research},
  volume={24},
  number={151},
  pages={1--35},
  year={2023}
}

@article{abramovich2021multiclass,
  title={Multiclass classification by sparse multinomial logistic regression},
  author={Abramovich, Felix and Grinshtein, Vadim and Levy, Tomer},
  journal={IEEE Transactions on Information Theory},
  volume={67},
  number={7},
  pages={4637--4646},
  year={2021},
  publisher={IEEE}
}

@article{lei2019data,
  title={Data-dependent generalization bounds for multi-class classification},
  author={Lei, Yunwen and Dogan, {\"U}r{\"u}n and Zhou, Ding-Xuan and Kloft, Marius},
  journal={IEEE Transactions on Information Theory},
  volume={65},
  number={5},
  pages={2995--3021},
  year={2019},
  publisher={IEEE}
}

@article{ahn2021trace,
  title={Trace ratio optimization for high-dimensional multi-class discrimination},
  author={Ahn, Jeongyoun and Chung, Hee Cheol and Jeon, Yongho},
  journal={Journal of Computational and Graphical Statistics},
  volume={30},
  number={1},
  pages={192--203},
  year={2021},
  publisher={Taylor \& Francis}
}

@article{jung2019penalized,
  title={Penalized orthogonal iteration for sparse estimation of generalized eigenvalue problem},
  author={Jung, Sungkyu and Ahn, Jeongyoun and Jeon, Yongho},
  journal={Journal of Computational and Graphical Statistics},
  volume={28},
  number={3},
  pages={710--721},
  year={2019},
  publisher={Taylor \& Francis}
}

@article{gaynanova2016simultaneous,
  title={Simultaneous sparse estimation of canonical vectors in the p >> N setting},
  author={Gaynanova, Irina and Booth, James G and Wells, Martin T},
  journal={Journal of the American Statistical Association},
  volume={111},
  number={514},
  pages={696--706},
  year={2016},
  publisher={Taylor \& Francis}
}

@article{Supp_LDRR,
    author = {Bing, Xin and Li, Bingqing and Wegkamp, Marten},
    title = {Supplement to "Linear discriminant regularized regression"},
    journal = {},
    year = {2025} 
}

@article{IzenmanRR,
    title = {Reduced-rank regression for the multivariate linear model},
    journal = {Journal of Multivariate Analysis},
    volume = {5},
    number = {2},
    pages = {248-264},
    year = {1975},
    issn = {0047-259X},
    doi = {https://doi.org/10.1016/0047-259X(75)90042-1}, 
    author = {Alan Julian Izenman}
}

@article{VINCENT2014771,
title = {Sparse group lasso and high dimensional multinomial classification},
journal = {Computational Statistics \& Data Analysis},
volume = {71},
pages = {771-786},
year = {2014},
issn = {0167-9473},
doi = {https://doi.org/10.1016/j.csda.2013.06.004},
url = {https://www.sciencedirect.com/science/article/pii/S0167947313002168},
author = {Martin Vincent and Niels Richard Hansen},
}
        }
    }

 \newpage

\appendix

    \cref{app_sec_proof} contains the main proofs. Technical lemmas  are stated in \cref{app_sec_tech_lem}.

 \section{Main proofs}\label{app_sec_proof}

	Throughout the proofs, we will use the following notation. Write
	\[
		\Delta_\i = \max_{\ell \in [L]} \be_\ell^\T\bM^\T \sw^{-1}\bM\be_\ell,\qquad \Delta_\op = \|\bM^\T \sw^{-1}\bM\|_\op.
	\]  
    For future reference, we note that 
    (under $\EE(X) = 0_p$)
    \[
        {1\over 4}\max_{k,\ell\in [L]} (\mu_k-\mu_\ell)^\T \sw^{-1}(\mu_k-\mu_\ell) \le \Delta_\i \le \max_{k,\ell\in [L]} (\mu_k-\mu_\ell)^\T \sw^{-1}(\mu_k-\mu_\ell)
    \] 
     so that $\Dt_\i \asymp \Delta$ under \cref{ass_cond_separation}.
	Further, we write 
	\[
		\pi_{\min} = \min_{\ell \in [L]} \pi_\ell,\qquad \pi_{\max} = \max_{\ell \in [L]} \pi_\ell.
	\]
    The following event will be frequently used in our proofs. 
   Define the event 
    \begin{equation}\label{def_event_pi}
        \cE_\pi = \left\{
        \max_{k\in [L]} \left|\wh \pi_k - \pi_k\right| \le C\sqrt{\log(n) \over nL}
        \right\}
    \end{equation}
    for some absolute constant $C>0$.
    \cref{lem_pi} in \cref{app_sec_aux_lem} ensures that   $\PP(\cE_\pi) \ge 1-2n^{-2}$ under \cref{ass_pi}.
    Also note that, on  $\cE_\pi$ and under the condition $L\log(n) \le  n$, 
    \begin{equation}\label{bd_log_pi}
        \max_{\ell \in [L]}\left|\log\left({\pi_\ell \over \wh \pi_\ell}\right)\right| \le \max_{\ell \in [L]}{|\wh \pi_\ell - \pi_\ell| \over \pi_\ell} \lesssim \sqrt{ L \log(n)\over n}
    \end{equation}
    as well as $\wh \pi_\ell \asymp \pi_\ell $ for all $\ell \in [L]$.


    	\subsection{Proof of \cref{prop_g}}\label{app_sec_proof_prop_g}
	Recall that, for any $x\in \RR^p$,
	\[
	\wh g(x) = \argmin_{\ell \in [L]} \wh G_\ell (x),\qquad g^*(x) = \argmin_{\ell \in [L]} G_\ell (x).
	\]
	For any $t_1,\ldots, t_L\ge 0$,  define the event 
	\begin{equation}\label{def_event_multi}
		\cE_t = \bigcap_{\ell \in [L]}\left\{
		\left|
		\wh G_\ell(X) - G_\ell(X) 
		\right| \le t_\ell
		\right\}.
	\end{equation}
    In the proof of \cref{prop_g}, we use the notation 
    $\{Y = k\} := \{Y = \be_k\}$.
	
	\begin{proof}
        We condition on $\bD$ throughout the proof. The law of  $\EE$ and $\PP$ is with respect to the randomness of $(X,Y)$, which is independent of $\bD$.
		By definition, the excess risk $\cR(\wh g)$ equals to
		\begin{align*}
        &\sum_{k \in [L]} \pi_k \Bigl\{
			\EE\left[
			\1\{\wh g(X) \ne k\} \mid Y = k
			\right] - 
			\EE\left[
			\1\{g^*(X) \ne k\} \mid Y = k
			\right]
			\Bigr\}\\
			&=  \sum_{k \in [L]} \pi_k 
			\EE\left[
			\1\{\wh g(X) \ne k, g^*(X) = k\} \mid Y = k
			\right]-\sum_{k \in [L]} \pi_k 
			\EE\left[
			\1\{\wh g(X) = k, g^*(X) \ne k\} \mid Y = k
			\right]\\
			&= \sum_{\substack{k,\ell \in [L]\\ k\ne \ell}} \pi_k 
			\EE\left[
			\1\{\wh g(X) = \ell, g^*(X) = k\} \mid Y = k
			\right]- \sum_{\substack{k,\ell \in [L]\\ k\ne \ell}}  \pi_k 
			\EE\left[
			\1\{\wh g(X) = k, g^*(X) = \ell\} \mid Y = k
			\right]\\
			&=  \sum_{\substack{k,\ell \in [L]\\ k\ne \ell}} 
			\Bigl\{\pi_k \EE\left[
			\1\{\wh g(X) = \ell, g^*(X) = k\} \mid Y = k
			\right]-
			\pi_\ell \EE\left[
			\1\{\wh g(X) = \ell, g^*(X) = k\} \mid Y = \ell\right]
			\Bigr\}.
		\end{align*}
		Write  $f_k(x)$ as the p.d.f. of $X=x \mid Y = k$ for each $k\in [L]$.  We find that
		\begin{align*}
			\cR(\wh g)
			&= \sum_{\substack{k,\ell \in [L]\\ k\ne \ell}} 
			\int_{\wh g = \ell,  g^* = k}\left(
			\pi_k f_k(x) - \pi_\ell f_\ell(x)\right) d x\\
			&=\sum_{\substack{k,\ell \in [L]\\ k\ne \ell}} 
			\int_{\wh g = \ell,  g^* = k} \pi_k f_k(x) \left(1 - \exp\left\{ G^{(\ell | k)}(x)\right\}\right) dx
		\end{align*}
		where, for all $k, \ell \in [L]$,
		\[
		G^{(\ell | k)}(x) :=  \left(x - {\mu_\ell + \mu_k \over 2}\right)^\T \sw^{-1}(\mu_\ell - \mu_k) + \log{\pi_\ell \over \pi_k} = {1\over 2}\left(G_k(x) - G_\ell (x)\right).
		\]
        Fix any $t_1,\ldots, t_L \ge 0$.
		Observe that the event $\{\wh g(X) = \ell,  g^*(X) = k\} \cap \cE_t$ implies
		\[
    		0 > 2G^{(\ell | k)}(X)  \overset{\cE_t}{\ge} \wh G_k(X) - \wh G_\ell(X) - (t_k + t_\ell) \ge  -(t_k + t_\ell).
		\] 
		By the basic inequality $1 + z \le  \exp(z)$ for all $z\in \RR$, we obtain that,
        \begin{align*}
        \cR(\wh g) &\le  \sum_{\substack{k,\ell \in [L]\\ k\ne \ell}} \pi_k  \PP(\cE_t^c   \cap \{\wh g(X) = \ell\}  \mid Y = k)\\
        & +  {1\over 2} \sum_{\substack{k,\ell \in [L]\\ k\ne \ell}}    \pi_k  (t_k + t_\ell) \EE\left[ \1\left \{-  t_k -  t_\ell \le 2G^{(\ell | k)}(X) \le 0, \wh g(X) = \ell \mid Y = k\right \} \right].
		\end{align*}
		For the first term, it is bounded from above by
		\begin{align}\label{bd_term_two}
			   \sum_{k=1}^L \pi_k \PP(\cE_t^c \mid Y = k) =  \PP(\cE_t^c ). 
		\end{align}
		Regarding the second term, it is bounded from above by
		\begin{align}\label{bd_term_one}
			 & {1\over 2}\sum_{k=1}^L \pi_k \sum_{\ell \in [L]\setminus \{k\}}   (t_k + t_\ell) \PP\left\{- t_k - t_\ell  \le 2 G^{(\ell | k)}(X) \le 0 \mid Y = k\right\}.
		\end{align}
		Note that $G^{(\ell | k)}(X) \mid Y = k$ is normally distributed
		$$
		  N \left(
		-{1\over 2}(\mu_k-\mu_\ell)^\T \sw^{-1} (\mu_k-\mu_\ell) + \log\left(\pi_\ell / \pi_k\right), ~ (\mu_k-\mu_\ell)^\T \sw^{-1} (\mu_k-\mu_\ell)
		\right).
		$$ 
		By the mean-value theorem and \cref{ass_cond_separation}, the quantity in \eqref{bd_term_one} is no greater than
		\begin{align*}
		  {1\over 4}\sum_{k=1}^L \pi_k \sum_{\ell \in [L]\setminus \{k\}}   {(t_k + t_\ell)^2 \over  \sqrt{c\Dt}} &\le {1\over 2 \sqrt{c\Dt}} \sum_{k=1}^L \pi_k \left[ (L-1)t_k^2 +  \sum_{\ell \in [L]\setminus \{k\}} t_\ell^2\right]  = {1\over 2  \sqrt{c\Dt}}\sum_{k=1}^L (L\pi_k + 1)  t_k^2 .
		\end{align*} 
        Together with \eqref{bd_term_two} and $\pi_k \le C/L$, the proof of \eqref{fast_rate} is complete.

       The bound in \eqref{slow_rate} follows by noting that 
        \begin{align*}
           &{1\over 2} \sum_{\substack{k,\ell \in [L]\\ k\ne \ell}}    \pi_k  (t_k + t_\ell) \EE\left[ \1\left \{-  t_k -  t_\ell \le 2G^{(\ell | k)}(X) \le 0, \wh g(X) = \ell \mid Y = k\right \} \right]\\
           &\le {1\over 2} \sum_{\substack{k,\ell \in [L]\\ k\ne \ell}}    \pi_k  (t_k + t_\ell) \PP\left\{\wh g(X) = \ell \mid Y = k\right \} \\
           &\le \max_{\ell \in [L]}t_\ell.
        \end{align*} 
        The proof is complete.
	\end{proof}

    \subsection{Proof of \cref{thm_G}}\label{app_sec_proof_thm_G}
	\begin{proof}
		Pick any $\ell \in [L]$. Recall that
		\begin{equation}\label{decomp_G_diff}
		|\wh G_\ell(X)  -  G_\ell(X) |  \le 
		|(\wh \mu_\ell - \mu_\ell)^\T B^*_\ell| +| ( \wh\mu_\ell- 2X)^\T (\wh B^*_\ell - B^*_\ell)| + 2|\log\left({\pi_\ell / \wh \pi_\ell}\right)|.
		\end{equation}
		Throughout the  proof, we work on $\cE_\pi$ in \eqref{def_event_pi}  which implies \eqref{bd_log_pi}, an upper bound for the last term on the right of \eqref{decomp_G_diff}. Recall that $\PP\{\cE_\pi\}\ge1-2n^{-2}$ and note 
  that we consider both the randomness of $(X,Y)$ and that of $\bD$ in this proof. 
		
		For the first term in \eqref{decomp_G_diff},  we invoke \cref{lem_mu} with $k = \ell$ and $v =B^*_\ell$, use the fact that  $B_\ell^{*\T}\sw  B_\ell^*  = \mu_\ell^\T \sw^{-1} \mu_\ell \le \Dt_\i$ 
        and take a union bound over $\ell\in [L]$ to conclude that the event 
		\begin{equation}\label{bd_mu_diff_B_star} 
        \cE_\mu:= \left\{ \left|(\wh \mu_\ell - \mu_\ell)^\T B^*_\ell\right|  \le  \sqrt{3\Delta_{\ell\ell}} \sqrt{L \log(n)\over n},~ \text{for all }\ell \in [L]\right\}
		\end{equation} 
        has probability at least $1- 2L n^{-3} \ge 1-n^{-2}$.
		Regarding the second term in \eqref{decomp_G_diff}
		\begin{align*}
			|\wh \mu_\ell^\T  (\wh B^*_\ell - B^*_\ell)| &\le |\mu_\ell^\T  (\wh B^*_\ell - B^*_\ell)| + 
			|(\wh \mu_\ell - \mu_\ell)^\T  (\wh B^*_\ell - B^*_\ell)|,
		\end{align*}
		the Cauchy-Schwarz inequality 
		gives 
		\begin{equation}\label{bd_mu_diff_B}
			|\mu_\ell^\T  (\wh B^*_\ell - B^*_\ell)| \le \| \sw^{-1/2}\mu_\ell\|_2 \|\sw^{1/2}(\wh B^*_\ell - B^*_\ell)\|_2 =   \|\sw^{1/2}(\wh B^*_\ell - B^*_\ell)\|_2 \sqrt{\mu_\ell^\T \sw^{-1} \mu_\ell}
		\end{equation}
        Finally, conditioning on $Y = \be_k$, the Gaussian tail of $X \mid Y, \bD$ gives that, for all $t\ge 0$,
		\[
		\PP\left\{	|X^\T(\wh B^*_\ell - B^*_\ell)| \ge |\mu_k^\T (\wh B^*_\ell - B^*_\ell)| + t \|\sw^{1/2}(\wh B^*_\ell - B^*_\ell)\|_2 
		~ \Big| ~  Y = \be_k, \bD \right\} \le 2e^{-t^2/2}.
		\]
		Choosing $t = 2\sqrt{  \log (n)}$,  taking the union bound over $k \in [L]$ and unconditioning yield that for all $\ell \in [L]$,
		\begin{align}\label{bd_X_B_star}\nonumber
    		|X^\T(\wh B^*_\ell - B^*_\ell)| &\le \sum_{k=1}^L \pi_k |\mu_k^\T (\wh B^*_\ell - B^*_\ell)|  +  2\|\sw^{1/2}(\wh B^*_\ell - B^*_\ell)\|_2 \sqrt{ \log(n)}\\
		&\le   \|\sw^{1/2}(\wh B^*_\ell - B^*_\ell)\|_2 \sqrt{ \Delta_\i+4\log(n)} 
		\end{align}
		holds with probability at least $1-2L n^{-2}$. The last step uses  \eqref{bd_mu_diff_B}.
		Collecting \cref{ass_cond_separation}, \eqref{bd_log_pi}, \eqref{bd_mu_diff_B_star}, \eqref{bd_mu_diff_B} and \eqref{bd_X_B_star} and $ \Dt_\i\lesssim\Delta$ completes the proof.
	\end{proof}

    \subsection{Proof of \cref{cor_g}}\label{app_sec_proof_cor_g}
 	\begin{proof}
        For any $\ell \in [L]$, let 
        \[
            {t_\ell \over C} =  \sqrt{\Dt}\sqrt{L\log(n) \over n} + \|\sw^{1/2}(\wh B^*_\ell -B^*_\ell)\|_2 \sqrt{\Dt + \log(n)} +    |(\wh \mu_\ell - \mu_\ell)^\T (\wh B^*_\ell- B^*_\ell)|.
        \]
        Recall from the proof of \cref{thm_G} that $1-\PP(\cE_\mu \cap \cE_\pi ) \le 2n^{-2}+n^{-2} =3n^{-2}$ and 
        \[
            \sum_{\ell=1}^L \PP\left\{ 
             \left\{\left|\wh G_\ell (X) - G_\ell (X)\right|  \ge  t_\ell\right\}   \cap \cE_\mu \cap \cE_\pi \right\}   \le {1\over n}.
        \]
        By the fact $\cR(\wh g) \le 1$ almost surely and invoking \cref{prop_g}, we find that 
        \begin{align*}
        \EE_{\bD}[\cR(\wh g)] & \le \EE_{\bD}\left[
            \min\left\{ C\sum_{\ell=1}^L t_\ell^2 + L \sum_{\ell=1}^L \PP\left\{
            |\wh G_\ell(X) - G_\ell(X)| \ge  t_\ell \mid \bD\right\}, ~ 1  \right\} 1\{\cE_\mu\cap \cE_\pi\}\right]\\
        &\quad + 1- \PP_{\bD}\{\cE_\mu\cap \cE_\pi\}\\
            &\lesssim  \EE_{\bD}\left[
            \min\left\{\Dt {L^2\log(n)\over n} + \cQ(\wh B^*) + {L\over n}, ~ 1  \right\} 1\{\cE_\mu\cap \cE_\pi\} \right] +  {1\over n}\\
            &\lesssim \EE_{\bD}\left[
            \min\left\{\Dt {L^2\log(n)\over n} + \cQ(\wh B^*), ~ 1  \right\}  \right],
        \end{align*}
        completing the proof.
 	\end{proof}

    \subsection{Proof of \cref{thm_Q}}\label{app_proof_thm_Q}
        The proof of \cref{thm_Q} uses the following lemma on $\|\wh H - H\|_\op$. The proof of \cref{lem_H_hat} appears at the end of this section.
	\begin{lemma}\label{lem_H_hat}
		Under the conditions of \cref{thm_Q}, with probability at least $1-n^{-1}$,
		\[
		\|\wh H - H\|_\op  \lesssim   \sqrt{\log(n) \over nL} +  {\omega_2 \over \sqrt L}.
		\]
	\end{lemma}

		\begin{proof}[Proof of \cref{thm_Q}]
			Define 
			\begin{equation}\label{rate_delta}
				\delta  = C  \sqrt{\log(n) \over n L} +  C  {\omega_2 \over \sqrt{L}}
			\end{equation}
            for some constant $C>0$. We work on the event  
			\[
    			\cE_H:= \left\{ 
    				\|\wh H - H\|_\op \le \delta 
    			\right\}.
			\] 
			\cref{lem_H_hat} ensures that $\cE_H$ holds with probability at least  $1-n^{-1}$. Moreover, on the event $\cE_H$, we find that 
		\begin{align*}
			\lambda_K(\wh H) &\ge \lambda_K(H) - \|\wh H - H\|_\op &&\text{by Weyl's inequality}\\ 
			&= \left[\lambda_1(\Omega)\right]^{-1} - \|\wh H - H\|_\op &&\text{by }\Omega = H^{-1}\\
			& \ge  {1\over \max_k (1/\pi_k)+  \Delta_\op} - \delta &&\text{by \eqref{def_Omega} and }\cE_H\\
			&\ge {c\over 4 L\Delta_\i} &&\text{by $\omega_2\sqrt{L} \le c$, \cref{ass_pi} and }\Delta_\op \le L \Delta_\i.
		\end{align*}
		Hence, on the event $\cE_H$, $\wh H$ is invertible. 
		We proceed to prove 
		\begin{align}\label{bd_sw_B_diff}
			&   \|\sw^{1/2}(\wh B^* -B^*)\|_F^2   ~ \lesssim   ~   \left(\omega_2^2 +   \delta^2 L\Delta_\i \right) L^2 \Delta_\i^2,\\ 	\label{bd_mu_diff_B_diff}
			&  \sum_{\ell=1}^L  \left[(\wh \mu_\ell - \mu_\ell)^\T (\wh B^*_\ell- B^*_\ell)\right]^2   \lesssim    \left(
			\omega_1^2 {\log(n\vee p) \over n} + 		\delta^2 L\Delta_\i {\log(n)\over n} 
			\right)L^3 \Delta_\i^2
		\end{align}
		which, in conjunction with \eqref{rate_delta}, $\Delta\asymp 1$, $L\log(n) \le cn$ and $\omega_2\sqrt{L} \le c$, yield the claim.  
 
        To prove \eqref{bd_sw_B_diff}, 
		we write $\wh\Omega = \wh H^{-1}$ for simplicity. 
		By definition, 
		\begin{align}\label{eq_B_star_diff}\nonumber
			\wh B^*  - B^*  &= (\wh B - B)\wh \Omega  + B(\wh \Omega-\Omega)    \\\nonumber
			&= (\wh B - B)\wh \Omega +B^* H(\wh \Omega-\Omega) &&\text{by \cref{lem_key}}\\\nonumber
			&=  (\wh B - B)\wh \Omega +B^*(H\wh \Omega-\bI_L) &&\text{by } H\Omega = \bI_L\\
			&= (\wh B - B)\wh \Omega +B^*(H-\wh H)\wh \Omega &&\text{by } \wh H\wh \Omega = \bI_L.
		\end{align}
        It then follows that 
		\begin{align*}
			&\|\sw^{1/2}(\wh B^*   - B^*)\|_F\\ 
			& \le  
			\|\sw^{1/2}  (\wh B - B)\wh \Omega  \|_F+ 
			\|\sw^{1/2}B^*(H-\wh H)\wh \Omega \|_F
			\\
			& \le \|\sw^{1/2}(\wh B - B)\|_F \| \wh \Omega\|_\op +  \|\sw^{1/2}B^*\|_F ~ \|(H-\wh H)\wh \Omega\|_\op\\
			&\le \|\sw^{1/2}(\wh B - B)\|_F \| \wh \Omega\|_\op + \sqrt{L\Delta_\i} \|\wh H- H\|_{\op}\|\wh \Omega\|_\op.
		\end{align*}
        In the last step, we use $\|\sw^{1/2}B^*\|_F^2 = \tr(B^{*\T} \sw B^*) = \tr(\bM^\T \sw^{-1} \bM) \le L\Delta_\i$.
		We bound  $\|\wh \Omega\|_\op$ as follows: 
		\begin{align*}
			\|\wh \Omega\|_\op   &\le \|\Omega\|_\op  +  \|\wh \Omega - \Omega\|_\op \\
			& = \|\Omega\|_\op   +    \| \Omega(H -  \wh  H)\wh \Omega\|_\op \\
			& \le \|\Omega\|_\op   + \|\Omega\|_\op \|H - \wh H\|_\op \|\wh \Omega\|_\op. 
		\end{align*}
		Since the event $\cE_H$ and $\omega_2\sqrt{L} \le c$ ensure that 
		\[
		\|\Omega\|_\op \|H - \wh H\|_\op \le \left(
		\max_{k\in [K]} {1\over \pi_k} + \Delta_\op 
		\right) \delta  \lesssim \delta  L \Delta_\i \le {1\over 2},
		\]
		we conclude 
		\begin{equation}\label{bd_Omega_op}
			\|\wh \Omega\|_\op \le 2\|\Omega\|_\op \lesssim L\Delta_\i.
		\end{equation} 
		It now follows from \cref{ass_B} that
		\begin{align}\label{bd_sw_diff_B_star}\nonumber
			\|\sw^{1/2}(\wh B^*  - B^*)\|_F &\lesssim \left(\|\sw^{1/2}(\wh B - B)\|_F  +   \delta  \sqrt{L \Delta_\i}\right)  L\Delta_\i\\
			&\le  \left(\omega_2  +   \delta  \sqrt{L\Delta_\i}\right) L\Delta_\i,
		\end{align}
		proving \eqref{bd_sw_B_diff}.\\
			 
        We now  prove \eqref{bd_mu_diff_B_diff}. Using the identity \eqref{eq_B_star_diff}, we have,  on the event $\cE_H$,
			\begin{align}\label{decomp_mu_diff_B_star_diff}\nonumber
				 &\sum_{\ell = 1}^L \left[(\wh \mu_\ell - \mu_\ell)^\T  (\wh B^*_\ell - B^*_\ell)\right]^2\\
                &\le  2\sum_{\ell = 1}^L \left[(\wh \mu_\ell - \mu_\ell)^\T  (\wh B  - B)\wh \Omega_\ell \right]^2 + 2\sum_{\ell = 1}^L \left[(\wh \mu_\ell - \mu_\ell)^\T  B^*(H - \wh H) \wh \Omega_\ell \right]^2.
			\end{align}
        For the first term on the right,  we find 
        \begin{align*}
            \sum_{\ell = 1}^L \left[(\wh \mu_\ell - \mu_\ell)^\T  (\wh B  - B)\wh \Omega_\ell \right]^2 &= \sum_{\ell = 1}^L  (\wh \mu_\ell - \mu_\ell)^\T  (\wh B  - B)\wh \Omega_\ell \wh \Omega_\ell^\T (\wh B  - B)^\T   (\wh \mu_\ell - \mu_\ell) \\
            &\le \max_{\ell \in [L]} \|\wh \mu_\ell - \mu_\ell\|_\i^2 \sum_{\ell = 1}^L  \|(\wh B  - B)\wh \Omega_\ell\wh \Omega_\ell^\T (\wh B  - B)^\T \|_1.
        \end{align*}
        On the one hand, on the event   $\cE_\pi$, \cref{lem_mu}  with $k=\ell$ and $v = \be_j$, in conjunction with union bounds over $j\in[p]$ and $\ell \in [L]$,  yields that
        \begin{equation}\label{bd_mu_hat_diff}
            \max_{\ell \in [L]} \|\wh \mu_\ell - \mu_\ell\|_\i^2 \lesssim {\|\sw\|_\i}{L\log(n\vee p) \over n}
        \end{equation}
        with probability at least $1-n^{-2}$. On the other hand, we notice that 
        \begin{align*}
             \sum_{\ell = 1}^L  \|(\wh B  - B)\wh \Omega_\ell\wh \Omega_\ell^\T (\wh B  - B)^\T \|_1  &= \sum_{\ell = 1}^L  \sum_{i,j\in [p]} \left|\be_i^\T (\wh B  - B)\wh \Omega_\ell\wh \Omega_\ell^\T (\wh B  - B)^\T \be_j \right|\\
            &\le \sum_{i,j\in [p]}  \left(\sum_{\ell = 1}^L 
                    (\be_i^\T (\wh B  - B)\wh \Omega_\ell)^2\right)^{1/2}  \left(\sum_{\ell = 1}^L 
                (\be_j^\T (\wh B  - B)\wh \Omega_\ell)^2\right)^{1/2}  \\
            &= \sum_{i,j\in [p]}  
                     \|\wh \Omega(\wh B-B)^\T \be_i \|_2 \|\wh \Omega(\wh B-B)^\T \be_j\|_2 \\
            &\le \|\wh \Omega\|_\op^2  \sum_{i,j\in [p]}  \|(\wh B-B)^\T \be_i\|_2  \|(\wh B-B)^\T \be_j\|_2\\
            &\le \|\wh \Omega\|_\op^2 \|\wh B-B\|_{1,2}^2.
        \end{align*}
        By invoking  \cref{ass_B} and applying inequalities \eqref{bd_Omega_op} and \eqref{bd_mu_hat_diff}, we obtain
        \begin{align}\label{bd_mu_diff_B_diff_1}
            \sum_{\ell = 1}^L \left[(\wh \mu_\ell - \mu_\ell)^\T  (\wh B^*_\ell - B^*_\ell)\right]^2 \lesssim \omega_1^2{\|\sw\|_\i}{\log(n\vee p) \over n} L^3 \Delta_\i^2
        \end{align}
        with probability at least $1-n^{-2}$. Regarding the second term on the right in \eqref{decomp_mu_diff_B_star_diff}, we have 
        \begin{align*}
            \sum_{\ell = 1}^L \left[(\wh \mu_\ell - \mu_\ell)^\T  B^*(H - \wh H) \wh \Omega_\ell \right]^2 & \le   \sum_{\ell = 1}^L \|(\wh \mu_\ell - \mu_\ell)^\T  B^*\|_2^2 \|(H - \wh H)\wh \Omega_\ell\|_2^2 \\
            &\le \max_{\ell \in [L]} \|(\wh \mu_\ell - \mu_\ell)^\T  B^*\|_2^2 \|H - \wh H\|_\op^2 \sum_{\ell = 1}^L\|\wh \Omega_\ell\|_2^2\\
				&\le 	 \delta^2 \max_{\ell \in [L]} \|(\wh \mu_\ell - \mu_\ell)^\T  B^*\|_2^2   \|\wh \Omega\|_F^2 \\
				&\lesssim \delta^2 \max_{\ell \in [L]} \|(\wh \mu_\ell - \mu_\ell)^\T  B^*\|_2^2  ~ L \|\wh \Omega\|_\op^2 
        \end{align*} 
        on the event $\cE_H$. 
        Further invoking  \eqref{bd_mu_diff_B_star} and \eqref{bd_Omega_op} yields that 
			\begin{align}\label{bd_mu_diff_B_diff_2}
				\sum_{\ell = 1}^L \left[(\wh \mu_\ell - \mu_\ell)^\T  B^*(H - \wh H) \wh \Omega_\ell \right]^2
				&~ \lesssim     		\delta^2  { \log(n)\over n} L^4 \Delta_\i^3
			\end{align}
			with probability at least $1- n^{-2}$. Recall that $\Dt_\i \le C\Dt$ from \cref{ass_cond_separation}.
         Combining \eqref{bd_mu_diff_B_diff_1} with \eqref{bd_mu_diff_B_diff_2} proves \eqref{bd_mu_diff_B_diff}, and hence completes the proof.  
	\end{proof}


      \begin{proof}[Proof of \cref{lem_H_hat}] 
		Recall that 
		\begin{equation}\label{eq_H_diff}
			\begin{split}
				\wh H - H &= D_{\wh\pi} - D_\pi  - \left(
				\wh B^\T \wh \Sigma \wh B - B^\T \Sigma B
				\right)\\
				&= D_{\wh\pi} - D_\pi - B^\T (\wh \Sigma - \Sigma) B - (\wh B-B)^\T \wh \Sigma (\wh B-B)\\
				&\qquad  - (\wh B-B)^\T \wh \Sigma B - B^\T \wh \Sigma (\wh B-B).
			\end{split}
		\end{equation}
		By triangle inequality, we have
        \begin{align*}
			\|\wh H - H\|_\op  &\le  \|D_{\wh\pi} - D_\pi\|_\op   +  \|(\wh B-B)^\T \wh \Sigma (\wh B-B)\|_\op + 2\|(\wh B-B)^\T \wh \Sigma B\|_\op \\ 
			&\quad  + 
			\|B^\T (\wh \Sigma - \Sigma )B\|_\op.
		\end{align*}
        We proceed to bound each term on the right-hand-side separately. On the event $\cE_\pi$, we first find that 
		\begin{align}\label{bd_D_pi}
			\|D_{\wh\pi} - D_\pi\|_\op   \le \max_{\ell \in [L]} |\wh \pi_\ell - \pi_\ell| \lesssim \sqrt{\log(n) \over nL}.
		\end{align}
        Regarding the second term, \cref{ass_B} ensures that
        \begin{equation}\label{bd_quad_B_diff}
            \|(\wh B-B)^\T \wh \Sigma (\wh B-B)\|_\op \le \|\wh \Sigma^{1/2}(\wh B-B)\|_F^2 \le \omega_2^2.
        \end{equation}
        For the last two terms, invoking \cref{lem_B_Sigma_diff_B} ensures with probability at least $1-n^{-2}$,
		\begin{align}\label{bd_quad_B_Sigma_diff} 
            \|B^\T (\wh \Sigma - \Sigma )B\|_\op \lesssim {\Delta_\op   \over L+\Delta_\op}\sqrt{\Delta_\i\log(n) \over nL} +  {\Delta_\op \over L +\Delta_\op}{\Delta_\i\log(n) \over n}.
		\end{align}
		It also  follows that
		\begin{align*}
			 \|(\wh B-B)^\T \wh \Sigma B\|_\op &\le  \|\wh \Sigma^{1/2}(\wh B - B)\|_\op \|\wh\Sigma^{1/2}B\|_\op \\
			&\le  
			\|\wh \Sigma^{1/2}(\wh B - B)\|_\op \left(
			\|B^\T \Sigma B\|_\op + \|B^\T (\wh \Sigma - \Sigma )B\|_\op 
			\right)^{1/2}\\
			&\lesssim  \omega_2 \sqrt{\Delta_\op \over L(L+\Delta_\op)}
		\end{align*}
		where the last step uses $$\|B^\T \Sigma B\|_\op \overset{\eqref{eq_B}}{=} \|D_\pi \bM^\T B\|_\op \lesssim {\Delta_\op \over L(L+\Delta_\op)}.
		$$
		deduced from \cref{lem_B_fact}. Combining the above bounds  with \eqref{bd_D_pi}, \eqref{bd_quad_B_diff} and \eqref{bd_quad_B_Sigma_diff} gives 
		\begin{align*}
    		\|\wh H - H\|_\op  &\lesssim    \sqrt{\log(n) \over nL} + {\Delta_\op   \over L+\Delta_\op}\sqrt{\Delta_\i\log(n) \over nL}\\
      &\qquad +  {\Delta_\op \over 
    			L +\Delta_\op}{\Delta_\i\log(n) \over n}+  \omega_2 \sqrt{\Delta_\op \over L(L+\Delta_\op)}+\omega_2^2 
		\end{align*} 
		with probability at least $1-n^{-1}$. The result follows by using $\Delta_\i \asymp 1$ and $\omega_2\sqrt{L} \le c$ to collect terms.
	\end{proof}



     \subsection{Proof of \cref{thm_B_lasso}: the lasso estimator of $B$}\label{app_sec_proof_thm_B_lasso}

    The proof of \cref{thm_B_lasso} uses the following  Restricted Eigenvalue condition (RE) on the within-class covariance matrix $\sw$.
    \begin{definition}\label{ass_RE}
        For any integer $1\le s\le p$, define
        \begin{equation*}
            \kappa_s := \min_{S\subseteq [p], |S| \le s} \min_{v\in \cC(S, 3)}{v^\T \sw v \over v^\T v}
        \end{equation*}
        where $\cC(S, 3) := \{u\in \RR^p\setminus \{0\}: \|u_{S^c}\|_1 \le 3\|u_S\|_1\}$.
    \end{definition}
    We prove \cref{thm_B_lasso} under the following weaker condition than \cref{ass_sw}.
    \begin{ass}\label{ass_sw_weaker}
        For some $1\le s\le p$, there exists some absolute constants $0<c \le C<\i$ such that 
        $c<\kappa_s \le  \|\sw\|_\i \le C$.
    \end{ass}

    \begin{proof}[Proof of \cref{thm_B_lasso}]
        We first prove that 
        \begin{equation}\label{results_lasso}
            \begin{split}
                 &\max_{\ell \in [L]} {1\over \sqrt n}\|\bX (\wh B_\ell - B_\ell)\|_2 ~ \lesssim ~ \kappa_s^{-1/2}
            \lambda\sqrt{s}\\
            &\max_{\ell \in [L]} \|\Sigma^{1/2} (\wh B_\ell - B_\ell)\|_2 ~ \lesssim ~ \kappa_s^{-1/2}
            \lambda\sqrt{s}\\
            &\max_{\ell \in [L]}  \|\wh B_\ell - B_\ell\|_2 ~ \lesssim ~  \kappa_s^{-1}\lambda\sqrt{s}
            \end{split} 
        \end{equation}
        hold with  probability at least $1-2n^{-1}$. We work on the intersection of   the event 
         \begin{equation}\label{E_lambda}
            \cE_\lambda := 
            \left\{
                {1\over n}\|\bX^\T (\bY - \bX B)\|_\i \le {1\over 2}\lambda
            \right\}
        \end{equation}
        with the event 
        \begin{equation}\label{cond_RE}
            \cE_{RE} :=  \left\{\min_{S\subseteq [p],|S| \le s}\min_{v\in \cC(S, 3)}{\|\bX v\|_2^2 \over n\|v\|_2^2} \ge  {\kappa_s\over 2}\right\}.
        \end{equation}
        Recall $\cC(S, 3) := \{u\in \RR^p\setminus \{0\}: \|u_{S^c}\|_1 \le 3\|u_S\|_1\}$. 
        The main difficulty in proving \cref{thm_B_lasso} is to establish the order of $\lambda$ as well as to show that $\cE_\lambda\cap \cE_{RE}$ holds with overwhelming probability.  
        
        \cref{lem_XE} ensures that $\cE_\lambda$ holds with probability at least $1-n^{-1}$ for any   $\lambda$ satisfying \eqref{rate_lbd}. 
        
        We proceed to verify that  
        $\PP(\cE_{RE})\ge 1-(n\vee p)^{-2}$.
        To this end, for any set $S \subseteq [p]$, we write  $E_S := \text{span}\{\be_j: j\in S\}$. Under \cref{ass_sw_weaker}, we
        define the space
        \begin{equation}\label{def_Gamma_s}
        	\Gamma_s := \bigcup_{S\subseteq [p]:\ |S| = \lfloor Cs \rfloor } E_S
        \end{equation}
        for some constant   $C=C(\kappa_s, \|\sw\|_\i)> 1$.  
        Recall that  $\wh \Sigma = n^{-1}\bX^\T \bX$.   \cref{lem_Sigma_hat_sup}  and the condition $(s\vee L)\log(n\vee p) \le c ~  n$ imply that 
        \begin{equation}\label{isometry}
        	{9\over 10}\|\Sigma^{1/2}v\|_2 \le \|\wh\Sigma^{1/2}v\|_2 \le {11\over 10}\|\Sigma^{1/2} v\|_2\qquad \forall~ v\in \Gamma_s
        \end{equation} 
         holds with probability at least $1-(n\vee p)^{-2}$.
        According to
        \citet[Theorem 3]{RudelsonZhou}, \eqref{isometry} implies that  there exists $S\subseteq[p]$ with $|S|\le s$ such that
    	\begin{equation}\label{lb_RE}
    		 u^\T \wh \Sigma u \ge {1\over 2}u^\T \Sigma u
    	\end{equation} 
    	holds for any 
     $u \in \cC(S, 3)$. \cref{lb_RE} in conjunction with \cref{ass_sw_weaker} ensures that $\cE_{RE}$ holds.

        It is now relatively straightforward 
        to finish the proof of \eqref{results_lasso}.  Standard arguments \citep{BG2011,G2021} yield that on the event $\cE_\lambda$, for any $\ell \in [L]$, 
        \begin{align*}
        	{1\over 2n}\|\bX \wh B_\ell - \bX B_\ell \|_2^2  
        	&\le {\lambda \over 2}\|\wh B_\ell - B_\ell\|_1 + \lambda \|B_\ell\|_1 - \lambda\|\wh B_\ell\|_1
        \end{align*}
        from which one can deduce that 
       \begin{equation}\label{cone_B_diff}
       	\wh B_\ell - B_\ell \in \cC(S_\ell, 3)
       \end{equation}
  		 with  $S_\ell := \supp(B_\ell)$. Invoking \eqref{lb_RE}, we find on the event $\cE_{RE}$ 
  		\begin{align*}
  				{\kappa_s\over 4} \|\wh B_\ell - B_\ell\|_2^2 \le  	{1\over 4} \|\Sigma^{1/2}(\wh B_\ell - B_\ell)\|_2^2  &\le 	{1\over 2n}\|\bX \wh B_\ell - \bX B_\ell \|_2^2\\
  				 & \le  {3\lambda \over 2}\|\wh B_\ell - B_\ell\|_1\\&\le 6\lambda \|(\wh B_\ell - B_\ell)_{S_\ell}\|_1 &&\text{by \eqref{cone_B_diff}}\\
  				& \le 6\lambda \sqrt{s} \|\wh B_\ell - B_\ell\|_2.
  		 \end{align*}
  		We have proved \eqref{results_lasso}.

        Finally, the claim follows by noting that 
        \begin{align*}
            &\|\wh \Sigma^{1/2}(\wh B-B)\|_F \le \sqrt{L}\max_{\ell \in [L]}\|\wh \Sigma^{1/2}(\wh B_\ell-B_\ell)\|_2,\\
            &\|\sw^{1/2}(\wh B-B)\|_F \le \sqrt{L}\max_{\ell \in [L]}\|\sw^{1/2}(\wh B_\ell-B_\ell)\|_2 \le \sqrt{L}\max_{\ell \in [L]}\|\Sigma^{1/2}(\wh B_\ell-B_\ell)\|_2,\\
            &\|\wh B-B\|_{1,2} = \sum_{j=1}^p \|\wh B_{j\cdot} - B_{j\cdot}\|_2 \le 
                \sum_{j=1}^p \|\wh B_{j\cdot} - B_{j\cdot}\|_1 = \sum_{\ell=1}^L \|\wh B_\ell - B_\ell\|_1.
        \end{align*}  
        The proof is complete.
    \end{proof}

    \subsubsection{Two key lemmas used in the proof of \cref{thm_B_lasso}}

    \begin{lemma}\label{lem_XE}
        Under model \eqref{model} with \cref{ass_pi}, assume $L\log(n\vee p) \le n$. With probability at least $1-n^{-1}$, one has
        \[
            \max_{j\in [p],\ \ell \in [L]} \left |{1\over n}\bX_j^\T (\bY_\ell  - \bX B_\ell)\right|     \lesssim  \left(\sqrt{\|\sw\|_\i} + \|\bM\|_\i\right)\sqrt{\log(n\vee p) \over  nL}.
        \]
    \end{lemma}
    \begin{proof}
        Note that
        \begin{align}\label{decomp_XE}
            {1\over n}\bX^\T (\bY - \bX B) &= {1\over n}\bX^\T \bY - \sxy -\wh \Sigma B + \Sigma B
        \end{align}
        We first analyze the term $(\wh \Sigma - \Sigma)B$. 
        Pick any $j\in [p]$, $\ell \in [L]$ and fix $t \ge 0$. 
        We observe  from \cref{lem_pi} and the inequality $L \le n$ that $\cE_\pi$ holds with probability at least $1-2n^{-1}$.
        Next, we invoke \cref{lem_Sigma_hat} with 
        $v_1 = \be_j$, $v_2 = B_\ell$, $t = C\log(n\vee p)$ and  we use \eqref{bd_BSigmaB}, \eqref{bd_Bmu_2} and \eqref{bd_Bmu_sup} in \cref{lem_B_fact}  to derive that the inequalities 
        \begin{align*} 
            \left| \be_j^\T(\wh\Sigma - \Sigma) B_\ell
            \right|
            &\lesssim \sqrt{[\sw]_{jj} B_\ell^\T \sw B_\ell }\sqrt{\log(n\vee p) \over n}\\\nonumber
            &\quad  + 
            \|\bM_{j\cdot}\|_\i \sqrt{B_\ell^\T \bM \bM^\T B_\ell}  \sqrt{\log(n\vee p) \over nL} +  \|\bM_{j\cdot}\|_\i \|\bM^\T B_\ell\|_\i {\log(n\vee p) \over n}\\\nonumber
            &\quad +  \sqrt{[\sw]_{jj}  B_\ell^\T \bM \bM^\T B_\ell} \sqrt{ \log(n\vee p)\over nL}   +  \|\bM_{j\cdot}\|_2\sqrt{B_\ell^\T \sw B_\ell} \sqrt{\log (n\vee p)\over nL} \\\nonumber
            &\lesssim  \sqrt{[\sw]_{jj}+ \|\bM_{j\cdot}\|_\i^2}\sqrt{\log(n\vee p) \over nL} + \|\bM_{j\cdot}\|_\i {\log(n\vee p)\over n}\\
            &\lesssim \sqrt{\|\sw\|_\i + \|\bM\|_\i^2}\sqrt{\log(n\vee p) \over nL}
        \end{align*}
        hold with probability at least $1-(n\vee p)^{-3}$.  After we take the union bound over  $j\in[p]$ and $ \ell \in [L]$,  we have that
        \begin{align}\label{bd_XXB_diff} 
           \max_{j\in [p],\ \ell\in [L]} \left| \be_j^\T(\wh\Sigma - \Sigma) B_\ell
            \right|
            &\lesssim \sqrt{\|\sw\|_\i + \|\bM\|_\i^2}\sqrt{\log(n\vee p) \over nL}
        \end{align}
        holds
        with probability $1-n^{-1}$.  

        To bound the first term in \eqref{decomp_XE}, by recalling that $\sxy = \bM D_\pi$, we find
        \begin{align*}
            {1\over n}\bX_j^\T \bY_\ell - \be_j^\T \sxy \be_\ell &= {1\over n}\sum_{k=1}^L\sum_{i=1}^L \left( X_{ij}Y_{i\ell} - \bM_{j\ell} \pi_\ell \right) \1\{ Y_i=\be_\ell\} \\
            &= {1\over n} \sum_{i=1}^L \left(X_{ij}  - \bM_{j\ell} \right)\1\{ Y_i = \ell\} + \left({n_\ell \over n} - \pi_\ell \right)\bM_{j\ell}
            \\
            & = {n_\ell \over n} (\wh \bM_{j\ell} - \bM_{j\ell}) +  (\wh \pi_\ell - \pi_\ell)\bM_{j\ell}.
        \end{align*}
        Invoking \cref{ass_pi}, \cref{lem_pi} and \cref{lem_mu} with $t = \sqrt{\log(p\vee n)}$  and taking a union bound over $j\in[p]$ and $ \ell \in [L]$, we obtain 
        \begin{align}\label{bd_XY_diff}
            \max_{j\in[p],\ \ell\in [L]}\left|
                 {1\over n}\bX_j^\T \bY_\ell - \be_j^\T \sxy \be_\ell
            \right|
            \lesssim 
            \sqrt{\|\sw\|_\i \log(p\vee n) \over  nL} + \|\bM\|_\i \sqrt{\log(n) \over nL}
        \end{align}
        with probability $1-n^{-1}$. Combining \eqref{bd_XXB_diff} and \eqref{bd_XY_diff} completes the proof. 
    \end{proof}

        \medskip

     Recall the set 
    	$$
    	\Gamma_s = \bigcup_{S\subseteq[p]: |S| = \lfloor Cs\rfloor} E_S
    	$$
    	with $E_S = \text{span}\{\be_j: j\in S\}$. 
    The following lemma gives an upper bound for $v^\T (\wh\Sigma - \Sigma) v$, uniformly over $v\in \Gamma_s$.
    
    \begin{lemma}\label{lem_Sigma_hat_sup}
    	Under model  \eqref{model} with Assumptions \ref{ass_pi} \& \ref{ass_sw_weaker}, assume $(s\vee L)\log(n\vee p) \le n$ for some integer $1\le s < p$. Then, on the event $\cE_\pi$, with probability at least $1-(n\vee p)^{-2}$, the following holds uniformly over $ v\in \Gamma_s$,  
    	\[
    	v^\T (\wh \Sigma - \Sigma) v ~ \lesssim  ~ v^\T \Sigma v \left(\sqrt{s\log(n\vee p) \over n} + \sqrt{L\log(n) \over n}\right).
    	\]
    \end{lemma}
    \begin{proof}
    	Without loss of generality, we prove the result for 
    	\[
    	v \in  \Gamma_s' := \Gamma_s \cap \cS^p
    	\]
    	via a discretization argument.
    	For any subset $S\subseteq [p]$ with $|S| = \lfloor Cs\rfloor$, let $\cN_S$ be the $(1/3)$-net of $E_S\cap \cS^p$ satisfying 
    	\[
    	\cN_S \subset E_S\cap \cS^p,\qquad |\cN_S| \le  7^{|S|}
    	\]
    	The existence of such net is ensured by \citet[Lemma 23]{RudelsonZhou}. 
    	Hence, for any fixed $v \in E_S\cap \cS^p$, there exists $u\in \cN_S$ such that $\|u-v\|_2\le 1/3$, and we find
    	\begin{align*}
    		v^\T (\wh \Sigma - \Sigma)v  &\le (v-u)^\T (\wh \Sigma - \Sigma)v + u^\T  (\wh \Sigma - \Sigma)(v-u) + u^\T  (\wh \Sigma - \Sigma) u\\
    		&\le 2\|u-v\|_2 \sup_{v\in E_S\cap \cS^p} v^\T(\wh \Sigma - \Sigma)v  + u^\T  (\wh \Sigma - \Sigma) u.
    	\end{align*}
    	The second inequality uses the fact that $(u-v)/\|u-v\|_2 \in E_S\cap \cS^p$ whenever $u\ne v$. Taking the supremum over $v\in E_S\cap \cS^p$ and the maximum over $u\in \cN_S$ gives 
    	\[
    	\sup_{v\in E_S\cap \cS^p} v^\T (\wh \Sigma - \Sigma)v \le 3 \max_{u\in \cN_S} u^\T  (\wh \Sigma - \Sigma) u
    	\]
     and 
     \begin{align}\label{todo}
         \sup_{v\in \Gamma_s'} v^\T (\wh \Sigma - \Sigma)v \le 3 \max_{u\in \cN} u^\T  (\wh \Sigma - \Sigma) u.
     \end{align}
    	Here 
    	\[
    	\cN := \bigcup_{S\subseteq[p]:\ |S| = \lfloor Cs\rfloor} \cN_S
    	\]   has cardinality
    	\[
    	   |\cN| \le \sum_{|S| = \lfloor Cs\rfloor} |\cN_S| \lesssim  {p}^{Cs}.
    	\]
    	  We will bound the right hand side \eqref{todo} via the same arguments used to prove \cref{lem_Sigma_hat}, except for the control of the term $\sum_k (\wh\pi_k - \pi_k)v^\T \mu_k\mu_k^\T v$. Specifically, repeating the proof of \cref{lem_Sigma_hat} with $v_1=v_2=v$, for any $v\in \cN$, gives that for any $t\ge 0$
    	\begin{align*}
    		v^\T(\wh\Sigma - \Sigma) v &\lesssim  v^\T \sw v \left(\sqrt{t\over n} + {t\over n}\right) + \sum_{k=1}^L (\wh \pi_k - \pi_k) v^\T \mu_k\mu_k^\T v +2 {\|\bM^\T v\|_2\over \sqrt L}\sqrt{v^\T\sw v} \sqrt{t\over n}.
    	\end{align*}
     holds, with probability at least $1-6e^{-t}$.
    	Note that 
    	\[
    	{\|\bM^\T v\|_2^2\over L} \lesssim v^\T \bM D_\pi \bM^\T v
    	\]
    	under \cref{ass_pi}. By using the decomposition of $\Sigma$ in \eqref{eq_Sigma},  choosing $t = Cs\log(n\vee p)$ for some large constant $C>0$,
     we obtain that
    	\[
    	v^\T(\wh\Sigma - \Sigma) v  \lesssim  v^\T \Sigma v  \left(\sqrt{s\log(n\vee p)\over n} + {s\log(n\vee p)\over n}\right) + \sum_{k=1}^L (\wh \pi_k - \pi_k) v^\T \mu_k\mu_k^\T v  
    	\]
    	holds uniformly over $v\in \cN$, 
    	with probability at least $1-(n\vee p)^{-C''s}$. 
    	 On the event $\cE_\pi$, invoking  \cref{ass_pi}, we conclude that
    	\[
    	\sum_{k=1}^L (\wh \pi_k - \pi_k) v^\T \mu_k\mu_k^\T v \le \max_{k\in [L]}{|\wh\pi_k - \pi_k| \over \pi_k} v^\T \bM D_\pi \bM^\T v\lesssim v^\T \bM D_\pi \bM^\T v\sqrt{L\log(n) \over n}
    	\]
    	holds uniformly over  $\cN$. The proof  is complete in view of \eqref{todo}.
    \end{proof}

    \subsection{Proof of \cref{thm_B_rr}: the reduced-rank estimator of $B$}\label{app_sec_proof_thm_B_rr}
    
    \begin{proof}
    	Let $C>0$ be some constant to be specified later. 
    	We work on the event 
    	\[
    		\cE_\lambda' = \left\{
    			{1\over n}\|\bX^\T (\bY - \bX B)\|_\op \le   \sqrt{\lambda/ C}
    		\right\}
    	\]
    	intersecting with 
    	\[
    		\cE_X  = \left\{
    				\lambda_p(\wh \Sigma) \ge {1\over 2}\lambda_p(\Sigma) \ge {c\over 2}
    		\right\}.
    	\]
        Here $c$ is defined in \cref{ass_sw}.
    	\cref{lem_Sigma_diff_op}, \cref{lem_XE_op} and \cref{ass_sw} together with \eqref{rate_lbd_rr} ensure that  $\cE_X \cap \cE_\lambda'$ hold with probability at least $1-n^{-1}$. 
    	Write $\wh r = \rank(\wh B)$ and $r = \rank(B)$.
    	By the optimality of $\wh B$, we have 
    	\begin{align*}
    		{1\over n}\|\bY - \bX \wh B\|_F^2 + \lambda \wh r &  ~ \le ~  	{1\over n}\|\bY - \bX B\|_F^2 + \lambda  r.
    	\end{align*}
    	Working out the squares gives 
    	\begin{align*}
    		{1\over n}\|\bX \wh B - \bX B\|_F^2 & \le  {2\over n}| \langle \bY - \bX B, \bX(\wh B - B)\rangle| + \lambda (r - \wh r)\\
    		&\le 2\|\wh B -B\|_*  {1\over n}\|\bX^\T (\bY - \bX B)\|_\op + \lambda (r - \wh r)\\
    		&\le 2\sqrt{r+\wh r}~ \|\wh B - B\|_F\sqrt{\lambda/C} + \lambda (r - \wh r) &&\text{by $\cE_\lambda'$}.
    	\end{align*}
Using the properties on $\cE_X$, we find
\begin{align*}
     \| \wh B-B\|_F^2 &\le {2\over c} 2\sqrt{\wh r+r} \|\wh B-B\|_F \sqrt{\lambda/C} + {2\over c} \lambda(r-\wh r)\\
    &\le {8\lambda\over c^2C} (\wh r+r) + {1\over2}  \| \wh B-B\|_F^2 +{2\over c} \lambda(r-\wh r) &&\text{by $2xy\le x^2/2+2y^2$}\\
    &= {1\over2}  \| \wh B-B\|_F^2 +  {4\over c} \lambda r &&\text{for $C=4/c$}.
\end{align*}
In the same way, on the event $\cE_X$,
\begin{align*}
   {1\over n}  \|\bX \wh B - \bX B\|_F^2   
    		&\le \sqrt{{8\over cC} \lambda  (r+\wh r) }  ~   {1\over \sqrt n} \|\bX\wh B - \bX B\|_F   + \lambda (r - \wh r) \\
      &\le \frac{1}{2n} \|\bX\wh B - \bX B\|_F^2 + 2\lambda r&&\text{for $C=4/c$}.
\end{align*}
%
This implies the rates for both $\| \bX\wh B-\bX B\|_F^2$ and $\| \wh B-B\|_F^2$.
They imply in turn the rates for $\| \sw^{1/2} (\wh B - B)\|_F$ and $\|\wh B-B\|_{2,1}$ invoking
the inequalities 
\begin{align*}
\| \sw^{1/2} (\wh B - B)\|_F^2 &\le  \| \Sigma^{1/2} (\wh B - B)\|_F^2 
    \le {2\over c} \|\wh B-B\|_F^2 &&\text{by (\ref{eq_Sigma})}
   \end{align*}
    and
    \begin{align*}
        \|\wh B-B\|_{1,2}^2 &\le {p}\|\wh B- B\|_F^2.
    \end{align*}
    This   completes the proof. 
    \end{proof}

	\bigskip

	\begin{lemma}\label{lem_XE_op}
		Under model \eqref{model} and \cref{ass_pi}, assume $(p+L)\log(n\vee p) \le n$. With probability at least $1-2n^{-2}$, we have
		\[
		  {1\over n} \| \bX^\T (\bY  - \bX B)\|_\op      \lesssim  \sqrt{\|\sw\|_\op(1+\Delta_\i)}\sqrt{(p+L)\log(n) \over  nL}.
		\]
	\end{lemma} 
	\begin{proof}
		By \eqref{decomp_XE}, it suffices to bound from above 
		\[
			\rI = \left\|(\wh \Sigma - \Sigma) B\right\|_\op,\qquad \rII =\left\| {1\over n} \bX^\T \bY - \sxy \right\|_\op = \left\| {1\over n} \bX^\T \bY - \bM D_\pi\right\|_\op.
		\]
        For $\rI$,  by invoking \cref{lem_Sigma_diff_op}, we find with probability at least $1-n^{-2}$,
		\begin{align*}
			\rI & \le  \|\Sigma^{1/2}\|_\op \|\Sigma^{-1/2}(
			\wh \Sigma -\Sigma 
			)\Sigma^{-1/2}\|_\op \|
			\Sigma^{1/2} B\|_\op \\
			&\lesssim
			\sqrt{\|\sw\|_\op (1+\Delta_\i)} \sqrt{(p+L)\log(n)\over nL} 
		\end{align*}  
		where in the last step we also use 
		\[
			\|\Sigma\|_\op  = \|\sw^{1/2} (\bI_p + \sw^{-1/2}\bM D_\pi \bM^\T \sw^{-1/2})\sw^{1/2}\|_\op \le \|\sw\|_\op 
			\left(
				1 + {\Delta_\op \over L}
			\right) \le \|\sw\|_\op (1 + \Delta_\i)
		\]
		as well as 
		\[
			\|\Sigma^{1/2} B\|_\op^2 = \|B^\T \Sigma B\|_\op  = \|D_\pi \bM^\T B\|_\op  \lesssim {\Delta_\op \over L(L+\Delta_\op)} \le  {1\over L}
		\]
		from \cref{lem_B_fact}. 
		
		Regarding $\rII$, by standard discretization argument, it suffices to bound from above 
		\[
			\max_{v\in \cN_p(1/3), u \in \cN_L(1/3)}   v^\T \left({1\over n}\bX^\T \bY - \bM D_\pi \right) u
		\]
		with $\cN_p(1/3)$ and $\cN_L(1/3)$ being the $(1/3)$-epsilon net of $\cS_p$ and $\cS_m$, respectively.  Fix any $v\in \cN_p(1/3)$ and $u\in \cN_L(1/3)$.  Observe that 
		\begin{align*}
			 v^\T \left({1\over n}\bX^\T \bY - \bM D_\pi \right) u & = v^\T (\wh \bM D_{\wh \pi} - \bM D_\pi )u \\
			& = v^\T (\wh \bM - \bM) D_{\wh \pi} u + v^\T \bM \left(D_{\wh \pi} -  D_\pi \right)u \\
			& = \sum_{k = 1}^L u_k {n_k \over n} v^\T (\wh \mu_k - \mu_k) + \sum_{k = 1}^L  (\wh \pi_k - \pi_k) u_k  v^\T \mu_k.
		\end{align*}
		Invoking \cref{lem_mu_sum} and \cref{lem_pi_sum} yields that, for any $t\ge 0$, with probability $1-4e^{-t/2}$,
		\begin{align*}
			\rII  &\lesssim \|u\|_2 \sqrt{t v^\T \sw v \over n } \sqrt{\max_k n_k \over n} + \sqrt{t \sum_{k=1}^m \pi_k u_k^2 (v^\T \mu_k)^2 \over n} + {t \|u\|_\i \|v^\T \bM\|_\i \over n}\\
			&\lesssim   \sqrt{t\|\sw\|_\op \over nL} + \sqrt{t  \|\sw\|_\op \Delta_\i \over nL} +  \sqrt{\|\sw\|_\op \Delta_\i} {t \over n}
		\end{align*}
		where in the last step we use  
		\[
					\|v^\T \bM\|_\i = \|v^\T \sw^{1/2}\|_2 \|\sw^{-1/2}\bM\|_{2,\i} \le \|\sw\|_\op^{1/2} \sqrt{\Delta_\i}.
		\] 
        Recall that $|\cN_p(1/3)| \le 7^p$ and $|\cN_L(1/3)|\le 7^L$.
		Choosing $t =C(p+L)\log(n)$, taking the union bounds over $\cN_p(1/3)$ and $\cN_L(1/3)$ and combining with the bound of $\rI$ complete the proof. 
	\end{proof}

     \subsection{Proof of the equivalence between \eqref{bayes_rule} and \eqref{Bayes_FDA_B_star}}\label{app_sec_proof_equiv}
     \begin{proof}
            Recall that $B^* = \sw^{-1}\bM$. 
            We find that 
            \begin{align*}
                  B^* (B^{*\T}\sw B^*)^+ B^{*\T} 
                &=   \sw^{-1}\bM ( \bM^\T \sw^{-1} \bM)^+\bM^\T \sw^{-1} 
            \end{align*}
            so that 
            \begin{align*}
                &\argmin_{\ell \in [L]}~ (x- \mu_\ell)^\T B^* (B^{*\T}\sw B^*)^+ B^{*\T} (x-\mu_\ell)\\
                &= \argmin_{\ell \in [L]}~  2x^\T \sw^{-1}\bM ( \bM^\T \sw^{-1} \bM)^+\bM^\T \sw^{-1}  \mu_\ell  + \mu_\ell^\T \sw^{-1}\bM ( \bM^\T \sw^{-1} \bM)^+\bM^\T \sw^{-1}\mu_\ell.
            \end{align*}
            Write the singular value decomposition of $\sw^{-1/2} \bM$ as $U\Lambda V^\T$ with $U \in \RR^{p\times r}$ containing the first $r$ singular vectors and $r = \rank(\sw^{-1/2}\bM)$. It follows that the above equals to 
            \begin{align*}
                & \argmin_{\ell \in [L]}~  2x^\T \sw^{-1/2} UU^\T \sw^{-1/2}\bM \be_\ell   + \be_\ell^\T \bM^\T   \sw^{-1/2} UU^\T \sw^{-1/2}\bM \be_\ell  \\
                &= \argmin_{\ell \in [L]}~  2x^\T \sw^{-1}\bM \be_\ell   + \be_\ell^\T \bM^\T   \sw^{-1}\bM \be_\ell  
            \end{align*}
            which, in view of  \eqref{bayes_rule} and \eqref{def_G}, completes the proof.
        \end{proof}

\section{Technical lemmas}\label{app_sec_tech_lem}
    
    The following lemma provides a few useful facts on quantities related with $B$ in \eqref{def_B}.
    Write 
    \begin{equation}\label{def_Omega}
        \Omega := D_\pi^{-1} + \dist.
    \end{equation}
    \begin{lemma}\label{lem_B_fact}
        With $B$ defined in \eqref{def_B}, we have
        \begin{align}\label{mu_B}
            & \bM^\T B =  \dist ~ \Omega^{-1},\\\label{eq_Omega} 
            &H = D_\pi -  D_\pi\bM^\T B = \Omega^{-1},\\\label{quad_B}
            & B^\T \sw B  = \Omega^{-1} \dist ~ \Omega^{-1}.
        \end{align}
    	Consequently, for any $k,\ell \in [L]$, we have 
    	\begin{align}\label{bd_BSigmaB}
    		 & B_\ell^\T \sw B_\ell  \le B_\ell^\T \Sigma B_\ell \le \pi_\ell \left(1 \wedge \pi_\ell \Delta_{\ell \ell}\right),\\\label{bd_Bmu_2}
    		& B_\ell^\T \bM \bM^\T B_\ell \le \left(
    		 1 \wedge   {\pi_\ell^2 \over \pi_{\min}} \Delta_{\ell \ell}
    		\right),\\\label{bd_Bmu_sup}
    		&  \left|\mu_k^\T B_\ell\right| \le \left(
    		1 \wedge  \pi_\ell \sqrt{\Delta_{kk}\Delta_{\ell \ell}}
    		\right).
    	\end{align}
    \end{lemma}
    \begin{proof}
        From \cref{lem_key} and its proof, we know that 
        $$
            B  =  B^*\left(D_\pi -  D_\pi\bM^\T B\right) = \sw^{-1}\bM \left(D_\pi -  D_\pi\bM^\T B\right) 
        $$
        so that 
        \[
            \bM^\T B  =  \dist \left(D_\pi -  D_\pi\bM^\T B\right) = \dist  D_\pi -  \dist   D_\pi\bM^\T B.
        \]
        Rearranging terms gives
        \begin{align}\nonumber
            \bM^\T B  &= \left(\bI_L + \dist D_\pi\right)^{-1} \dist D_\pi\\\nonumber
            & = \bI_L - \left(\bI_L + \dist D_\pi\right)^{-1}\\\label{mu_B_prime}
            & = \bI_L - D_\pi^{-1}\left(D_\pi^{-1} + \dist \right)^{-1}\\\nonumber
            &= \dist \left(D_\pi^{-1} + \dist \right)^{-1}
        \end{align} 
        proving \eqref{mu_B}.
        By using the above identity \eqref{mu_B_prime}, we have
        \[
            D_\pi -  D_\pi\bM^\T B =  \left(D_\pi^{-1} + \dist \right)^{-1},
        \]
        proving
        \eqref{eq_Omega}.  
        Since \cref{lem_key} and \eqref{eq_Omega} imply 
        \[  
            B = B^*   \Omega^{-1},
        \]
        \eqref{quad_B} follows from
        \[
            B^\T \sw B = \Omega^{-1} B^{*\T}\sw B^* \Omega^{-1} = \Omega^{-1}\dist \Omega^{-1}.
        \] 
        
        To prove \eqref{bd_BSigmaB} -- \eqref{bd_Bmu_sup}, pick any $k,\ell \in [L]$. 
        By using \eqref{eq_Sigma} twice and  \eqref{eq_B}, we find that 
        \begin{equation}\label{bd_quad_mu}
        	B_\ell^\T \sw B_\ell  \le B_\ell^\T \Sigma B_\ell  = \pi_\ell^2 \mu_\ell^\T \Sigma^{-1}\mu_\ell \le \pi_\ell^2  \mu_\ell^\T \sw^{-1}\mu_\ell
		\end{equation}
        yielding the second bound in   \eqref{bd_BSigmaB}. The other bound in \eqref{bd_BSigmaB}  follows by observing that 
        \begin{equation}\label{bd_quad_mu_prime}
        	\pi_\ell^2 \mu_\ell^\T \Sigma^{-1}\mu_\ell = \pi_\ell \be_\ell^\T D_\pi^{1/2} \bM^\T \Sigma^{-1}\bM D_\pi^{1/2}\be_\ell  \le \pi_\ell \|\Sigma^{-1/2}\bM D_\pi \bM^\T \Sigma^{-1/2}\|_\op \overset{\eqref{eq_Sigma}}{\le} \pi_\ell.
      		\end{equation}
        For \eqref{bd_Bmu_2}, on the one hand, similar arguments yield 
        \[
        	B_\ell^\T \bM \bM^\T B_\ell \le {\pi_\ell^2\over \pi_{\min}} \mu_\ell^\T \Sigma^{-1}  \bM D_\pi \bM^\T \Sigma^{-1} \mu_\ell \le {\pi_\ell^2\over \pi_{\min}}\mu_\ell^\T \Sigma^{-1}\mu_\ell.
        \]
        On the other hand, using \eqref{mu_B} and \eqref{eq_Omega} proves 
        \[
        	B_\ell^\T \bM \bM^\T B_\ell \le \be_\ell^\T \Omega^{-1} (\dist)^2 \Omega^{-1}\be_\ell \le  1.
        \]
        Finally, the last statement follows by noting that 
        \[
        	|\mu_k^\T B_\ell|^2 \le B_\ell^\T \bM \bM^\T B_\ell \le 1
        \]
        and 
        \[
        	|\mu_k^\T B_\ell| = \pi_\ell |\mu_k^\T \Sigma^{-1} \mu_\ell| \le \pi_\ell\sqrt{\mu_k^\T \Sigma^{-1}\mu_k} \sqrt{\mu_\ell^\T \Sigma^{-1}\mu_\ell}.
        \]
    \end{proof}

    The following lemmas contain deviation inequalities related with $\wh\bM-\bM$ and $\wh \pi - \pi$.

    \begin{lemma}\label{lem_mu_sum}
        Let $v\in\RR^p$ and $u\in \RR^L$ be any fixed vectors. Under \eqref{model}, for any $t\ge 0$, with probability at least $1-2e^{-t^2/2}$, we have
        \[
            \left|
                    \sum_{k=1}^L {n_k\over n} v^\T(\wh\mu_k - \mu_k) u_k
            \right| \le t \|u\|_2 \sqrt{ v^\T \sw v \over n } \sqrt{\max_k n_k \over n}.
        \]  
    \end{lemma}
    \begin{proof}
        Recall that, for any $k\in[L]$,
        \[
            \wh\mu_k = {1\over n_k}\sum_{i=1}^nX_i\1\{ Y_i= \be_k\}.
        \]
        We obtain
        \begin{align*}
             \sum_{k=1}^L {n_k\over n} v^\T(\wh\mu_k - \mu_k) u_k 
            & = v^\T \sum_{k=1}^L  {1\over n}\sum_{i=1}^n\left(X_i - \mu_k\right)u_k \1\{ Y_i= \be_k\}\\
            &=  v^\T \sum_{k=1}^L  {1\over n}\sum_{i=1}^n\left(X_i - \bM Y_i\right)u^\T Y_i \1\{ Y_i= \be_k\}\\
            &= {1\over n}\sum_{i=1}^n v^\T \left(X_i - \bM Y_i\right)u^\T Y_i .
        \end{align*}
        For $i\in [n]$, by conditioning on $\bY$, we know that $W_i := X_i - \bM Y_i$ are i.i.d. from $\cN_p(0, \sw)$, hence 
        \[
            {1\over n}\sum_{i=1}^n v^\T W_i u^\T Y_i ~ \Big | ~ \bY \sim \cN\left(
                0, {v^\T\sw v\over n} {u^\T \bY^\T \bY u \over n}
            \right).
        \]
        The claim follows by invoking the standard Gaussian tail probability bounds  and after unconditioning. 
    \end{proof}

    \begin{lemma}\label{lem_pi_sum}
        Let $u,v\in \RR^L$ be any fixed vectors. Under \eqref{model}, for any $t\ge 0$, with probability at least $1-2e^{-t/2}$,
        \[
            \left|
                    \sum_{k=1}^L  (\wh\pi_k - \pi_k) u_k v_k
            \right| \lesssim \sqrt{t \sum_{k=1}^L \pi_k u_k^2v_k^2 \over n} + {t \|u\|_\i \|v\|_\i \over n}.
        \]  
    \end{lemma}
    \begin{proof}
Observe that we can write
\begin{align*}
             \sum_{k=1}^L  (\wh\pi_k - \pi_k) u_k v_k &=               \sum_{k=1}^L {1\over n}\sum_{i=1}^n 
               Y_{ik}  u_k v_k - \sum_{k=1}^L \pi_k  u_k v_k\\
            &=  {1\over n}\sum_{i=1}^n Z_i - \pi_k  u_k v_k
        \end{align*}
Here
        \[ Z_i = \sum_{k=1}^L  Y_{ik}  u_k v_k\]
         are independent, bounded random variables
(as          $|Z_i| \le \|u\|_\i \|v\|_\i$)
with mean
        \[
            \EE[Z_i] =  \sum_{k=1}^L \pi_k  u_k v_k
        \]
        and variance 
        \[
            \Var(Z_i) = \Var(y_i^\T \diag(u) v) \le \sum_{k=1}^L \pi_k v_k^2 u_k^2
        \]
        for $i\in[n]$.
     The proof follows after a straightforward application of  Bernstein's inequality for bounded random variables. 
          \end{proof}

    \medskip
    
    The following lemma controls the difference  
    \[
        \wh \Sigma_w - \sw := {1\over n}\sum_{k=1}^L \sum_{i=1}^n\1\{ Y_i=\be_k\} (X_i - \mu_k)(X_i - \mu_k)^\T - \sw.
    \]
    \begin{lemma}\label{lem_sw}
        Let  $u,v\in \RR^p$ be  any fixed vectors.  
        Under model \eqref{model}, for any $t\ge 0$, with probability at least $1-2e^{-t}$, one has 
        \[
            \left|u^\T \left(\wh \Sigma_w - \sw \right) v\right| \lesssim \sqrt{u^\T \sw u}\sqrt{v^\T \sw v}\left(\sqrt{t\over n} + {t\over n}\right)
        \]
    \end{lemma}
    \begin{proof}
        Start with
        \begin{align*}
            \wh \Sigma_w - \sw  &= {1\over n}\sum_{k=1}^L \sum_{i=1}^n \1\{ Y_i =\be_k\}  (X_i - \mu_k)(X_i - \mu_k)^\T - \sw\\
            &= {1\over n}\sum_{i=1}^n\left[(X_i - \bM Y_i)(X_i - \bM Y_i)^\T - \sw\right].
        \end{align*}
        Recall that,  conditioning $\bY$,  $W_i = X_i - \bM Y_i$ for $i\in [n]$ are independent, $\cN_p(0, \sw)$. The result follows from standard concentration inequalities for the quadratic term of Gaussian random vectors. 
    \end{proof}

    \subsection{Concentration inequalities related with  $\wh \Sigma - \Sigma$}
    \label{sec_proof_sigma}
    
    Recall that $\wh \Sigma = n^{-1}\bX^\T \bX$ and  the definition of the event $\cE_\pi$ is given in \eqref{def_event_pi}. 
    
    \begin{lemma}\label{lem_Sigma_hat}
    	Under model \eqref{model} and \cref{ass_pi},
         assume $L\log (n) \le n$. On the event $\cE_\pi$, for any fixed $v_1, v_2\in \RR^p$ and any $t\ge 0$, with probability at least $1-8e^{-t}$, one has
    	\begin{align*}
    		\left| v_1^\T (\wh \Sigma - \Sigma) v_2\right| &\lesssim ~ 
    		\sqrt{v_1^\T \sw v_1}\sqrt{v_2^\T \sw v_2} \left(\sqrt{t\over n} + {t\over n}\right)\\
    		&\quad + \sqrt{\sum_{k=1}^L (v_1^\T \mu_k)^2 (v_2^\T\mu_k)^2
    			\over L} \sqrt{t  \over n} + \|\bM^\T v_1\|_\i \|\bM^\T v_2\|_\i{t  \over n} \\
    		&\quad + \left({\|\bM^\T v_1\|_2\over \sqrt L}\sqrt{v_2^\T\sw v_2}+{\|\bM^\T v_2\|_2\over \sqrt L}\sqrt{v_1^\T\sw v_1}\right)\sqrt{t\over n}.
    	\end{align*}
    \end{lemma}
    \begin{proof}
    	By definition, we have  
    	\begin{align*}
    		\wh\Sigma &= {1\over n}\bX^\T \bX  = {1\over n}\sum_{k=1}^L \sum_{i=1}^n \1\{ Y_i=\be_k\}  X_i X_i^\T \\
    		&=  {1\over n}\sum_{k=1}^L\left[
    		\sum_{i=1}^n \1\{ Y_i=\be_k\} (X_i - \mu_k)(X_i-\mu_k)^\T + n_k(\wh \mu_k  \mu_k^\T + \mu_k \wh \mu_k^\T) - n_k\mu_k\mu_k^\T
    		\right]\\
    		&= \wh \Sigma_w +  \sum_{k=1}^L{n_k\over n}\mu_k\mu_k^\T + \sum_{k=1}^L {n_k\over n}\left[
    		(\wh\mu_k - \mu_k) \mu_k^\T +  \mu_k (\wh\mu_k - \mu_k)^\T
    		\right]
    	\end{align*}
    	where we write 
    	\[
    	\wh \Sigma_w := {1\over n}\sum_{k=1}^L 
    	\sum_{i=1}^n \1\{ Y_i=\be_k\} (X_i - \mu_k)(X_i-\mu_k)^\T.
    	\]
    	By further using the decomposition of $\Sigma$ in \eqref{eq_Sigma} and $\wh\pi_k = n_k/n$ for $k\in [L]$, we find that 
    	\begin{align*}
    		\wh\Sigma  - \Sigma &= \wh\Sigma_w - \sw + \sum_{k=1}^L (\wh \pi_k - \pi_k) \mu_k\mu_k^\T +  \sum_{k=1}^L \wh\pi_k\left[
    		(\wh\mu_k - \mu_k) \mu_k^\T +  \mu_k (\wh\mu_k - \mu_k)^\T
    		\right].
    	\end{align*}
    	The first term $\wh\Sigma_w - \sw
     $ is bounded in  \cref{lem_sw}. We bound the second term by invoking \cref{ass_pi} and \cref{lem_pi_sum} with $u = \bM^\T v_1$ and $v = \bM^\T v_2$. This yields 
    	\begin{align}\label{bd_pi_diff}\nonumber
    		\left|\sum_{k=1}^L (\wh \pi_k - \pi_k) v_1^\T \mu_k\mu_k^\T v_2\right| &\lesssim \sqrt{t \sum_{k=1}^L  \pi_k (v_1^\T \mu_k)^2 (v_2^\T\mu_k)^2 \over n} + {t \|\bM^\T v_1\|_\i \|\bM^\T v_2\|_\i \over n}\\
    		&\lesssim \sqrt{\sum_{k=1}^L (v_1^\T \mu_k)^2 (v_2^\T\mu_k)^2
    			\over L} \sqrt{t  \over n} + \|\bM^\T v_1\|_\i \|\bM^\T v_2\|_\i{t  \over n} 
    	\end{align}
    	with probability at least $1-2e^{-t}$. 
    	Furthermore, on the event $\cE_\pi$, \cref{ass_pi} and $L\log(n) \le n$ imply 
    	\[
    	\wh\pi_k \lesssim \pi_k + \sqrt{\log(n) \over nL} \lesssim {1\over L},\quad \forall \ k\in [L],
    	\]
    	and after we invoke \cref{lem_mu_sum} twice with $v = v_1, u = \bM^\T v_2$  and $v = v_2, u = \bM^\T v_1$, respectively, we get that
    	\begin{align}\label{bd_mu_diff_1}
    		& \sum_{k=1}^L \wh\pi_k 
    		v_1^\T (\wh\mu_k - \mu_k) \mu_k^\T v_2 \lesssim  \|\bM^\T v_1\|_2 \sqrt{t v_2^\T\sw v_2 \over n } \sqrt{\max_k \wh \pi_k} \lesssim {\|\bM^\T v_1\|_2\over \sqrt L}\sqrt{t v_2^\T\sw v_2 \over n },\\ \label{bd_mu_diff_2}
    		&  \sum_{k=1}^L \wh\pi_k v_1^\T \mu_k (\wh\mu_k - \mu_k)^\T v_2 \lesssim          \|\bM^\T v_2\|_2 \sqrt{t v_1^\T \sw v_1 \over n } \sqrt{\max_k \wh \pi_k}\lesssim {\|\bM^\T v_2\|_2\over \sqrt L}\sqrt{t v_1^\T \sw v_1 \over n }
    	\end{align} hold
    	with probability at least $1-4e^{-t}$. Collecting all three bounds completes the proof.  
    \end{proof}
    
    \medskip 
    
    As an immediate application of \cref{lem_Sigma_hat}, we have the following bounds on the sup-norm and operator norm of  
    the matrix $B^\T (\wh \Sigma - \Sigma) B$.
    
    \begin{lemma}\label{lem_B_Sigma_diff_B}
    	Under conditions of \cref{lem_Sigma_hat}, assume $\Delta_\i \gtrsim 1$. On the event $\cE_\pi$, the following holds  with probability at least $1-n^{-2}$,
    	\begin{align*}
    		& \|B^\T (\wh \Sigma - \Sigma) B\|_\i \lesssim  \left(1 \wedge {\Delta_\i  \over L}\right)^{3/2}  \sqrt{\log (n)\over nL},\\
    		& \|B^\T (\wh \Sigma - \Sigma) B\|_\op\lesssim    {\Delta_\op   \over L+\Delta_\op}\sqrt{\Delta_\i\log(n) \over nL} +  {\Delta_\op \over 
    			L +\Delta_\op}{\Delta_\i\log(n) \over n} .
    	\end{align*}
    \end{lemma}
    \begin{proof}
    	For the sup-norm bound, it suffices to bound from above 
    	$|B_k^\T (\wh \Sigma - \Sigma) B_k|$  for any $k\in [L]$. Invoking \cref{lem_Sigma_hat} with $v_1 =v_2 =  B_k$  and $t = \log(n)$ together with \eqref{bd_BSigmaB}, \eqref{bd_Bmu_2} and \eqref{bd_Bmu_sup} of \cref{lem_B_fact} gives that, with probability at least $1-n^{-2}$,
    	\begin{align*}
    		|B_k^\T (\wh \Sigma - \Sigma) B_k| 
    		&\lesssim  B_k^\T \sw B_k  \sqrt{\log(n) \over n}  + 
    		\|\bM^\T B_k\|_\i \sqrt{B_k^\T \bM \bM^\T B_k}  \sqrt{\log(n) \over nL} \\\nonumber
    		&\quad +  \|\bM^\T B_k \|_\i^2 {\log(n) \over n}+ 2\sqrt{B_k^\T \sw B_k  B_k^\T \bM \bM^\T B_k}   \sqrt{\log (n)\over nL} \\\nonumber
    		&\lesssim  \left(1 \wedge {\Delta_\i  \over L}\right)\sqrt{\log(n)\over nL^2}+ 
    		\left(1 \wedge {\Delta_\i  \over L}\right)^{3/2}\sqrt{\log (n)\over nL}+  \left(1 \wedge {\Delta_\i  \over L}\right)^{2}{\log(n)\over n}\\
    		&\lesssim  
    		\left(1 \wedge {\Delta_\i  \over L}\right)^{3/2}\sqrt{\log (n)\over nL}. 
    	\end{align*}
    	The last step uses $\Delta_\i \gtrsim 1$ to collect terms. This proves the first claim.
    	
    	Regarding the operator norm bound, a standard discretization argument gives 
    	\[
    	\|B^\T (\wh \Sigma - \Sigma) B\|_\op \le 3 \max_{u \in \cN_L(1/2)} u^\T B^\T (\wh \Sigma - \Sigma) B u
    	\]
    	where $\cN_L(1/3)$ denotes the $(1/3)$-epsilon net of $\{u \in \RR^L:  \|u\|_2 = 1\}$. Note that $|\cN_L(1/3)| \le 7^L$. Then invoking \cref{lem_Sigma_hat} with $t = C L\log(n)$ and using
    	\begin{align*}
    		&\|B^\T \sw B\|_\op \le {\Delta_\op   \over L(L+\Delta_\op)},\\ 
    		&\|B^\T \bM \|_\op \le {\Delta_\op   \over L+\Delta_\op},\\ 
    		&\|\bM^\T B u\|_\i \le \sqrt{\Delta_\i\Delta_\op \over 
    			L(L +\Delta_\op)}
    	\end{align*} 
    	deduced from 
    	\cref{lem_B_fact} give that, with probability at least $1- 8 n^{-CL + L\log(7)} \ge 1-n^{-L}$, 
    	\begin{align*}
    		u^\T B^\T (\wh \Sigma - \Sigma) Bu  
    		&\lesssim  u^\T B ^\T \sw B u  \sqrt{L\log(n) \over n}  + 
    		\|\bM^\T B u\|_\i \sqrt{u^\T B^\T \bM \bM^\T B u}  \sqrt{\log(n) \over n} \\\nonumber
    		&\quad +  \|\bM^\T B u \|_\i^2 {L\log(n) \over n}+ 2\sqrt{u^\T B^\T \sw B u u^\T B^\T \bM \bM^\T B u}   \sqrt{\log (n)\over n} \\\nonumber
    		&\lesssim     {\Delta_\op   \over L+\Delta_\op}\sqrt{\Delta_\i\log(n) \over nL} +  {\Delta_\op \over 
    			L +\Delta_\op}{\Delta_\i\log(n) \over n} 
    	\end{align*}
    	holds uniformly over $u \in \cN_L(1/3)$. We also use $\Delta_\i \gtrsim 1$ to simplify expressions in the last step above. This completes the proof. 
    \end{proof}

    The following lemma provides upper bounds of the operator norm of $(\wh \Sigma - \Sigma)$. 
    
    \begin{lemma}\label{lem_Sigma_diff_op}
    	Under model \eqref{model}  and \cref{ass_sw}, assume $(p+L)\log(n) \le n$. Then with probability at least $1-n^{-2}$, the following holds uniformly over $v \in \cS^p$
    	\[
    	v^\T (\wh \Sigma - \Sigma) v ~ \lesssim  ~ v^\T \Sigma v \left(\sqrt{p\log(n) \over n} + \sqrt{L\log(n) \over n}\right).
    	\]
    	As a result, with the same probability, we have 
    	\[
    	\lambda_p(\wh \Sigma) \ge {1\over 2} \lambda_p(\Sigma).
    	\]
    \end{lemma}
    \begin{proof}
    	The proof of the first statement follows from the same arguments of proving \cref{lem_Sigma_hat_sup} and is thus omitted. 
    	
    	For the second statement, for any $v\in \cS^p$, we have 
    	\[
    	v^\T \wh \Sigma v  = v^\T \Sigma^{1/2} \left(\Sigma^{-1/2}\wh \Sigma \Sigma^{-1/2}\right)  \Sigma^{1/2} v \ge v^\T\Sigma v  ~ \lambda_p \left(
    	\Sigma^{-1/2}\wh \Sigma \Sigma^{-1/2}
    	\right).
    	\]
    	Since, by Weyl's inequality and on the event that the first result holds,
    	\[
    	\lambda_p \left(
    	\Sigma^{-1/2}\wh \Sigma \Sigma^{-1/2}
    	\right) \ge 1 - \|\Sigma^{-1/2}(\wh \Sigma - \Sigma) \Sigma^{-1/2}\|_\op \ge 1 - \left(\sqrt{p\log(n) \over n} + \sqrt{L\log(n) \over n}\right)
    	\]
    	we conclude 
    	\[
    	v^\T \wh \Sigma v \ge {1\over 2}v^\T \Sigma v
    	\]
    	uniformly over $\cS^p$, completing the proof. 
    \end{proof}

    \subsection{Auxiliary lemmas}\label{app_sec_aux_lem}

    The following lemmas are proved in \cite{BW2023}. 
    \begin{lemma}\label{lem_pi}
        Assume $\pi_{\min} \ge 2\log n/n$ for some sufficiently large constant $C$. Then, for any $k\in [L]$
        	\[
        	\PP\left\{
        	|\wh\pi_k - \pi_k| < \sqrt{16\pi_k (1-\pi_k)\log n\over n}
        	\right\}\ge 1-n^{-2}.
        	\]
        	Furthermore, if $\pi_{\min} \ge C\log n/ n$ for some sufficiently large constant $C$, then 
        	\[
        	\PP\left\{
        	c\pi_k \le \wh \pi_k \le c'\pi_k
        	\right\} \ge 1-n^{-2}.
        \]
    \end{lemma}

    \begin{lemma}\label{lem_mu}
    	Suppose that model \eqref{model} holds.  For any deterministic vector $v\in \RR^p$ and $k\in [L]$,  for all $t>0$,
    		\[
    		\PP\left\{
    		\left|v^\T (\wh \mu_k-\mu_k)\right| \ge t\sqrt{v^\T  \sw v \over n_k} 
    		\right\} \le 2e^{-{t^2/ 2}}.
    		\] 
    \end{lemma}

\end{document}